\documentclass{article}
\usepackage{macros}
\usepackage{benign_macros}

\title{Improved Algorithms for Low Rank Approximation from Sparsity}
\author{
    David P. Woodruff \\
    Carnegie Mellon University \\
    \texttt{dwoodruf@cs.cmu.edu}
    \and
    Taisuke Yasuda \\
    Carnegie Mellon University \\
    \texttt{taisukey@cs.cmu.edu}
}
\date{}
\begin{document}

\begin{titlepage}
\maketitle
\thispagestyle{empty}

\begin{abstract}
    We overcome two major bottlenecks in the study of low rank approximation by assuming the low rank factors themselves are sparse. Specifically,
    \begin{enumerate}[label={(\arabic*)}]
    \item for low rank approximation with spectral norm error, we show how to improve the best known $\nnz(\bfA) k / \sqrt{\eps}$ running time to $\nnz(\bfA)/\sqrt{\eps}$ running time plus low order terms depending on the sparsity of the low rank factors, and
    \item for streaming algorithms for Frobenius norm error, we show how to bypass the known $\Omega(nk/\eps)$ memory lower bound and obtain an $s k (\log n)/ \poly(\eps)$ memory bound, where $s$ is the number of non-zeros of each low rank factor. Although this algorithm runs in exponential time, as it must under standard complexity-theoretic assumptions, we also present polynomial time algorithms using $\poly(s,k,\log n,\eps^{-1})$ memory that output rank $k$ approximations supported on an $O(sk/\eps)\times O(sk/\eps)$ submatrix.
    \end{enumerate}
    Both the prior $\nnz(\bfA) k / \sqrt{\eps}$ running time and the $nk/\eps$ memory for these problems were long-standing barriers; our results give a natural way of overcoming them assuming sparsity of the low rank factors.
\end{abstract}

\end{titlepage}


\newpage

\section{Introduction}

In modern computational problems, one often encounters large scale and high dimensional datasets that are too large and complex to be understood as is. A common approach to dealing with this complexity is to find \emph{low rank approximations} to these datasets. That is, given a matrix $\bfA\in\mathbb R^{n\times d}$ and a target rank $k$, one seeks a rank $k$ matrix $\bfB\in\mathbb R^{n\times d}$ such that $\bfA$ is well-approximated by $\bfB$. This is advantageous, as a rank $k$ matrix can be written as $\bfB = \bfU\cdot \bfV^\top$ for an $n\times k$ matrix $\bfU$ and a $d\times k$ matrix $\bfV$, which only takes $(n+d)k$ parameters to describe rather than $nd$. Oftentimes, the task of finding this $\bfB$ is accomplished by computing a rank $k$ matrix $\bfB$ that approximately minimizes the reconstruction error $\norm*{\bfA - \bfB}$, where $\norm*{\cdot}$ is some norm, e.g., the Frobenius norm or the spectral norm. 

In practice, the coordinates often have physical interpretations, and in such cases, it is desirable to be able to interpret the low rank approximation found. However, $\bfB$ as obtained above, now described with $(n+d)k$ parameters, is often still too large to interpret easily. Thus, one could further ask for the low rank approximation factor $\bfB$ to be sparse as well. One way to formalize this is to require that $\bfB$ can be written as the sum of $k$ rank $1$ matrices, each of which is supported on an $s\times s$ submatrix for some sparsity parameter $s$. We refer to such a $\bfB$ as being an \emph{$s\times s$-sparse rank $k$ matrix}. 

\begin{Definition}[$s\times s$-sparse rank $k$ matrix]\label{def:sxs-sparse-rank-k-matrix}
    For a sparsity parameter $s\in\mathbb N$ and rank parameter $k \leq n$, an $n\times d$ matrix $\bfB$ is $s\times s$ sparse rank $k$ if it can be written as
    \[
        \bfB = \sum_{i=1}^k \tau_i \bfx_i\bfy_i^\top
    \]
    where $\tau_i \in\mathbb R$ and $\bfx_i\in\mathbb R^n$ and $\bfy_i\in\mathbb R^d$ are $s$-sparse unit vectors for each $i\in[k]$. When the ambient dimensions $n$ and $d$ are clear from context, we write $\mathcal S_{s,k}\subseteq\mathbb R^{n\times d}$ for the set of $s\times s$-sparse rank $k$ matrices. 
\end{Definition}

The above notion of sparse low rank matrices can be either used as an \emph{assumption} or a \emph{constraint}. We study both variants of the low rank approximation problem, and in both cases, we obtain improved algorithms over prior low rank approximation algorithms in certain settings. 

In the first setting, one \emph{assumes} that the optimal rank $k$ approximation takes the form of a sparse low rank matrix. This scenario is a phenomenon known as \emph{localization} of eigenvectors, and occurs frequently in many applications \cite{DBLP:journals/nla/HernandezBCY21, DBLP:journals/corr/abs-2009-04414}, for example in quantum many-body problems \cite{lagendijk2009fifty, nandkishore2015many} and network analysis \cite{pastor2018eigenvector}. In this case, the optimal solution to the low rank approximation problem can be computed by an SVD, which has efficient approximation algorithms \cite{DBLP:conf/stoc/ClarksonW13,DBLP:conf/nips/MuscoM15,DBLP:conf/nips/ZhuL16}. However, one may further ask the question of whether we can do better than the best known algorithms for approximate SVD under this additional sparsity assumption on the singular vectors. This question was studied in the work of \cite{DBLP:journals/nla/HernandezBCY21} and a followup work of \cite{DBLP:journals/corr/abs-2009-04414}, which studied algorithms for computing eigenvectors in symmetric matrices with localized eigenvectors. In \cite{DBLP:journals/nla/HernandezBCY21}, the authors study an algorithm for finding a small submatrix containing the supports of the leading eigenvectors by greedily adding rows and columns without formal guarantees, and \cite{DBLP:journals/corr/abs-2009-04414} seek to improve this approach using reinforcement learning techniques. 

In the second setting, we instead deal with an arbitrary input matrix $\bfA$, but still seek a sparse low rank matrix to approximate $\bfA$. Thus, we \emph{constrain} our low rank approximation to be sparse, and optimize over this smaller space instead. It turns out that this constraint makes the problem significantly harder in time complexity (as we show later), which is a phenomenon that is well-documented in the context of the related sparse PCA problem (see Section \ref{sec:related-work} for a discussion). Nonetheless, we show that this formulation of the problem is in fact \emph{easier} than the unconstrained low rank approximation problem in the context of streaming algorithms. 

\subsection{Randomized Block Krylov Methods for Sparse Singular Vectors}

Our first contribution gives new algorithms for the spectral low rank approximation problem, when we are promised that the top $k$ singular components are $s$-sparse. Let $\bfA\in\mathbb R^{n\times d}$ be a rank $r$ matrix whose top $k$ left and right singular vectors are $s$-sparse. The best previous known general case algorithm of \cite{DBLP:conf/nips/MuscoM15} (see Theorem \ref{thm:randomized-block-krylov}), based on block Krylov methods, runs in time dominated by roughly
\[
  \frac{\nnz(\bfA)k}{\sqrt\eps}. 
\]
We improve this bound to roughly
\[
  \frac{\nnz(\bfA)}{\sqrt\eps}. 
\]

\begin{theorem}[Spectral low rank approximation for sparse singular vectors]\label{thm:slra-sparse-block-krylov}
    Let $\bfA\in\mathbb R^{n\times d}$ be a rank $r$ matrix whose top $k$ left and right singular vectors are supported on $s$ coordinates. Then, there is an algorithm which outputs $\hat\bfA$ such that
    \[
        \norm*{\bfA - \hat\bfA}_2^2 \leq (1+\eps)\norm*{\bfA - \bfA_k}_2^2
    \]
    in time
    \[
        O\parens*{\frac{\nnz(\bfA)}{\sqrt\eps} + \frac{n}{\eps}}\log \frac{srk\log n}{\eps} + \poly\parens*{s,k,\frac1\eps,\log n}.
    \]
\end{theorem}

\begin{Remark}
    By a simple modification of our analyses, the condition for $s$-sparse singular vectors can be replaced by an approximate sparsity condition, such that for the $j$th singular component, there is a set of $s$ coordinates which captures a
    \[
        1 - \frac{\eps}{k\sqrt r}\frac{\sigma_{k+1}}{\sigma_j}
    \]
    fraction of the $\ell_2$ mass of the left singular vector, and similarly for the right singular vector. 

    We note that this robustness condition is a strong assumption, and we believe that an important direction for future work is to find algorithms that handle more relaxed assumptions on the approximate sparsity of the singular vectors. 
\end{Remark}

Our algorithm proceeds in two steps: (1) we first identify small subsets of coordinates $S\subset [n]$ and $T\subset[d]$ that contain all large coordinates of the sparse left and right singular vectors (2) we proceed with standard low rank approximation algorithms on the subset of coordinates. The idea is that finding top singular vectors in the smaller submatrix is sufficient to find the top singular vectors in the original matrix, since all of our original singular vectors roughly retain their singular values in the submatrix, while the other singular components can only decrease. The observation for this general approach has already been made in \cite{DBLP:journals/nla/HernandezBCY21, DBLP:journals/corr/abs-2009-04414}. However, to implement this idea efficiently, several additional new ideas are needed, and our algorithm carefully combines power method, Chebyshev polynomials, and binary search to obtain the result of Theorem \ref{thm:slra-sparse-block-krylov}. 

\subsubsection{Overview of Techniques}\label{sec:sparse-krylov-overview}

Our first idea is that with a budget of $\nnz(\bfA)/\sqrt\eps$ running time, we can run na\"ive power method for $1/\sqrt\eps$ iterations initialized with a single random Gaussian vector $\bfg\sim\mathcal N(0,\bfI_d)$ to compute $(\bfA\bfA^\top)^{1/\sqrt\eps}\bfA\bfg$. Using the SVD $\bfA = \bfU\bfSigma\bfV^\top$ of $\bfA$, we may write this as $\bfU\bfSigma^{O(1/\sqrt\eps)}\bfV^\top\bfg$. Then by the rotational invariance of the Gaussian, this random vector is distributed as a random linear combination of the left singular vectors of $\bfA$, where the $i$th left singular vector is scaled by roughly $\sigma_i^{1/\sqrt\eps}$. Then if $i\in[k]$ is such that $\sigma_i \geq (1+\sqrt\eps)\sigma_{k+1}$, then $\sigma_i^{1/\sqrt\eps}$ is a constant factor larger than $\sigma_{k+1}^{1/\sqrt\eps}$, so the $s$ entries corresponding to the $i$th left singular vector stand out in $(\bfA\bfA^\top)^{1/\sqrt\eps}\bfA\bfg$. Thus, selecting the top $sk$ entries with largest absolute value in $(\bfA\bfA^\top)^{1/\sqrt\eps}\bfA\bfg$ retrieves a superset of the support of the left singular vectors with singular value $\sigma_i$ for which $\sigma_i \geq (1+\sqrt\eps)\sigma_{k+1}$. We can repeat on the right side as well to obtain the support of the large right singular vectors. 

The above approach is enough to find the large singular components with singular value at least $(1+\sqrt\eps)\sigma_{k+1}$, but we must find singular values all the way down to $(1+\eps)\sigma_{k+1}$ for a $(1+\eps)$ relative error approximation. To do this, we use the approach of \cite{DBLP:conf/nips/MuscoM15} of using Chebyshev polynomials, which, given a location parameter $\alpha$, gives us a degree $1/\sqrt\eps$ polynomial $p$ for which $p(x)$, for all $x\geq (1+\eps)\alpha$, is a constant times greater than any $p(x)$ for $x\leq \alpha$ (see Lemma \ref{lem:chebyshev-poly} for the mathematical statement). If we knew the location $\alpha = \sigma_{k+1}$, then we could compute $p(\bfA)\bfg$ in $\nnz(\bfA)/\sqrt\eps$ time and use the same approach as before to find the support of all singular components $i$ for which $\sigma_i \geq (1+\eps)\sigma_{k+1}$. The challenge, of course, is that we do not know $\sigma_{k+1}$. 

We first show how to find the value of $\sigma_{k+1}$ up to a $(1+\sqrt\eps)$ factor. To this end, we first show that if $\sigma_i$ for $i\in[k]$ is large, i.e.\ $\sigma_i \geq (1+\sqrt\eps)\sigma_{k+1}$, then we can find $\sigma_i$ up to a $(1+\eps)$ factor using the set of $sk$ large coordinates on the left and right located before, using the power method. However, note that we do not know for which $i$ this is true. That is, if we let $\hat\bfA$ be the $sk\times sk$ submatrix supported on the large coordinates identified using the power method, we expect the large singular values of $\hat\bfA$ to be good estimates of the large singular values of $\bfA$, but we do not know which of them are large enough to actually be good estimates. 

To address this, let $i\in[k]$, and first note that $\sigma_i(\hat\bfA)$ is always a lower bound on $\sigma_i(\bfA)$ by the Cauchy interlacing theorem. Furthermore, suppose that $\hat\bfB$ is a rank $i-1$ approximation to $\hat\bfA$. Then, $\norm*{\bfA - \hat\bfB}_2$ serves as an upper bound for $\sigma_i(\bfA)$, as
\[
    \sigma_i(\bfA) = \min_{\text{rank $i-1$ $\bfC$}} \norm*{\bfA - \bfC}_2 \leq \norm*{\bfA - \hat\bfB}_2.
\]
We show that for $i\in[k]$ such that $\sigma_i \geq (1+\sqrt\eps)\sigma_{k+1}$, these are good lower and upper bounds on the singular value $\sigma_i(\bfA)$, i.e., they are within $(1+\eps)$ factors of each other. Furthermore, they can both be computed in time roughly
\[
    \frac{\nnz(\bfA)}{\sqrt\eps} + \poly(s,k,\eps^{-1}).
\]
Thus, we have an extremely efficient way to certify our estimates to the singular values $\sigma_i(\bfA)$, if they are large enough. We then consider the following binary search strategy over the singular values: if the upper and lower bounds are within $(1+\eps)$ factors of each other, then we keep searching lower, and otherwise, we search higher. If the $\sigma_{i_*}(\bfA)$ found is such that $\sigma_{i_*}(\bfA) \geq (1+\sqrt\eps)\sigma_{k+1}(\bfA)$, then the top $i_*$ singular components are found in the initial power method step accurately enough so that $\norm*{\bfA-\hat\bfB}_2$ is close to $\sigma_{i_*+1}(\bfA) \leq (1+\sqrt\eps)\sigma_{k+1}(\bfA)$, where $\hat\bfB$ is a rank $i_*$ approximation $\hat\bfB$ of $\hat\bfA$. Otherwise, $\sigma_{i_*}(\bfA)$ itself is within a $(1+\sqrt\eps)$ factor of $\sigma_{k+1}(\bfA)$. 

Now that we are within a $(1+\sqrt\eps)$ factor of $\sigma_{k+1}(\bfA)$, we just need $1/\sqrt\eps$ guesses in powers of $(1+\eps)$ in order to guess $\sigma_{k+1}(\bfA)$ up to a factor of $(1+\eps)$. We can in fact afford to guess all of these locations $\alpha$, compute the corresponding Chebyshev polynomial $p$, compute $p(\bfA)\bfg$ from precomputed Krylov iterates, select the top $sk$ entries, and then add the entries to the support that we consider. 

With the support superset in hand, we finish the algorithm by performing an approximate SVD on this submatrix. Our full discussion can be found in Section \ref{sec:sparse-krylov}. 

\subsection{Streaming Frobenius Low Rank Approximation with Sparse Factors}

We now discuss our results on the Frobenius norm sparse low rank approximation problem. For the Frobenius norm, a variant of the sparse low rank approximation problem has previously been studied in \cite{DBLP:journals/siammax/ZhangZS02,DBLP:journals/siammax/ZhangZS04}, who slightly relax the sparsity requirement compared to our formulation.

We consider the \emph{streaming setting}, in which our input matrix $\bfA$ arrives as one pass over entrywise additive updates $\bfA_{i,j} \leftarrow \bfA_{i,j} + \Delta$ for $\Delta\in\mathbb R$. The goal is to design algorithms which use low space. A standard approach is to consider \emph{sketching algorithms}, in which we choose a random $m\times (nd)$ matrix $\bfS$ and maintain $\bfS(\vc(\bfA))$, where $\vc(\bfA)$ is the flattening of the $n\times d$ matrix $\bfA$ into an $nd$-dimensional vector, and then postprocess $\bfS(\vc(\bfA))$ for the final result. This yields an algorithm of space roughly $\tilde O(m)$, up to bit complexity factors. We refer to each row of $\bfS$ as a \emph{measurement}. In fact, it is known that, with some restrictions, any turnstile streaming algorithm can be implemented as a linear sketch with roughly the same amount of memory \cite{DBLP:conf/stoc/LiNW14, DBLP:conf/coco/AiHLW16, DBLP:conf/stoc/KallaugherP20}, so in some sense, the number of measurements required for a sketching algorithm in a turnstile stream captures the space complexity of the streaming problem. 

For the standard Frobenius norm low rank approximation problem, it is known since the work of \cite{DBLP:conf/stoc/ClarksonW09, DBLP:conf/stoc/BoutsidisWZ16} that $\tilde \Theta(nk/\eps)$ measurements are necessary and sufficient to output a rank $k$ approximation of an $n\times d$ matrix $\bfA$ with a relative error Frobenius norm guarantee.

\begin{theorem}[\cite{DBLP:conf/stoc/ClarksonW09,DBLP:conf/stoc/BoutsidisWZ16}]\label{lem:frob-lra-sketching-complexity}
    Let $\bfA\in\mathbb R^{n\times d}$. Then, any randomized sketching algorithm which outputs a rank $k$ matrix $\bfB$ such that
    \[
        \norm*{\bfA - \bfB}_F^2 \leq (1+\eps)\min_{\text{rank $k$ $\bfC$}}\norm*{\bfA - \bfC}_F^2
    \]
    with probability at least $2/3$ requires $\Theta((n+d)k/\eps)$ measurements, and this lower bound is tight up to constant factors. 
\end{theorem}

\subsubsection{Upper Bounds}

We overcome the above negative result by showing that we can find an $s\times s$-sparse rank $k$ matrix $\bfB$ such that 
\[
    \norm*{\bfA - \bfB}_F^2 \leq (1+\eps)\min_{\bfC\in\mathcal S_{s,k}}\norm*{\bfA - \bfC}_F^2
\]
in only $O(sk(\log n)/\eps^2)$ measurements, which for constant $\eps$ replaces the dependence on the ambient dimension $n$ by a dependence of $s$ on the desired sparsity. Our algorithm achieving this space bound iterates over a net and runs in exponential time, which is necessary as we show in Section \ref{sec:computational-complexity}. We additionally present polynomial time algorithms with some relaxations by allowing the output to be $s'\times s'$-sparse rank $k$, for $s' = O(sk/\eps)$, one with a \emph{relative error guarantee} introduced above, as well as an improved algorithm with the weaker \emph{additive error guarantee}, i.e.,
\[
    \norm*{\bfA - \bfB}_F^2 \leq \min_{\bfC\in\mathcal S_{s,k}}\norm*{\bfA - \bfC}_F^2 + \eps\norm*{\bfA}_F^2.
\] 

\begin{table}[h]
    \centering
    \begin{tabular}{ccccc}
        Theorem & Sparsity & Measurements & Running Time & Error \\
        \hline
        Theorem \ref{thm:net-iterate-fslra} & $s\times s$ & $sk\cdot\poly(\log n, \eps^{-1})$ & Exponential Time & Relative Error \\
        Theorem \ref{thm:frob-slra-polytime-rel-err} & $O(sk/\eps)\times O(sk/\eps)$ & $s^2k^2\cdot\poly(\log n, \eps^{-1})$ & Polynomial Time & Relative Error \\
        Theorem \ref{thm:frob-slra-polytime-add-err} & $O(sk/\eps)\times O(sk/\eps)$ & $sk^2\cdot\poly(\log n, \eps^{-1})$ & Polynomial Time & Additive Error \\
    \end{tabular}
    \caption{Our results for Frobenius norm sparse low rank approximation.}
    \label{tab:frob-norm-slra}
\end{table}

Both of our polynomial time bicriteria algorithms proceed by identifying $O(sk/\eps)$ heavy rows and columns that must contain an approximately optimal $s\times s$-sparse rank $k$ matrix, and then computing an approximate SVD on them. Although this algorithm is simple in the offline model, the challenge here is to implement this algorithm in one pass over a stream with low memory.

By employing variations on classical techniques for finding heavy hitters in a stream, we show that the identification of the indices of $O(sk/\eps)$ heavy rows and columns can be done in roughly $O(sk/\eps)$ measurements. For our first relative error algorithm, we also maintain a separate heavy hitters data structure which recovers every entry of $\bfA$ up to an additive error of
\[
    \frac{\eps^4}{s^2 k^2}\norm*{\bfA_{\overline{[s^2k^2/\eps^4]}}}_F^2,
\]
where $\bfA_{\overline{[s^2k^2/\eps^4]}}$ denotes the matrix obtained by deleting the $s^2k^2/\eps^4$ largest entries of $\bfA$ in absolute value. This takes an additional $s^2k^2/\eps^4$ measurements. Then, we recover the $O(sk/\eps)\times O(sk/\eps)$ submatrix of the heavy rows and columns up to an overall additive Frobenius norm error of
\[
    O\parens*{\frac{sk}{\eps}} \cdot O\parens*{\frac{sk}{\eps}} \cdot \frac{\eps^4}{s^2 k^2}\norm*{\bfA_{\overline{[s^2k^2/\eps^4]}}}_F^2 \leq O(\eps^2) \norm*{\bfA - \bfB}_F^2
\]
where $\bfB$ is any matrix supported on at most $s^2 k^2/\eps^4$ entries. This is small enough so that we can simply perform an offline approximate SVD on the recovered entries. 

In order to improve upon the measurement bound of this simple algorithm, we change our approximate SVD routine from a black box offline SVD algorithm to the $\ell_2$ sampling based approach of \cite{DBLP:journals/jacm/FriezeKV04, DBLP:journals/siamcomp/DrineasKM06a}, which first performs an adaptive sampling routine in order to reduce the dimension before performing an SVD. Although this has a weaker additive error guarantee, we show that by carefully combining this approach with heavy hitter  algorithms, we can remove a factor of $s$ in the measurement bound. 

\subsubsection{Lower Bounds}

Our upper bounds are complemented by relatively straightforward reductions from known lower bounds for related problems. Perhaps the most obvious is the lower bound for the usual Frobenius norm low rank approximation problem (Theorem \ref{lem:frob-lra-sketching-complexity}), which already gives an $\Omega(sk/\eps)$ lower bound. 

\begin{corollary}\label{cor:sketching-lb-flra-lra}
    Let $\bfA\in\mathbb R^{n\times d}$ be such that there exists a unique $\bfB\in\mathcal S_{s,k}$ that minimizes $\norm*{\bfA - \bfB}_F^2$ over $\bfB\in\mathcal S_{s,k}$ that is supported on an $s\times s$ submatrix and has orthogonal factors. Then, any randomized sketching algorithm which outputs a matrix $\bfB\in\mathcal S_{s,k}$ such that
    \[
        \norm*{\bfA - \bfB}_F^2 \leq (1+\eps)\min_{\text{rank $k$ $\bfC$}}\norm*{\bfA - \bfC}_F^2
    \]
    with probability at least $2/3$ requires $\Omega(sk/\eps)$ measurements. 
\end{corollary}

Importantly, this lower bound holds even when there is a unique optimal solution that is promised to lie on an $s\times s$ submatrix, and the sparse low rank factors are orthogonal. Note that none of these properties is guaranteed for the general problem. In fact, given the promise that an optimal solution exists on an $s\times s$ submatrix, our additive error polynomial time algorithm improves to a measurement bound of $sk\cdot \poly(\log n,\eps^{-1})$, which has a tight dependence on $s$ and $k$. 

Despite this natural lower bound, tighter lower bounds can be obtained from reductions from work in the sparse recovery literature when the optimal solution need not be supported on an $s\times s$ submatrix. In the sparse recovery problem, one is given a hidden vector $\bfx\in\mathbb R^{n}$, and the goal is to output an $s$-sparse vector $\hat\bfx$ such that
\[
    \norm*{\bfx - \hat\bfx}_2^2 \leq (1+\eps)\min_{\text{$s$-sparse $\bfx'$}}\norm*{\bfx - \bfx'}_2^2,
\]
given linear measurements of the vector. The following lower bounds are known for this problem:

\begin{theorem}[\cite{DBLP:conf/focs/PriceW11}]\label{thm:sparse-recovery-sketching-complexity}
    Let $\bfx\in\mathbb R^n$. Then, any randomized sketching algorithm which outputs an $\hat\bfx$ such that
    \[
        \norm*{\bfx - \hat\bfx}_2^2 \leq (1+\eps)\min_{\text{$s$-sparse $\bfx'$}}\norm*{\bfx - \bfx'}_2^2
    \]
    with probability at least $2/3$ requires $\Omega(s(\log(n/s))/\eps)$ measurements. Furthermore, if the algorithm guarantees that $\hat\bfx$ is $s$-sparse, then $\Omega(s/\eps^2)$ measurements are required. 
\end{theorem}

Note that an $sk$-sparse vector can be written as the sum of $k$ many $s$-sparse vectors. This simple observation immediately gives the following:

\begin{corollary}\label{cor:sketching-lb-fslra}
    Let $\bfA\in\mathbb R^{n\times d}$. Then, any randomized sketching algorithm which outputs a $\hat\bfB$ such that
    \[
        \norm*{\bfA - \hat\bfB}_F^2 \leq (1+\eps)\min_{\bfC\in\mathcal S_{s,k}}\norm*{\bfA - \bfC}_2^2
    \]
    with probability at least $2/3$ requires $\Omega(sk(\log(n/sk))/\eps)$ measurements. Furthermore, if the algorithm guarantees that $\hat\bfB$ is in $\mathcal S_{s,k}$, then $\Omega(sk/\eps^2)$ measurements are required.
\end{corollary}

Thus, in the general case, both the $\log n$ factor as well as the quadratic dependence on $\eps$ are required. Table \ref{tab:frob-norm-slra-lb} summarizes the above discussion:

\begin{table}[h]
    \centering
    \begin{tabular}{ccc}
        Theorem & Lower Bound & Notes \\
        \hline
        Corollary \ref{cor:sketching-lb-flra-lra} & $sk/\eps$ & Optimal solution supported on $s\times s$ submatrix \\
        Corollary \ref{cor:sketching-lb-fslra} & $sk(\log(n/sk))/\eps$ & Non-sparse output allowed  \\
        Corollary \ref{cor:sketching-lb-fslra} & $sk/\eps^2$ & Output required to be in $\mathcal S_{s,k}$ \\
    \end{tabular}
    \caption{Measurement lower bounds for Frobenius norm sparse low rank approximation.}
    \label{tab:frob-norm-slra-lb}
\end{table}

Our full discussion can be found in Section \ref{sec:slra}. 

\subsection{Gaussian Noise Spectral Sparse Low Rank Approximation}

As a final contribution, we study the sketching complexity of sparse low rank approximation in the common setting where the input matrix $\bfA\in\mathbb R^{n\times n}$ is the sum of a low rank signal and additive Gaussian noise. That is,
\[
    \bfA = \lambda \bfX + \bfG
\]
where $\bfG_{i,j}\sim\mathcal N(0,1)$, $\norm*{\bfX}_2 = 1$, and $\lambda$ is roughly $O(\sqrt n)$. This ``low rank signal plus Gaussian noise'' setting is ubiquitous in both theory and in practice, in areas such as matrix denoising (see, e.g., \cite{DBLP:journals/ma/ShabalinN13, DBLP:journals/tsp/CandesST13} and references therein) which has applications to image denoising. Here, the normalization is such that two parts $\bfG$ and $\lambda \bfX$ have roughly the same operator norm. Two common computational problems in this setting are 
\begin{enumerate}[label={(\arabic*)}]
    \item detection, the problem of distinguishing between the above distribution and $\bfA \sim \mathcal N(0,1)^{n\times n}$, and
    \item estimation, the problem of computing an estimate $\hat\bfX$ such that $\norm*{\hat\bfX - \bfX}_2$ is small. 
\end{enumerate}
We study these two problems in the streaming setting, where the main resource of study is the number of measurements required for the computational task. 

It is known that without the sparsity requirement for $\bfX$, $\Omega(n^2)$ measurements are required to even detect the presence of a rank $1$ $\bfX$ \cite[Corollary 3.3]{DBLP:journals/siamcomp/LiNW19}, that is, one must essentially read every entry of the matrix. 

\begin{theorem}[\cite{DBLP:journals/siamcomp/LiNW19}]\label{thm:operator-norm-lb}
    Let $\bfA\in\mathbb R^{n\times n}$ be drawn as either $\bfG\sim\mathcal N(0,1)^{n\times n}$ or $\bfG + \lambda\bfu\bfv^\top$ for $\bfu\sim\mathcal N(0,\bfI_n)$ and $\bfv\sim\mathcal N(0,\bfI_d)$ drawn independently and $\lambda = \Theta(1/\sqrt n)$. Then, any randomized sketching algorithm that distinguishes between these two cases must make at least $\Omega(n^2)$ measurements. 
\end{theorem}

As in the previous setting of Frobenius norm low rank approximation, we ask whether this bound can be improved when we assume that $\bfX$ is $s\times s$-sparse rank $k$. In this setting, we further restrict our attention to $\bfX$ with disjoint support, in order to obtain space-efficient algorithms.

\begin{Definition}[Disjoint $s\times s$-sparse rank $k$ matrix]
    An $s\times s$-sparse rank $k$ matrix $\bfB\in\mathcal S_{s,k}$ for
    \[
        \bfB = \sum_{i=1}^k \tau_i \bfx_i\bfy_i^\top
    \]
    is disjoint if the support of the component matrices $\bfx_i\bfy_i^\top$ are pairwise disjoint. When the ambient dimensions $n$ and $d$ are clear from the context, we write $\mathcal O_{s,k}\subseteq\mathbb R^{n\times d}$ for the set of disjoint $s\times s$-sparse rank $k$ matrices. 
\end{Definition}

It is often desirable for the components of a sparse low rank approximation to be orthogonal \cite{DBLP:conf/nips/Mackey08,DBLP:journals/tsp/BenidisSBP16}, and one could view this assumption as a strengthening of this notion.

For the problem of detecting the ``signal'' $\bfX$, we obtain a lower bound of
\[
    \tilde\Omega(n + s^2 k)
\]
measurements, where the first term comes from a $4$-norm estimation lower bound of \cite{DBLP:conf/focs/AndoniKO11}, while the second term is achieved by a generalization of the construction of Theorem \ref{thm:operator-norm-lb}. We complement these lower bounds with algorithms achieving a measurement bound of
\[
    \tilde O(n + s^2 k^{4/3})
\]
which is off by a factor of $k^{1/3}$ compared to the lower bound. The $\tilde O(n)$ term follows also from a $4$-norm estimation algorithm, while the $\tilde O(s^2 k^{4/3})$ term comes from an analysis of an algorithm which samples a submatrix and iterates over a net to detect the signal. 

For the problem of estimation of $\bfX$, we show that, $\Theta(nsk\log(n/s))$ measurements is tight for obtaining an $\bfX'$ such that $\norm*{\bfX-\bfX'}_2 \leq \frac1{10}\sqrt n$. The upper bound again is a net iteration argument, while the lower bound is a modification of sparse recovery lower bounds of \cite{DBLP:conf/focs/PriceW11}. 

\begin{table}[h]
    \centering
    \begin{tabular}{ccc}
        Problem & Upper Bound & Lower Bound \\
        \hline
        Detection ($s\leq \sqrt{n/k\log n}$) & $\tilde O(n)$ & $\tilde\Omega(n)$ \cite{DBLP:conf/icalp/AndoniNPW13} \\
        Detection ($s\geq \sqrt{n/k\log n}$) & $\tilde O(s^2k^{4/3})$ & $\Omega(s^2k)$  \\
        Estimation & $\tilde O(nsk)$ & $\tilde \Omega(nsk)$ \cite{DBLP:conf/focs/PriceW11}
    \end{tabular}
    \caption{Sketching complexity of Gaussian noise sparse low rank approximation.}
    \label{tab:sketching-complexity-sparse-pca}
\end{table}

Our full discussion can be found in Section \ref{sec:gaussian-noise}. 

\subsection{Other Related Work}\label{sec:related-work}

\paragraph{Sparse PCA.}

The closely related problem of sparse PCA also seeks sparse low rank approximations and has been studied intensely in a variety of settings. In this problem, one seeks only a sparse right factor, that is, an $s$-sparse unit vector $\bfx$ that maximizes the quantity $\bfx^\top\bfA^\top\bfA\bfx$. We refer to \cite{ning2015literature, DBLP:journals/pieee/ZouX18} for surveys on various formulations, algorithms, and hardness results on sparse PCA. Note the distinction between this problem and our problem, which seeks to minimize the reconstruction error with both sparse left and right factors. From the perspective of computational complexity, it is known that sparse PCA is NP-hard even to approximate \cite{DBLP:conf/icml/MoghaddamWA06, DBLP:journals/ipl/Magdon-Ismail17}, with improved inapproximability results under other complexity theoretic assumptions \cite{DBLP:conf/colt/ChanPR16}. Sparse PCA has also recently been studied from an average-case complexity perspective under the natural input distribution of the form ``signal plus noise'', see e.g., \cite{DBLP:conf/colt/BrennanB19}. 

\subsection{Notation}

We denote the $i$th standard basis vector $(0, 0, \dots, 0, 1, 0, \dots, 0)$ by $\bfe_i$. For a matrix $\bfA\in\mathbb R^{n\times d}$, we denote the number of nonzero entries in $\bfA$ by $\nnz(\bfA)$, the spectral norm by
\[
    \norm*{\bfA}_2 = \max_{\bfx\neq 0}\frac{\norm*{\bfA\bfx}_2}{\norm*{\bfx}_2},
\]
and the Frobenius norm by
\[
    \norm*{\bfA}_F = \sqrt{\sum_{i=1}^n\sum_{j=1}^d (\bfA_{i,j})^2}.
\] 

We will often use coordinate projection and sampling matrices:

\begin{Definition}\label{def:coordinate-projection-and-sampling-matrices}
For a set $S\subseteq[n]$, we write $\bfP_S$ for the $n\times n$ \emph{coordinate projection matrix} that zeros out all rows except for those indexed by $S$, i.e., 
\[
    (\bfP_S)_{i,j} = \begin{cases}
        1 & \text{if $i = j \in S$} \\
        0 & \text{otherwise}
    \end{cases}
\]
and we write $\bfS_S$ for the $\abs{S}\times n$ \emph{sampling matrix} for $S$, i.e.,
\[
    (\bfS_S)_{i,j} = \begin{cases}
        1 & \text{if $j \in S$ is the $i$th element of $S$} \\
        0 & \text{otherwise}
    \end{cases}
\]
\end{Definition}

For a matrix $\bfA\in\mathbb R^{n\times d}$, we write 
\[
    \bfA = \bfU\bfSigma\bfV^\top = \sum_{i=1}^{\rank(\bfA)} \sigma_i(\bfA)\bfu_i\bfv_i^\top
\]
for its singular value decomposition (SVD), where $\bfU$ and $\bfV$ have orthonormal columns, $\bfSigma$ is a diagonal matrix with $\sigma_i(\bfA)$ along its diagonal, and $\bfu_i$ and $\bfv_i$ are the $i$th columns of $\bfU$ and $\bfV$, respectively. For $k\in[\rank(\bfA)]$, we write $\bfA_k = \bfU_k\bfSigma_k\bfV_k^\top$ for the truncated SVD of $\bfA$. Here, $\bfU_k$ is an $n\times k$ matrix with the first $k$ $\bfu_i$ as its columns, $\bfV_k$ is a $d\times k$ matrix with the first $k$ $\bfv_i$ as its columns, and $\bfSigma_k$ is the $k\times k$ diagonal matrix with the first $k$ $\sigma_i(\bfA)$ along the diagonal. 

By now, extremely efficient approximation algorithms for the SVD are known \cite{DBLP:conf/stoc/ClarksonW13,DBLP:conf/nips/MuscoM15,DBLP:conf/nips/ZhuL16}. In particular, \cite{DBLP:conf/nips/MuscoM15} show the following for relative error spectral norm SVD using the \emph{randomized block Krylov iteration} algorithm:

\begin{theorem}[Randomized block Krylov iteration \cite{DBLP:conf/nips/MuscoM15}]\label{thm:randomized-block-krylov}
    Let $\bfA\in\mathbb R^{n\times d}$. Then, there is an algorithm \cite[Algorithm 2]{DBLP:conf/nips/MuscoM15} running in time
    \[
        O\parens*{\nnz(\bfA)\frac{k\log d}{\sqrt\eps}} + n\poly(k,\log d,\eps^{-1})
    \]
    \cite[Theorem 7]{DBLP:conf/nips/MuscoM15} which outputs a matrix $\bfZ\in\mathbb R^{n\times k}$ with orthogonal columns with the following guarantee \cite[Theorem 10]{DBLP:conf/nips/MuscoM15}: 
    \[
        \norm*{\bfA - \bfZ\bfZ^\top\bfA}_2 \leq (1+\eps)\norm*{\bfA - \bfA_k}_2
    \]
\end{theorem}
\section{Preliminaries}\label{sec:prelim}

\subsection{\texorpdfstring{$\eps$}{epsilon}-Nets}

We will need the notion of $\eps$-nets for our analyses. See, e.g., Section 4.2 of \cite{vershynin2018high} for a standard reference. 

\begin{Definition}[$\eps$-net]
    For a set $K$, a subset $\mathcal N\subseteq K$ is an \emph{$\eps$-net} of $K$ if for every $\bfx\in K$ there exists an $\hat\bfx$ such that
    \[
        \norm*{\bfx - \hat\bfx}_2 \leq \eps.
    \]
\end{Definition}

For $K = \mathbb S^{d-1}$ the unit sphere in $d$ dimensions, there exist small $\eps$-nets. 

\begin{Lemma}[Corollary 4.2.13, \cite{vershynin2015estimation}]\label{lem:eps-net}
    Let $\eps\in(0,1)$. There exists an $\eps$-net $\mathcal N$ of $\mathbb S^{d-1}$ of size at most
    \[
        \abs*{\mathcal N} \leq \parens*{\frac3\eps}^d.
    \]
\end{Lemma}

Using this, we can construct an $\eps$-net for $\mathcal S_{s,k} \cap \mathbb S^{n\times d - 1}$, where $\mathbb S^{n\times d - 1}$ denotes the set of all $n\times d$ matrices with Frobenius norm $1$, provided that the coefficients $\tau_i$ in Definition \ref{def:sxs-sparse-rank-k-matrix} are bounded by $\poly(n)$.

\begin{corollary}[$\eps$-net for $\mathcal S_{s,k}$]\label{cor:sxs-sparse-rank-k-eps-net}
    Let $\eps\in(0,1)$. There exists an $\eps$-net $\mathcal N$ of $\mathcal S_{s,k}^{\mathrm{bounded}} \cap \mathbb S^{n\times d - 1}$ of size at most
    \[
        \abs*{\mathcal N} \leq \poly(nk/\eps)^{sk},
    \]
    where $\mathcal S_{s,k}^{\mathrm{bounded}}$ is the set of $s\times s$-sparse rank $k$ matrices with $\tau_i$ bounded by $\poly(n)$.
\end{corollary}
\begin{proof}
    We first construct an $\eps$-net $\mathcal N$ over rank $1$ $s\times s$ sparse matrices in the Frobenius norm and bound its size. Consider an $(\eps/2)$-net for $\mathbb S^{s-1}$, of size at most $(6/\eps)^s$ (see Lemma \ref{lem:eps-net}). Then for any rank $1$ $s\times s$ sparse matrix $\sigma\bfu\bfv^\top$, we may find vectors $\bfu'$ and $\bfv'$ in the net such that
    \begin{align*}
        \norm*{\bfu\bfv^\top - \bfu'(\bfv')^\top}_F &\leq \norm*{\bfu\bfv^\top - \bfu'\bfv^\top}_F + \norm*{\bfu'\bfv^\top - \bfu'(\bfv')^\top}_F \\
        &\leq \norm*{\bfu - \bfu'}_2\norm*{\bfv}_2 + \norm*{\bfu'}_2\norm*{\bfv-\bfv'}_2 \\
        &\leq \eps.
    \end{align*}
    Then we choose a sparsity pattern for this submatrix, which requires a net of total size at most
    \[
        \binom{n}{s}^{2}\parens*{\frac{6}{\eps}}^{2s} \leq \parens*{\frac{6en}{\eps s}}^{2s}
    \]

    Now suppose $\bfX\in\mathcal S_{s,k}$ with
    \[
        \bfX = \sum_{i=1}^k \tau_i \bfx_i\bfy_i^\top.
    \]
    We now approximate each component up to $\eps/k$ Frobenius norm error. Because $\tau_i \leq \poly(n)$, for each $i\in[k]$, we can find a $\tau_i'$ among at most $\poly(n)k/\eps$ candidates such that
    \[
        \abs*{\tau_i - \tau_i'} \leq \frac{\eps}{2k}.
    \]
    Furthermore, we find $\bfx_i'$ and $\bfy_i'$ in an $\eps/\poly(n)k$-net for $s\times s$-sparse rank $1$ matrices such that
    \[
        \norm*{\bfx_i'\bfy_i'^\top - \bfx_i\bfy_i^\top}_F \leq \frac{\eps}{\poly(n)k}.
    \]
    Then,
    \begin{align*}
        \norm*{\tau_i' \bfx_i'\bfy_i'^\top - \tau_i \bfx_i\bfy_i^\top}_F &\leq \norm*{\tau_i' \bfx_i'\bfy_i'^\top - \tau_i' \bfx_i\bfy_i^\top}_F + \norm*{\tau_i' \bfx_i\bfy_i^\top - \tau_i \bfx_i\bfy_i^\top}_F \\
        &\leq \abs*{\tau_i}\norm*{\bfx_i'\bfy_i'^\top - \bfx_i\bfy_i^\top}_F + \abs*{\tau_i - \tau_i'}\norm*{\bfx_i\bfy_i^\top}_F \\
        &\leq \frac{\eps}{2k} + \frac{\eps}{2k} \\
        &\leq \frac{\eps}{k}.
    \end{align*}
    Then summing over the $k$ components, we find an $\bfX'\in\mathcal S_{s,k}$ for
    \[
        \bfX' = \sum_{i=1}^k \tau_i' \bfx_i' \bfy_i'^\top
    \]
    that is at most a distance of $\eps$ to $\bfX$, by the triangle inequality. The total net size is
    \[
        \bracks*{\frac{\poly(n)k}{\eps}\parens*{\frac{\poly(n)k}{\eps s}}^{2s}}^k = \poly(nk/\eps)^{sk}. 
    \]
\end{proof}

\subsection{\textsf{CountSketch} \texorpdfstring{\cite{DBLP:journals/tcs/CharikarCF04}}{[CCF04]}} \label{sec:countsketch}

We collect some results on \textsf{CountSketch} \cite{DBLP:journals/tcs/CharikarCF04}, a versatile primitive for the design of sketching algorithms. 

\begin{Definition}[\textsf{CountSketch} \cite{DBLP:journals/tcs/CharikarCF04}]
    A \textsf{CountSketch} matrix is a distribution over $r\times n$ matrices $\bfS$ where samples are drawn as follows. Let $h:[n]\to[r]$ and $\sigma:[n]\to\{\pm1\}$ be random hash functions. Then,
    \[
        (\bfS)_{i,j} = \begin{cases}
            \sigma_j & \text{if $i = h(j)$} \\
            0 & \text{otherwise}
        \end{cases}.
    \]
    Equivalently,
    \[
        \bfS = \bfH\bfD
    \]
    where $\bfD$ is an $n\times n$ diagonal matrix with $\sigma(j)$ in the $j$th diagonal position for $j\in[n]$, and $\bfH$ is an $r\times n$ matrix where the $j$th column is the random standard basis vector $\bfe_{h(j)}$ for $j\in[n]$. 

    That is, $\bfS$ acts on a vector by multiplying each coordinate by a random sign and hashing it to one of $r$ buckets. 
\end{Definition}

\subsubsection{Tail Error Guarantee}

The first guarantee achieved by \textsf{CountSketch}, introduced by \cite{DBLP:journals/tcs/CharikarCF04}, is the \emph{tail error guarantee}.

\begin{Lemma}[Tail error guarantee \cite{DBLP:journals/tcs/CharikarCF04}]\label{lem:cs-tail-error}
    Let $\bfx\in\mathbb R^n$. Let $B\in\mathbb N$ and $\delta\in(0,1)$. Consider $r = O\parens*{\log\frac1\delta}$ independent copies of $O(B)\times n$ \textsf{CountSketch} matrices $\bfS^{(1)}, \bfS^{(2)}, \dots, \bfS^{(r)}$. Then, given
    \[
        \bfS^{(1)}\bfx, \bfS^{(2)}\bfx, \dots, \bfS^{(r)}\bfx,
    \]
    for each $i\in[n]$, one can recover an estimate $\hat\bfx_i$ such that
    \[
        \Pr\braces*{\abs*{\hat\bfx_i - \bfx_i}^2 \leq \frac1B \norm*{\bfx_{\overline{[B]}}}_2^2} \geq 1 - \delta,
    \]
    where $\bfx_{\overline{[B]}}$ denotes the vector obtained by zeroing out the $B$ coordinates of $\bfx$ with largest absolute value. 
\end{Lemma}

\subsubsection{Row-wise Approximation}

The tail error guarantee generalizes to the following guarantee for estimating the rows of a matrix \cite[Lemma 2.2]{DBLP:conf/stoc/MahabadiRWZ20}:

\begin{Lemma}[Row-wise approximation \cite{DBLP:conf/stoc/MahabadiRWZ20}]\label{lem:cs-row-wise-approximation}
    Let $\bfA\in\mathbb R^{n\times d}$, $\eps>0$, and $\delta\in(0,1)$. Consider $r = O(\log(\frac{1}{\delta}))$ independent copies of $O(\frac1\eps)\times n$ \textsf{CountSketch} matrices $\bfS^{(1)}, \bfS^{(2)}, \dots, \bfS^{(r)}$. Then, given
    \[
        \bfS^{(1)}\bfA, \bfS^{(2)}\bfA, \dots, \bfS^{(r)}\bfA,
    \]
    one can obtain an estimate $\hat\bfA\in\mathbb R^{n\times d}$ such that for each $i\in[n]$,
    \[
        \Pr\braces*{\norm*{\bfe_i^\top\hat\bfA - \bfe_i^\top\bfA}_2^2 \leq \eps\norm*{\bfA_{\overline{[1/\eps]},*}}_F^2} \geq 1 - \delta,
    \]
    where $\bfA_{\overline{[1/\eps]},*}$ denotes the matrix with the top $1/\eps$ rows with largest $\ell_2$ norm zeroed out.
\end{Lemma}

As a corollary, we obtain the following lemma for approximating row norms with an additive-multiplicative error guarantee.

\begin{corollary}[Row norm approximation]\label{lem:cs-row-norm-approximation}
    Let $\bfA\in\mathbb R^{n\times d}$, $\eps>0$, $\alpha>0$, and $\delta\in(0,1)$. Consider $r = O(\log(\frac{1}{\delta}))$ independent copies of $O(\frac1\eps)\times n$ \textsf{CountSketch} matrices $\bfS^{(1)}, \bfS^{(2)}, \dots, \bfS^{(r)}$ and an  $O(\frac1{\alpha^2}(\log\frac{n}{\delta})) \times d$ i.i.d.\ Gaussian matrix $\bfG$. Then, given
    \[
        \bfS^{(1)}\bfA\bfG^\top, \bfS^{(2)}\bfA\bfG^\top, \dots, \bfS^{(r)}\bfA\bfG^\top,
    \]
    one can obtain an estimate $\widehat{\bfA\bfG^\top}\in\mathbb R^{n\times d}$ such that for each $i\in[n]$,
    \[
        \abs*{\norm*{\bfe_i^\top\widehat{\bfA\bfG^\top}}_2 - \norm*{\bfe_i^\top\bfA}_2} \leq \alpha\norm*{\bfe_i^\top\bfA}_2 + \sqrt{(1+\alpha)\eps}\norm*{\bfA_{\overline{[1/\eps]},*}}_F
    \]
    where $\bfA_{\overline{[1/\eps]},*}$ denotes the matrix with the top $1/\eps$ rows with largest $\ell_2$ norm zeroed out.
\end{corollary}
\begin{proof}
    Lemma \ref{lem:cs-row-norm-approximation} provides us with an estimate $\widehat{\bfA\bfG^\top}$ such that
    \[
        \Pr\braces*{\norm*{\bfe_i^\top\widehat{\bfA\bfG^\top} - \bfe_i^\top\bfA\bfG^\top}_2^2 \leq \eps\norm*{(\bfA\bfG^\top)_{\overline{[1/\eps]},*}}_F^2} \geq 1 - \frac{\delta}{2}.
    \]
    Note that $\bfG$ is a Johnson-Lindenstrauss transform (see, e.g., \cite{DBLP:journals/rsa/DasguptaG03}) so that
    \[
        \Pr\braces*{\forall i\in[n], \norm*{\bfe_i^\top\bfA\bfG^\top}_2^2 = (1\pm\alpha)\norm*{\bfe_i^\top\bfA}_2^2} \geq 1 - \frac\delta2.
    \]
    By a union bound, both of these happen with probability at least $1 - \delta$. We now condition on this event.
    
    We first bound $\norm*{(\bfA\bfG^\top)_{\overline{[1/\eps]},*}}_F^2$. This quantity is at most the Frobenius norm of $\bfA\bfG^\top$ after deleting \emph{any} set of $1/\eps$ rows. By letting these $1/\eps$ rows be the $1/\eps$ heaviest rows of $\bfA$, say $S\subseteq[n]$, we have the bound
    \[
        \norm*{(\bfA\bfG^\top)_{\overline{[1/\eps]},*}}_F^2 \leq \sum_{i\in[n]\setminus S} \norm*{\bfe_i\bfA\bfG^\top}_2^2 \leq \sum_{i\in[n]\setminus S} (1+\alpha)\norm*{\bfe_i\bfA}_2^2 = (1+\alpha)\norm*{\bfA_{\overline{[1/\eps]},*}}_F^2.
    \]

    Note then that
    \begin{align*}
        \norm*{\bfe_i^\top\widehat{\bfA\bfG^\top}}_2 &= \norm*{\bfe_i^\top\bfA\bfG^\top}_2 \pm \norm*{\bfe_i^\top\widehat{\bfA\bfG^\top} - \bfe_i^\top\bfA\bfG^\top}_2 \\
        &= (1\pm\alpha)\norm*{\bfe_i^\top\bfA}_2 \pm \sqrt\eps\norm*{(\bfA\bfG^\top)_{\overline{[1/\eps]},*}}_F \\
        &= (1\pm\alpha)\norm*{\bfe_i^\top\bfA}_2 \pm \sqrt{(1+\alpha)\eps}\norm*{\bfA_{\overline{[1/\eps]},*}}_F.\qedhere
    \end{align*}
\end{proof}

\subsubsection{Approximate Matrix Product}

The final guarantee for \textsf{CountSketch} is \emph{approximate matrix product}, from \cite[Lemma 32]{DBLP:conf/stoc/ClarksonW13} (see also \cite{DBLP:conf/soda/ThorupZ04} and Theorems 2.8 and 2.9 of \cite{DBLP:journals/fttcs/Woodruff14}).

\begin{Lemma}[Approximate matrix product for \textsf{CountSketch}]\label{lem:cs-apm}
    Let $\bfA\in\mathbb R^{n\times d}$, $\bfB\in\mathbb R^{m\times d}$, and $\eps>0$. Consider an $r\times d$ \textsf{CountSketch} matrix $\bfS$ for $r = \Omega(\eps^{-2})$. Then,
    \[
        \Pr\braces*{\norm*{\bfA\bfS^\top\bfS\bfB^\top - \bfA\bfB^\top}_F^2 \leq \eps^2 \norm*{\bfA}_F^2 \norm*{\bfB}_F^2} \geq \frac{99}{100}.
    \]
\end{Lemma}
\section{Sparse Singular Vectors}\label{sec:sparse-krylov}

In this section, we discuss our results on performing an approximate SVD with relative spectral norm error, when we are promised that the input matrix $\bfA\in\mathbb R^{n\times d}$ has top $k$ left and right singular vectors that are $s$-sparse. 

\subsection{Approximating Singular Components}

To carry out our plan as described in the introduction (Section \ref{sec:sparse-krylov-overview}), we first calculate the magnitude of coordinates that we need to capture in order to achieve a relative error spectral approximation. We follow \cite{DBLP:conf/nips/MuscoM15} and make use of the fact that additive Frobenius norm low rank approximation implies additive spectral norm low rank approximation, originally due to \cite{DBLP:journals/siamsc/Gu15}. 

\begin{Lemma}[Theorem 3.4 of \cite{DBLP:journals/siamsc/Gu15}]\label{lem:additive-frob-to-spectral}
    For any $\bfA\in\mathbb R^{n\times d}$, let $\bfB\in\mathbb R^{n\times d}$ be any rank $k$ matrix satisfying $\norm*{\bfA-\bfB}_F^2 \leq \norm{\bfA-\bfA_k}_F^2 + \eta$. Then,
    \[
        \norm*{\bfA-\bfB}_2^2 \leq \norm{\bfA-\bfA_k}_2^2 + \eta. 
    \]
\end{Lemma}

By the above result, it suffices to find a rank $k$ matrix $\bfB$ such that
\[
    \norm*{\bfA-\bfB}_F^2 \leq \norm{\bfA-\bfA_k}_F^2 + \eps \sigma_{k+1}^2. 
\]
Using this, we show that it suffices to find all coordinates of the top left singular vectors $\bfU\bfe_j$ such that
\[
    \abs{\bfe_i^\top \bfU\bfe_j} \geq \frac{\eps}{k\sqrt{sr}}\frac{\sigma_{k+1}}{\sigma_j}, 
\]
and similarly, all coordinates of the top right singular vectors $\bfV\bfe_j$ such that
\[
    \abs{\bfe_i^\top \bfV\bfe_j} \geq \frac{\eps}{k\sqrt{sr}}\frac{\sigma_{k+1}}{\sigma_j}.
\]

\begin{Lemma}\label{lem:approx-singular-components}
    Let $\bfA\in\mathbb R^{n\times d}$ have rank $r$ with singular value decomposition $\bfA = \bfU\bfSigma\bfV^\top$, and let $\eps\in(0,1/2)$. Let $S\subset[n]$ and $T\subset[d]$ be a set of coordinates such that
    \begin{align*}
        S &\supset \bigcup_{j\in[r]} \braces*{i\in[n] : \abs{\bfe_i^\top \bfU\bfe_j} \geq \frac{\eps}{k\sqrt{sr}}\frac{\sigma_{k+1}}{\sigma_j}} \\
        T &\supset \bigcup_{j\in[r]} \braces*{i\in[d] : \abs{\bfe_i^\top \bfV\bfe_j} \geq \frac{\eps}{k\sqrt{sr}}\frac{\sigma_{k+1}}{\sigma_j}}
    \end{align*}
    Let $\bfB$ be a rank $k$ matrix such that
    \[
        \norm*{\bfP_S\bfA\bfP_T - \bfB}_F^2 \leq \min_{\text{rank $k$ $\bfC$}} \norm*{\bfP_S\bfA\bfP_T - \bfC}_F^2 + \eta.
    \]
    Then,
    \[
        \norm*{\bfA - \bfB}_F^2 \leq \norm*{\bfA - \bfA_k}_F^2 + 8\eps\sigma_{k+1}^2 + \eta.
    \]
\end{Lemma}
\begin{proof}
    Note first that
    \[
        \norm*{\bfA - \bfA_k}_F^2 = \sum_{t=k+1}^{r}\sigma_{t}^2(\bfA) \leq \sum_{t=k+1}^{r}\sigma_{k+1}^2(\bfA)\leq \sigma_{k+1}^2(\bfA)r.
    \]
    Then,
    \begin{align*}
        \norm*{\bfA - \bfB}_F^2 &= \norm*{\bfA - \bfP_S\bfA\bfP_T}_F^2 + \norm*{\bfP_S\bfA\bfP_T - \bfB}_F^2 \\
        &\stackrel{(1)}{\leq} \norm*{\bfA - \bfP_S\bfA\bfP_T}_F^2 + \norm*{\bfP_S\bfA\bfP_T - \bfA_k}_F^2 + \eta \\
        &\stackrel{(2)}{=} \norm*{\bfA - \bfP_S\bfA\bfP_T}_F^2 + \norm*{\bfP_S\bfA\bfP_T - \bfP_S\bfA_k\bfP_T}_F^2 + \norm*{\bfA_k - \bfP_S\bfA_k\bfP_T}_F^2 + \eta \\
        &\stackrel{(3)}{=} \norm*{\bfA - \bfP_S\bfA_k\bfP_T}_F^2 + \norm*{\bfA_k - \bfP_{S}\bfA_k\bfP_T}_F^2 + \eta \\
        &\stackrel{(4)}{\leq} \parens*{\norm*{\bfA - \bfA_k}_F + \norm*{\bfA_k - \bfP_S\bfA_k\bfP_T}_F}^2 + \norm*{\bfA_k - \bfP_{S}\bfA_k\bfP_T}_F^2 + \eta \\
        &= \norm*{\bfA - \bfA_k}_F^2 + 2\norm*{\bfA - \bfA_k}_F\norm*{\bfA_k - \bfP_S\bfA_k\bfP_T}_F + 2\norm*{\bfA_k - \bfP_{S}\bfA_k\bfP_T}_F^2 + \eta \\
        &\leq \norm*{\bfA - \bfA_k}_F^2 + 2\sigma_{k+1}\sqrt r\norm*{\bfA_k - \bfP_{S}\bfA_k\bfP_T}_F + 2\norm*{\bfA_k - \bfP_{S}\bfA_k\bfP_T}_F^2 + \eta
    \end{align*}
    In the above, the inequality (1) is due to the approximate optimality of $\bfB$, the identities (2) and (3) are by the Pythagorean theorem, and inequality (4) is the triangle inequality. Finally, we calculate that
    \begin{align*}
        \norm*{\bfA_k - \bfP_{S}\bfA_k\bfP_T}_F &\leq \norm*{\bfA_k - \bfP_{S}\bfA_k}_F + \norm*{\bfP_{S}\bfA_k - \bfP_S\bfA_k\bfP_T}_F  \\
        &= \norm*{\sum_{j=1}^k \sigma_j \bfP_{\overline S}\bfU\bfe_j (\bfV\bfe_j)^\top}_F + \norm*{\sum_{j=1}^k \sigma_j \bfP_S\bfU\bfe_j (\bfV\bfe_j)^\top \bfP_{\overline T}}_F \\
        &\leq \sum_{j=1}^k \sigma_j\norm*{\bfP_{\overline S} \bfU\bfe_j}_2 \norm*{\bfV\bfe_j}_2 + \sigma_j\norm*{\bfP_{S} \bfU\bfe_j}_2 \norm*{\bfP_{\overline T}\bfV\bfe_j}_2 \\
        &\leq \sum_{j=1}^k 2\sigma_j \parens*{\frac{\eps}{k\sqrt{sr}}\frac{\sigma_{k+1}}{\sigma_j}} \sqrt s \\
        &= \frac{2\eps}{\sqrt r}\sigma_{k+1}
    \end{align*}
    so the previous bound is
    \begin{align*}
        \norm*{\bfA - \bfB}_F^2 &\leq \norm*{\bfA - \bfA_k}_F^2 + 4\sigma_{k+1}\sqrt r\norm*{\bfA_k - \bfP_{S}\bfA_k\bfP_T}_F + 2\norm*{\bfA_k - \bfP_{S}\bfA_k\bfP_T}_F^2 + \eta \\
        &\leq \norm*{\bfA - \bfA_k}_F^2 + 4\eps\sigma_{k+1}^2 + \frac{8\eps^2}{r}\sigma_{k+1}^2 + \eta \\
        &\leq \norm*{\bfA - \bfA_k}_F^2 + 8\eps\sigma_{k+1}^2 + \eta.\qedhere
    \end{align*}
\end{proof}

\subsection{Finding the Support of Singular Vectors with Large Singular Value}\label{sec:singular-vec-support}

We next show how to find all large coordinates of singular vectors whose singular values $\sigma_j$ are at least a $(1+\sqrt\eps)$ factor larger than $\sigma_{k+1}$. By the results of the previous section, we seek to find all of the large coordinates of the top sparse singular vectors, which have absolute value at least
\[
    \tau_j \coloneqq \frac{\eps}{k\sqrt{sr}} \frac{\sigma_{k+1}}{\sigma_j}
\]
for the $j$th singular vector. 

Our identification of the large coordinates of the top sparse singular vectors starts from the standard analysis of the power method (see also, e.g., the overview of \cite{DBLP:conf/nips/MuscoM15}). If we run power method starting from a random Gaussian vector $\bfg\sim \mathcal N(0,\bfI_d)$, that is, we compute $(\bfA\bfA^\top)^q \bfA\bfg$ for some $q\in\mathbb N$, then we retrieve a random Gaussian linear combination of the left singular vectors $\bfU\bfe_j$, each scaled by $\sigma_j^{2q+1}$. This is a simple consequence of the rotational invariance of the Gaussian:

\begin{Lemma}
    Let $\bfg'\sim\mathcal N(0,\bfI_d)$ and let $q\in\mathbb N$. Let $\bfA\in\mathbb R^{n\times d}$ be a rank $r$ matrix and let $\bfA  = \bfU\bfSigma \bfV^\top$ be its singular value decomposition. Then, $(\bfA\bfA^\top)^q \bfA \bfg'$ has the same distribution as 
    \[
        \bfU\bfSigma^{2q+1}\bfg = \sum_{j=1}^r \bfg_j \sigma_j^{2q+1}\bfU\bfe_j
    \]
    for $\bfg\sim \mathcal N(0, \bfI_r)$. 
\end{Lemma}

Note then that for $\sigma_j \geq (1+\sqrt\eps)\sigma_{k+1}$, the $j$th singular vector is scaled more than the $(k+1)$-st singular vector by a factor of at least $(\sigma_j / \sigma_{k+1})^{2q + 1}$. For $q$ roughly order $1/\sqrt\eps$, this separates all large coordinates of the $j$th singular vector from the coordinates of the $(k+1)$-st singular vector. 

\begin{Lemma}\label{lem:large-gaussian-coordinate}
    For 
    \[
        q = O\parens*{\frac1{\sqrt\eps}\log \frac{sk^2\sqrt{sr\log n}}{\eps}},
    \]
    the $sk$ coordinates of $(\bfA\bfA^\top)^q\bfA\bfg$ with largest absolute value are guaranteed to contain all entries $i\in[n]$ for which there exists a $j\in[k]$ with $\sigma_j \geq (1+\sqrt\eps)\sigma_{k+1}$ and
    \[
        \abs*{\bfe_i^\top \bfU\bfe_j} \geq \tau_j.
    \]
\end{Lemma}
\begin{proof}
    For 
\[
    q = O\parens*{\frac1{\sqrt\eps}\log \frac{sk^2\sqrt{sr\log n}}{\eps}},
\]
the blow up factor $(\sigma_j/\sigma_{k+1})^{2q+1}$ is at least 
\[
    \parens*{\frac{\sigma_j}{\sigma_{k+1}}}^{2q+1} \geq (1+\sqrt\eps)^{2q}\frac{\sigma_j}{\sigma_{k+1}} = \Theta\parens*{\frac{sk^2\sqrt{sr\log n}}{\eps}}\frac{\sigma_j}{\sigma_{k+1}} = \Theta\parens*{\frac{sk\sqrt{\log n}}{\tau_j}}
\]
for the $j$th singular component. The time required to compute this vector $(\bfA\bfA^\top)^q\bfA\bfg$ is
\[
    O\parens*{\frac{\nnz(\bfA)}{\sqrt\eps}\log \frac{sk^2\sqrt{sr\log n}}{\eps}} = O\parens*{\frac{\nnz(\bfA)}{\sqrt\eps}\log \frac{srk\log n}{\eps}}
\]
Now note that for each $i\in[n]$, we have that
\[
    \bfe_i^\top \bfU\bfSigma^{2q+1}\bfg \sim \mathcal N\parens*{0, \norm*{\bfe_i^\top \bfU\bfSigma^{2q+1}}_2^2}.
\]
Since the maximum absolute value among $n$ Gaussians is $O(\sqrt{\log n})$ with constant probability, we have
\[
    \abs*{\bfe_i^\top \bfU\bfSigma^{2q+1}\bfg} \leq  O(\sqrt{\log n})\norm*{\bfe_i^\top \bfU\bfSigma^{2q+1}}_2.
\]
Furthermore, if we consider all $i$ in the support of the top $k$ singular vectors, which is at most $sk$ coordinates, then the minimum absolute value among the $sk$ Gaussians is
\[
    \abs*{\bfe_i^\top \bfU\bfSigma^{2q+1}\bfg} \geq  \Omega\parens*{\frac1{sk}}\norm*{\bfe_i^\top \bfU\bfSigma^{2q+1}}_2.
\]

Now consider a coordinate $i\in[n]$ such that
\[
    \abs*{\bfe_i^\top \bfU\bfe_j} \geq \tau_j
\]
for some $j\in[k]$ such that $\sigma_j \geq (1+\sqrt\eps)\sigma_{k+1}$. Then by the previous results,
\begin{align*}
    \abs*{\bfe_i^\top \bfU\bfSigma^{2q+1}\bfg} &\geq  \Omega\parens*{\frac1{sk}}\norm*{\bfe_i^\top \bfU\bfSigma^{2q+1}}_2 \\
    &\geq \Omega\parens*{\frac1{sk}}\sigma_{j}^{2q+1} \tau_j \\
    &= \Omega\parens*{\frac1{sk}}\sigma_{k+1}^{2q+1}\parens*{\frac{\sigma_j}{\sigma_{k+1}}}^{2q+1} \tau_j \\
    &\geq \Omega\parens*{\frac1{sk}}\sigma_{k+1}^{2q+1} \Theta\parens*{\frac{sk\sqrt{\log n}}{\tau_j}} \tau_j \\
    &= \Omega\parens*{\sigma_{k+1}^{2q+1}\sqrt{\log n}}.
\end{align*}
On the other hand, for any $i\in[n]$ that is outside of the at most $sk$ coordinates of the support of the top $k$ singular vectors, then
\[
    \abs*{\bfe_i^\top \bfU\bfSigma^{2q+1}\bfg} \leq  O(\sqrt{\log n})\norm*{\bfe_i^\top \bfU\bfSigma^{2q+1}}_2 \leq O(\sigma_{k+1}^{2q+1}\sqrt{\log n}).
\]
We thus conclude as desired.
\end{proof}

In other words, we can identify a set of $sk$ coordinates that contains all large entries of left singular vectors $j$ for which $\sigma_j \geq (1+\sqrt\eps)\sigma_{k+1}$. Repeating for the right singular vectors, we may identify the sets $S$ and $T$ as required by Lemma \ref{lem:approx-singular-components}. 

\subsection{Approximating Large Singular Values}

Our next task is to compute the singular values of $\bfA$ with $\sigma_j(\bfA) \geq (1+\sqrt\eps)\sigma_{k+1}(\bfA)$, up to $(1+\eps)$ factors. We first show that approximating the singular values of $\bfP_S\bfA\bfP_T$ directly approximates the singular values of $\bfA$, when the singular values are sufficiently large. 

\begin{Lemma}\label{lem:submatrix-approx-singular-values}
    Let $m$ be the number of singular values of $\bfA$ such that $\sigma_j(\bfA) \geq (1+\sqrt\eps)\sigma_{k+1}(\bfA)$. Let $S\subset[n]$ and $T\subset[d]$ be sets satisfying the hypotheses of Lemma \ref{lem:approx-singular-components}. Then for each $l\in[m]$, 
    \[
        (1-8\eps)\sigma_l^2(\bfA) \leq \sigma_l^2(\bfP_S\bfA\bfP_T) \leq \sigma_l^2(\bfA).
    \]
\end{Lemma}
\begin{proof}
    Recall the Cauchy interlacing theorem:
    \begin{theorem}[Cauchy interlacing theorem]
        Let $\bfM$ be a symmetric matrix and let $\bfN$ be a principal submatrix of size $l\times l$. Then for all $j\in[l]$, 
        \[
            \lambda_j(\bfM) \geq \lambda_j(\bfN) \geq \lambda_{n-l+j}(\bfM).
        \]
    \end{theorem}
    Then applying the interlacing theorem to $\bfM = \bfA\bfA^\top$ and $\bfN = \bfP_S\bfA\bfA^\top\bfP_S^\top$, we find that the singular values of $\bfP_S\bfA$ uniformly bound the top $sk$ singular values of $\bfA$ from below, and similarly, the singular values of $\bfP_S\bfA\bfP_T$ uniformly bound the singular values of $\bfP_S\bfA$ from below. We thus have that
    \[
        \sigma_j(\bfA) \geq \sigma_j(\bfP_S\bfA) \geq \sigma_j(\bfP_S\bfA\bfP_T)
    \]
    for all $j\in[sk]$. Furthermore, we know by Lemma \ref{lem:approx-singular-components} that for each $l\in[m]$, 
    \begin{equation}\label{eq:frob-guarantee}
        \norm*{\bfA - \bfA_l}_F^2 \leq \norm*{\bfA - (\bfP_S\bfA\bfP_T)_l}_F^2 \leq \norm*{\bfA - \bfA_l}_F^2 + 8\eps \sigma_{l+1}(\bfA)^2
    \end{equation}
    where $(\bfP_S\bfA\bfP_T)_l$ is the best rank $l$ approximation $\bfP_S\bfA\bfP_T$. Now note that
    \[
        \norm*{\bfA}_F^2 - \norm*{\bfA - \bfA_l}_F^2 = \norm*{\bfA_l}_F^2
    \]
    and
    \begin{align*}
        \angle*{\bfA - (\bfP_S\bfA\bfP_T)_l, (\bfP_S\bfA\bfP_T)_l} &= \angle*{\bfA - \bfP_S\bfA\bfP_T, (\bfP_S\bfA\bfP_T)_l} \\
        &\hspace{5em} + \angle*{\bfP_S\bfA\bfP_T - (\bfP_S\bfA\bfP_T)_l, (\bfP_S\bfA\bfP_T)_l} = 0
    \end{align*}
    so 
    \[
        \norm*{\bfA}_F^2 - \norm*{\bfA - (\bfP_S\bfA\bfP_T)_l}_F^2 = \norm*{(\bfP_S\bfA\bfP_T)_l}_F^2
    \]
    by the Pythagorean theorem. Then subtracting the inequalities of Equation \ref{eq:frob-guarantee} from $\norm*{\bfA}_F^2$, we have that
    \[
        \norm*{\bfA_l}_F^2 - 8\eps \sigma_{l+1}(\bfA)^2 \leq \norm*{(\bfP_S\bfA\bfP_T)_l}_F^2 \leq \norm*{\bfA_l}_F^2.
    \]
    Then, 
    \begin{align*}
        \sigma_{l}^2(\bfP_S\bfA\bfP_T) &= \norm*{(\bfP_S\bfA\bfP_T)_l}_F^2 - \norm*{(\bfP_S\bfA\bfP_T)_{l-1}}_F^2 \\
        &\geq \norm*{\bfA_l}_F^2 - 8\eps \sigma_{l+1}(\bfA)^2 - \norm*{\bfA_{l-1}}_F^2 \\
        &= \sigma_l^2(\bfA) - 8\eps \sigma_{l+1}(\bfA)^2 \\
        &\geq (1-8\eps)\sigma_l^2(\bfA)
    \end{align*}
    as desired.
    \end{proof}

We may use the existing results of \cite{DBLP:conf/nips/MuscoM15} to find $(1+\eps)$ factor approximations to the top $k$ singular values of $\bfP_S\bfA\bfP_T$ in time
\[
    O\parens*{\frac{\nnz(\bfP_S \bfA\bfP_T)k}{\sqrt\eps}\log(sk)} = O\parens*{\frac{s^2k^3}{\sqrt\eps}\log(sk)}. 
\]
However, note that given estimates for the singular values of $\bfP_S\bfA\bfP_T$, we do not know which ones are within a $(1+\eps)$ factor of the singular values of $\bfA$, since we do not know the number $m$ of singular values $j$ with $\sigma_j(\bfA) \geq (1+\sqrt\eps)\sigma_{k+1}(\bfA)$. However, by the Cauchy interlacing theorem, the singular values of $\bfP_S\bfA\bfP_T$ are always a lower bound on the singular values of $\bfA$, so it suffices to compute an upper bound for the singular values of $\bfA$ that are at most a $(1+\eps)$ factor larger than the lower bound. We obtain such an upper bound on the singular values of $\bfA$ by approximating $\norm*{\bfA - \bfB}_2$ for a rank $l$ matrix $\bfB$. Indeed, if $\bfB$ is rank $l$, then
\[
    \norm*{\bfA - \bfB}_2^2 \geq \min_{\text{rank $l$ $\bfC$}}\norm*{\bfA - \bfC}_F^2 = \sigma_{l+1}(\bfA)^2. 
\]
This idea is executed in the following lemma. 

\begin{Lemma}\label{lem:approximate-upper-lower-bound}
Let $S\subset[n]$ and $T\subset[d]$ be sets of size $sk$ each that satisfy the hypotheses of Lemma \ref{lem:approx-singular-components}. Given such $S$ and $T$ and an index $j\in[k]$, there is a randomized algorithm that runs in time
\[
    O\parens*{\frac{\nnz(\bfA) + s^2k^3}{\sqrt\eps} \log(sk)}
\]
and outputs numbers $U$ and $L$ such that
\[
    L \leq \sigma_j^2(\bfA) \leq U
\]
with probability at least $0.99$. Furthermore, if $j\in[m]$, where $m$ is the number of singular values $j$ with $\sigma_j \geq (1+\sqrt\eps)\sigma_{k+1}$, we have that
\[
    \frac{U}{L} \leq \frac{1+10\eps}{1-9\eps} \leq 1 + 20\eps. 
\]
\end{Lemma}
\begin{proof}
    We first show how to obtain the lower bound $L$. By the Cauchy interlacing theorem (as in Lemma \ref{lem:submatrix-approx-singular-values}), we have that
    \[
        \sigma_j(\bfP_S\bfA\bfP_T)\leq \sigma_j(\bfA).
    \]
    Then by the randomized block Krylov algorithm of \cite{DBLP:conf/nips/MuscoM15} (see Theorem \ref{thm:randomized-block-krylov}), we may find an estimate $L$ to $\sigma_j(\bfP_S\bfA\bfP_T)$ such that
    \[
        (1-\eps)\sigma_j(\bfP_S\bfA\bfP_T) \leq L \leq \sigma_j(\bfP_S\bfA\bfP_T)
    \]
    in time
    \[
        O\parens*{\frac{\nnz(\bfP_S\bfA\bfP_T)k}{\sqrt\eps}\log(sk)} = O\parens*{\frac{s^2k^3}{\sqrt\eps}\log(sk)}. 
    \]
    Furthermore, if $j\in[m]$, then by Lemma \ref{lem:submatrix-approx-singular-values}, 
    \[
        L \geq (1-\eps)\sigma_j(\bfP_S\bfA\bfP_T) \geq (1-\eps)(1-8\eps)\sigma_j(\bfA) \geq (1-9\eps)\sigma_j(\bfA).
    \]
    
    For the upper bound, we use the rank $j$ approximation $\bfB$ obtained by running the randomized block Krylov algorithm of \cite{DBLP:conf/nips/MuscoM15} on $\bfP_S\bfA\bfP_T$. Note that
    \[
        \sigma_j(\bfA) = \min_{\text{rank $j$ $\bfC$}} \norm*{\bfA-\bfC}_2 \leq \norm*{\bfA-\bfB}_2
    \]
    for any rank $j-1$ matrix $\bfB$. By the results of \cite{DBLP:conf/nips/MuscoM15}, we may compute an estimate $U$ such that
    \[
        (1+\eps)\norm*{\bfA-\bfB}_2 \geq U \geq \norm*{\bfA-\bfB}_2
    \]
    in time
    \[
        O\parens*{\frac{\nnz(\bfA - \bfB)}{\sqrt\eps}} = O\parens*{\frac{\nnz(\bfA) + s^2k^2}{\sqrt\eps}}. 
    \]
    Furthermore, for $j\in[m]$, if we find a rank $j-1$ matrix $\bfB$ such that 
    \begin{align*}
        \norm*{\bfP_S\bfA\bfP_T - \bfB}_F^2 &\leq \min_{\text{rank $j-1$ $\bfC$}} \norm*{\bfP_S\bfA\bfP_T - \bfC}_F^2 + \eps \sigma_{j}(\bfP_S\bfA\bfP_T)^2 \\
        &\leq \min_{\text{rank $j-1$ $\bfC$}} \norm*{\bfP_S\bfA\bfP_T - \bfC}_F^2 + \eps \sigma_{j}(\bfA)^2,
    \end{align*}
    which we can by the results of \cite{DBLP:conf/nips/MuscoM15} as before, then by Lemma \ref{lem:approx-singular-components},
    \[
        \norm*{\bfA - \bfB}_F^2 \leq \norm*{\bfA - \bfA_{j-1}}_F^2 + 9\eps\sigma_{j}^2(\bfA).
    \]
    By Lemma \ref{lem:additive-frob-to-spectral}, this implies that
    \[
        \norm*{\bfA - \bfB}_2^2 \leq \norm*{\bfA - \bfA_{j-1}}_2^2 + 9\eps\sigma_{j}^2(\bfA)= (1+9\eps)\sigma_{j}^2(\bfA).\qedhere
    \]
\end{proof}

We now show how to use the above result to efficiently find a $(1+\sqrt\eps)$ factor approximation to $\sigma_{k+1}(\bfA)$ using binary search. 

\begin{Lemma}
    There is a randomized algorithm that runs in time
    \[
        O\parens*{\frac{\nnz(\bfA) + s^2k^3}{\sqrt\eps} \log(sk)(\log k)}
    \]
    that finds a $(1+\sqrt\eps)$ factor approximation to $\sigma_{k+1}(\bfA)$. 
\end{Lemma}
\begin{proof}
If $\sigma_{k}(\bfA) \geq (1+\sqrt\eps)\sigma_{k+1}(\bfA)$, then deflating off the top $k$ components already gives a $(1+\eps)$ factor approximation to $\sigma_{k+1}(\bfA)$. Otherwise, we proceed with binary search as follows. 

Suppose we consider $j\in[k]$. If the upper and lower bounds for $\sigma_j(\bfA)$ in Lemma \ref{lem:approximate-upper-lower-bound} are within a $(1+O(\eps))$ factor, then we know that $\sigma_{k+1}(\bfA)$ is smaller than this, up to a $(1\pm O(\eps))$ factor. On the other hand, if the upper and lower bounds for $\sigma_j(\bfA)$ are further than a $(1+O(\eps))$ factor, then $\sigma_j(\bfA) \leq (1+\sqrt\eps)\sigma_{k+1}(\bfA)$, since otherwise the upper and lower bounds for $\sigma_j(\bfA)$ would have matched up to a $(1\pm O(\eps))$ factor by the second guarantee of Lemma \ref{lem:approximate-upper-lower-bound}. Thus, we may use binary search over the at most $k$ singular values in at most $O(\log k)$ calls to the algorithm of Lemma \ref{lem:approximate-upper-lower-bound}. 
\end{proof}

\subsection{Approximating Small Singular Values}

With a $(1+\sqrt\eps)$ factor approximation to $\sigma_{k+1}(\bfA)$ in hand, we now zoom into the singular values between $\sigma_{k+1}(\bfA)$ and $(1+\sqrt\eps)\sigma_{k+1}(\bfA)$. We consider partitioning this $(1+\sqrt\eps)$ factor window into $O(1/\sqrt\eps)$ buckets that increase in powers of $(1+\eps)$, that is
\[
    L, L(1+\eps), L(1+\eps)^2, L(1+\eps)^3, \dots, L(1+\eps)^{O(1/\sqrt\eps)} = (1+\sqrt\eps)L
\]
where $L$ is a lower bound on $\sigma_{k+1}(\bfA)$, up to a $(1+\sqrt\eps)$ factor. Our idea now is to simply enumerate over these $O(1/\sqrt\eps)$ guesses to a $(1\pm\eps)$-approximation of $\sigma_{k+1}(\bfA)$, and then choose the best result. 

With only a $(1+\eps)$ factor gap in the singular values, using power method as before will require roughly (ignoring log factors) $1/\eps$ iterations, which takes time roughly $\nnz(\bfA)/\eps$ to separate out the singular components, which is above our target budget. However, using Chebyshev polynomials, it is known that a $(1+\eps)$ factor gap in the singular values can be separated with only roughly $1/\sqrt\eps$ iterations \cite{DBLP:conf/nips/MuscoM15} which takes time only $\nnz(\bfA)/\sqrt\eps$. The main lemma for this technique is the following:

\begin{Lemma}[Lemma 5, \cite{DBLP:conf/nips/MuscoM15}]\label{lem:chebyshev-poly}
    Given a specified value $\alpha > 0$, gap $\gamma\in(0,1]$, and $q\geq 1$, there exists a degree $q$ polynomial $p(x)$ such that:
    \begin{enumerate}
        \item $p((1+\gamma)\alpha) = (1+\gamma)\alpha$
        \item $p(x)\geq x$ for all $x\geq (1+\gamma)\alpha$
        \item $\abs{p(x)} \leq \frac\alpha{2^{q\sqrt{\gamma}-1}}$ for all $x\in[0,\alpha]$
    \end{enumerate}
    Furthermore, when q is odd, the polynomial only contains odd powered monomials. 
\end{Lemma}

In words, the above lemma states that there is a polynomial that ``jumps'' by a factor of $2^{q\sqrt\gamma-1}$ in a window of size $(1+\gamma)$ at a specified location $\alpha$. The difference between this lemma and our power method analysis from before is that we must specify the location of our ``jump'', $\alpha$, in order to use the above polynomial in the Krylov method, whereas in the power method, the polynomial $p(x) = x^q$ had the ``jump'' property at any location $\alpha$. Thus, in order to use the above lemma, we must \emph{first} specify our jump location $\alpha$, and then proceed with our previous techniques. 

Our procedure is thus as follows. We first compute Krylov iterates $(\bfA\bfA^\top)^i \bfA\bfg$ for $i\in[q]$, where $\bfg\sim\mathcal N(0, \bfI_d)$ and
\[
    q = O\parens*{\frac{1}{\sqrt\eps}\log\frac{sk^2\sqrt{sr\log n}}{\eps}}.
\]
We then proceed with our enumeration procedure. We guess a bucket $\alpha = L(1+\eps)^t$ for some $t\in[O(1/\sqrt\eps)]$, and then consider the degree $q$ polynomial $p_\alpha(x)$ that jumps by a $2^{q\sqrt\eps-1}$ factor at $\alpha$ by Lemma \ref{lem:chebyshev-poly}. Then, we may compute the vector $\bfU p_\alpha(\bfSigma)\bfV^\top \bfg$ as a linear combination of the Krylov iterates
\[
    (\bfA\bfA^\top)^i \bfA\bfg = \bfU\bfSigma^{2i+1}\bfV^\top\bfg
\]
where the coefficients of the linear combination are the coefficients of the polynomial $p_\alpha$. Next, we take the top $sk$ entries of $\bfU p_\alpha(\bfSigma)\bfV^\top \bfg$ as sets $S_\alpha$ and $T_\alpha$, combine them with the $sk$ entries $S$ and $T$ obtained earlier by the power method, and then take our new subset of entries to be 
\begin{align*}
    S' \coloneqq S \cup \bigcup_\alpha S_\alpha \\
    T' \coloneqq T \cup \bigcup_\alpha T_\alpha
\end{align*}
Finally we compute a rank $k$ matrix $\bfB$ such that
\[
    \norm*{\bfP_{S'}\bfA\bfP_{T'} - \bfB}_F^2 \leq \min_{\text{rank $k$ $\bfC$}} \norm*{\bfP_{S'}\bfA\bfP_{T'} - \bfC}_F^2 + \eps\sigma_{k+1}^2(\bfP_{S'}\bfA\bfP_{T'})
\]
using the results of \cite{DBLP:conf/nips/MuscoM15}.

Note that if the $\alpha$ we choose satisfies $\alpha\in [\sigma_{k+1}(\bfA), (1+\eps)\sigma_{k+1}(\bfA)]$, then all singular values $j$ that are at least a $(1+\eps)$ factor larger than $\alpha$ and at most $\Theta(1)\sigma_{k+1}(\bfA)$ are scaled by at least a factor of
\[
    2^{q\sqrt\eps - 1} = \Theta\parens*{\frac{sk^2\sqrt{sr\log n}}{\eps}} = \Theta\parens*{\frac{sk^2\sqrt{sr\log n}}{\eps}}\frac{\sigma_j}{\sigma_{k+1}} = \Theta\parens*{\frac{sk\sqrt{\log n}}{\tau_j}},
\]
which means we may recover all coordinates of the $j$th singular vectors that are at least $\tau_j$ for these singular values, as done in the analyses in Section \ref{sec:singular-vec-support}. Thus, we have that
\begin{align*}
    \norm*{\bfA - \bfB}_2^2 &\leq \norm*{\bfA - \bfA_l}_2^2 + 8\eps\sigma_{l+1}^2(\bfA) + \eps\sigma_{l+1}^2(\bfP_{S'}\bfA\bfP_{T'}) \\
    &\leq \norm*{\bfA - \bfA_l}_2^2 + 9\eps\sigma_{l+1}^2(\bfA) \\
    &= (1+9\eps)\sigma_{l+1}^2(\bfA) \\
    &\leq (1+9\eps)(1+\eps)\sigma_{k+1}^2(\bfA) \\
    &\leq (1+11\eps)\sigma_{k+1}^2(\bfA)
\end{align*}
by Lemma \ref{lem:approx-singular-components}, where $l\in[k]$ is such that $\sigma_{k+1}^2(\bfA) \leq \sigma_{l+1}^2(\bfA) \leq (1+\eps)\sigma_{k+1}^2(\bfA)$. 

The initial computation of the Krylov iterates takes time
\[
    O(\nnz(\bfA)q) = O\parens*{\frac{\nnz(\bfA)}{\sqrt\eps}\log\frac{srk\log n}{\eps}}
\]
and a single guess of $\alpha$ takes time
\[
    O\parens*{nq} = O\parens*{\frac{n}{\sqrt\eps}\log\frac{srk\log n}{\eps}}
\]
which we repeat $O(1/\sqrt\eps)$ times, so the total running time in this section is
\[
    O\parens*{\parens*{\frac{\nnz(\bfA)}{\sqrt\eps} + \frac{n}{\eps}}\log\frac{srk\log n}{\eps}}. 
\]
We then additionally run an approximate SVD using Theorem \ref{thm:randomized-block-krylov} on the $O(sk/\sqrt\eps)\times O(sk/\sqrt\eps)$ matrix, which adds an $s^2 k^3(\log (sk))/\eps^{3/2}$ term, for a running time of
\[
    O\parens*{\parens*{\frac{\nnz(\bfA)}{\sqrt\eps} + \frac{n}{\eps}}\log\frac{srk\log n}{\eps} + \frac{s^2k^3}{\eps^{3/2}}\log(sk)}. 
\]
This dominates the running times of the previous steps and thus is the running time of our entire algorithm. 
\section{Frobenius Sparse Low Rank Approximation}\label{sec:slra}

We now switch to discussing sparse low rank approximation with no assumptions on the input matrix, in the Frobenius norm, in the streaming setting. 

\subsection{Algorithms}

\subsubsection{Exponential Time Algorithms for Sparse Output}

Our first result is a na\"ive exponential time algorithm which iterates over an $\eps$-net to find a good $s\times s$ sparse rank $k$ approximation, taking time roughly $\exp(O(sk\log n))$ and only $O(sk(\log n)/\eps^2)$ sketching dimensions. In order to make our net argument, we will assume for this section that the $\tau_i$ in Definition \ref{def:sxs-sparse-rank-k-matrix} are bounded by $\tau_i\leq\poly(n)$, to support Corollary \ref{cor:sxs-sparse-rank-k-eps-net}. The exponential dependence on $s$ and $k$ is necessary, as we show later in Section \ref{sec:computational-complexity}. The sketch is just a Gaussian matrix, and the analysis of the algorithm follows from Gordon's theorem, which is a dimensionality reduction result which shows the preservation of norms of sets under random projections, when the target dimension scales with the \emph{Gaussian width} of the set.  

\begin{Definition}[Gaussian width]
    The \emph{Gaussian width} of a subset $S\subset\mathbb R^n$ is defined as
    \[
        w(S) \coloneqq \E_{\bfg\sim\mathcal N(0,\bfI_n)}\bracks*{\sup_{\bfx\in S}\angle*{\bfg, \bfx}}.
    \]
\end{Definition}

Gordon's theorem essentially states that a random projection to approximately
\[
    \frac1{\eps^2}\parens*{w(S)^2 + \log\frac1\delta}
\]
dimensions suffices to preserve the norms of all points in $S$, with probability at least $1-\delta$.

\begin{theorem}[Gordon's Theorem (Theorem 9.11, \cite{bandeira2020mathematics})]\label{thm:gordon}
Let $\bfG\sim\mathcal N(0,1)^{m\times d}$ and let $S\subseteq \mathbb S^{d-1}$ be a closed subset of the $d$-dimensional unit sphere. Let
\[
    a_m \coloneqq \E_{\bfg\sim\mathcal N(0,\bfI_{m})}\norm{\bfg}_2 = \Theta(\sqrt{m}).
\]
Then for $\eps > \sqrt{w(S)^2/a_m^2}$, 
\[
    \Pr\braces*{(1-\eps)\norm{\bfx}_2 \leq \frac1{a_m}\norm{\bfG\bfx}_2 \leq (1+\eps)\norm{\bfx}_2, \forall \bfx\in S} \geq 1 - 2\exp\parens*{-\parens*{\eps - \frac{w(S)}{a_m}}^2 m}.
\]
\end{theorem}

The use of Gordon's theorem for low rank matrix recovery is standard, see e.g.\ \cite{vershynin2015estimation}. We show that this technique can be used for sparse low rank matrix approximation as well, by setting the failure rate $\delta$ to $\binom{n}{s}^{-2}$ and then iterating over all $s\times s$ submatrices of $\bfA$. 

\begin{theorem}\label{thm:net-iterate-fslra}
    There is a randomized sketching algorithm which solves the Frobenius sparse low rank approximation with
    \[
        O\parens*{\frac{sk}{\eps^2}\log\frac{n}{s}}
    \]
    measurements and takes time
    \[
        (\poly(n)k/\eps)^{sk}
    \]
    to output a $\bfD\in\mathcal S_{s,k}$ such that
    \[
        \norm{\bfA - \bfD}_F^2 \leq (1+\eps)\min_{\bfC\in\mathcal S_{s,k}}\norm*{\bfA - \bfC}_F^2.
    \]
\end{theorem}
\begin{proof}
    Note that we may write the objective as 
    \[
        \min_{\text{$s\times s$ sparse rank $k$ $\bfC$}} \norm{\bfA - \bfC}_F^2 = \min_{S, T\in \binom{[n]}{s}\times \binom{[n]}{s}} \min_{\text{rank $k$ $\bfC$}}\norm{\bfA\vert_{S\times T} - \bfC}_F^2
    \]
    where $\bfA\vert_{S\times T}$ denotes the restriction of $\bfA$ to the submatrix indexed by coordinates $S \times T$. 
    
    By \cite[Proposition 10.4]{vershynin2015estimation}, the Gaussian mean width of the set
    \[
        D = \braces*{\bfX\in\mathbb R^{s\times s} : \norm{\bfX}_F^2 = 1, \rank(\bfX) \leq k}
    \]
    is at most
    \[
        w(D) \leq 4\sqrt{sk}. 
    \]
    We now fix $S,T\in\binom{[n]}{s}\times\binom{[n]}{s}$. Then by the translational invariance of Gaussian width \cite[Proposition 3.5]{vershynin2015estimation} as well as the fact that padding zeros does not change the Gaussian width, we also have that $w(D')\leq 4\sqrt{sk}$ for the set
    \[
        D' = \braces*{\bfA - \bfX \in\mathbb R^{n\times d} : \norm{\bfX}_F^2 = 1, \rank(\bfX) \leq k, \text{$\bfX$ supported on $S\times T$}}.
    \]
    Now let
    \[
        m = \Omega\parens*{\frac1{\eps^2}\parens*{w(D') + \log\binom{n}{s}}} = \Omega\parens*{\frac{sk}{\eps^2}\log \frac{n}{s}}.
    \]
    Then by Gordon's theorem (Theorem \ref{thm:gordon}), an $m\times s^2$ Gaussian matrix $G\sim\mathcal N(0,1)^{m\times s^2}$ will satisfy
    \[
        \norm*{\frac1{a_m} \bfG\vc(\bfA - \bfX)}_2^2 = (1\pm\eps) \norm{\bfA - \bfX}_F^2
    \]
    for all $\bfX$ of rank at most $k$ supported on $S\times T$, with probability at least $1 - \binom{n}{s}^{-2}/200$. Thus, solving the problem
    \[
        \min_{\bfX\in D}\norm*{\frac1{a_m} \bfG\vc(\bfA - \bfX)}_2^2
    \]
    provides a $(1+\eps)$ multiplicative error solution for
    \[
        \min_{\bfX\in D}\norm{\bfA - \bfX}_F^2
    \]
    with probability at least $1 - \binom{n}{s}^{-2}/200$. By a union bound, this holds for all $S,T\in\binom{[n]}{s}\times\binom{[n]}{s}$ with probability at least $1 - 1/200$. 
    
    To solve the approximate minimization problem supported on $S\times T$, we simply iterate over an $\eps$-net using Corollary \ref{cor:sxs-sparse-rank-k-eps-net}, which has size at most
    \[
        (\poly(n)k/\eps)^{sk}.\qedhere
    \]
\end{proof}

\subsubsection{Polynomial Time Bicriteria Algorithm with Relative Error}

Although the above algorithm runs in exponential time, by our results in Section \ref{sec:computational-complexity}, we cannot hope for polynomial time algorithms for the sparse low rank approximation problem. In order to find polynomial time algorithms, we relax our requirement of outputting an  $s\times s$-sparse rank $k$ matrix and instead allow for larger matrices. 

As a warm-up to our more technically involved additive error bicriteria algorithm of Section \ref{sec:frob-slra-polytime-add-err} and Theorem \ref{thm:frob-slra-polytime-add-err} achieving roughly $\tilde O(sk^2)$ measurements, we prove a relative error algorithm achieving roughly $\tilde O(s^2 k^2)$ measurements. More specifically, we show the following:

\begin{theorem}\label{thm:frob-slra-polytime-rel-err}
    Let $\bfA\in\mathbb R^{n\times d}$. There is a randomized sketching algorithm which makes
    \[
        O\parens*{\frac{s^2k^2}{\eps^4}\log(nd) + \frac{sk}{\eps}\log^2(nd)}
    \]
    measurements to $\bfA$ and outputs a rank $k$ matrix $\bfD$ which is supported on an $O(sk/\eps)\times O(sk/\eps)$ submatrix such that
    \[
        \norm*{\bfA-\bfD}_F^2 \leq (1+\eps)\min_{\bfB\in\mathcal S_{s,k}}\norm*{\bfA-\bfB}_F^2
    \]
    in polynomial time.
\end{theorem}

Throughout this section, let $\bfB$ denote any $s\times s$ sparse rank $k$ matrix. For our relative error polynomial time algorithm, our first observation is that
\[
    \norm*{\bfA - \bfB}_F^2 \geq \norm*{\bfA_{\overline{[s^2k]}}}_F^2
\]
where $\bfA_{\overline{[s^2k]}}$ is the matrix obtained by zeroing out the $s^2 k$ largest entries in absolute value from $\bfA$. We can similarly note that 
\[
    \norm*{\bfA - \bfB}_F^2 \geq \norm*{\bfA_{\overline{[sk]},*}}_F^2
\]
where $\bfA_{\overline{[sk]},*}$ denotes the matrix obtained by zeroing out the $sk$ heaviest rows of $\bfA$, since $\bfB$ is supported on at most $sk$ rows. These tail guarantees are compatible with guarantees achievable efficiently in a stream using \textsf{CountSketch} (see Section \ref{sec:countsketch}). 

Our next observation is to note that a relative error solution is supported on the intersection of the rows and columns of $\bfA$ with squared norm at least $\tau$, for
\[
    \tau = \frac{\eps}{sk}\norm*{\bfA_{\overline{[sk/\eps]},*}}_F^2 \leq \frac{\eps}{sk}\norm*{\bfA - \bfB}_F^2,
\]
since we miss at most $sk$ rows of the optimal solution, and each of these will have squared norm at most $\tau$. This is formalized in the following lemma.

\begin{Lemma}\label{lem:heavy-rows-cols}
    Consider an $s\times s$-sparse rank $k$ matrix
    \[
        \bfB = \sum_{i=1}^k \tau_i \bfx_i \bfy_i^\top
    \]
    that minimizes $\norm*{\bfA - \bfB}_F^2$. Let $\tau>0$ be a threshold parameter. Then, there exists a rank $k$ matrix $\bfD$ supported on the rows and columns of $\bfA$ with norm at least $\tau$ such that
    \[
        \norm*{\bfA - \bfD}_F^2 \leq \norm*{\bfA - \bfB}_F^2 + sk\tau.
    \]
\end{Lemma}
\begin{proof}
Let $W$ be the support of an $sk\times sk$ submatrix that contains $\bfB$, let $X\subseteq W$ be the part of $W$ contained in columns of $\bfA$ with norm at most $\tau$, and let $W\setminus X$ be the part of $W$ contained in columns of $\bfA$ with norm at least $\tau$ (see Figure \ref{fig:matrix-supports}). 

\begin{prooffig}
    \centering
    \captionsetup{type=figure}
    \begin{tikzpicture}
        \begin{scope}

            \draw[step=0.5] (0,0) grid (6,6);
            \coordinate (input);

            \draw[line width=1mm,dashed] (2,2) -- (5,2);
            \draw[line width=1mm,dashed] (2,2) -- (2,5);
            \draw[line width=1mm,dashed] (2,5) -- (5,5);
            \draw[line width=1mm,dashed] (5,2) -- (5,5);

            \draw [fill=gray, draw=black, ultra thick] (2,2) rectangle (3,3);
            \draw [fill=gray, draw=black, ultra thick] (3,2) rectangle (4,3);
            \draw [fill=gray, draw=black, ultra thick] (4,4) rectangle (5,5);
            \draw [fill=gray, draw=black, ultra thick] (3,3.5) rectangle (4,4.5);
            \draw [fill=gray, draw=black, ultra thick] (3.5,3) rectangle (4.5,4);

            \node at (4,3.5) {$\bfB$};

        \end{scope}

        \begin{scope}[xshift=8cm]
            \draw[step=0.5] (0,0) grid (6,6);
            \coordinate (input);

            \draw [fill=gray, draw=none, ultra thick] (2,2) rectangle (5,5);

            \draw[line width=1.5mm,line cap=round] (2,2) -- (5,2);
            \draw[line width=1.5mm,line cap=round] (2,2) -- (2,5);
            \draw[line width=1.5mm,line cap=round] (2,5) -- (5,5);
            \draw[line width=1.5mm,line cap=round] (5,2) -- (5,5);

            \draw[line width=1.5mm,line cap=round] (3,2) -- (3,5);

            \node at (2.5,3.5) {$X$};
            \node at (4,3.5) {$W\setminus X$};
        \end{scope}
    \end{tikzpicture}
    
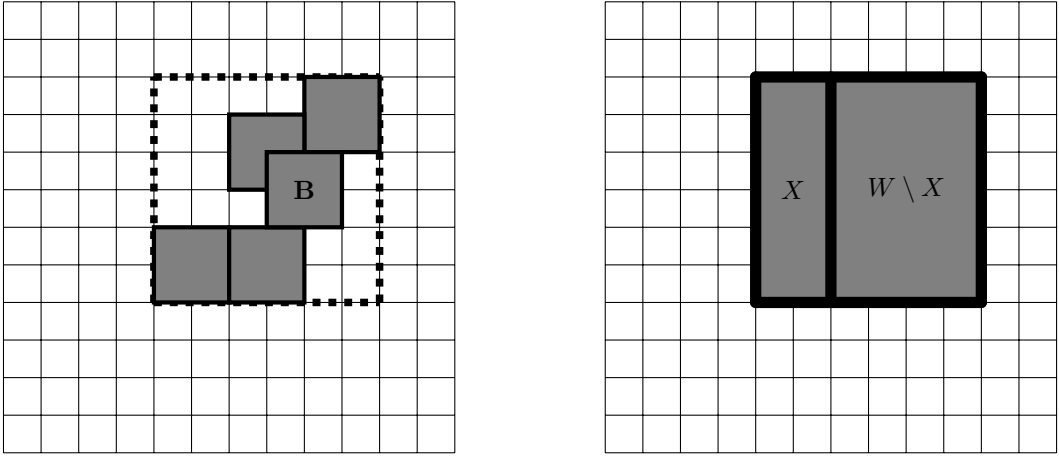
\captionof{figure}{The supports $W, X, W\setminus X$ to be used in the proof.}
    \label{fig:matrix-supports}
\end{prooffig}

Let $\bfA\mid_W$ denote the matrix $\bfA$ restricted to $W$, and similarly define $\bfA\mid_X$ and $\bfA\mid_{W\setminus X}$. Then, by repeatedly using the Pythagorean theorem, we have that
\begin{align*}
    \norm*{\bfA - (\bfA\mid_W)_k}_F^2 &= \norm*{\bfA - \bfA\mid_W}_F^2 + \norm*{\bfA\mid_W - (\bfA\mid_W)_k}_F^2 \\
    &= \norm*{\bfA - \bfA\mid_{W\setminus X}}_F^2 - \norm*{\bfA\mid_X}_F^2 + \norm*{\bfA\mid_W - (\bfA\mid_W)_k}_F^2 \\
    &= \norm*{\bfA - \bfA\mid_{W\setminus X}}_F^2 - \norm*{\bfA\mid_X}_F^2 + \norm*{\bfA\mid_W - (\bfA\mid_W)_k}_F^2 \\
    &\hspace{3em}+\norm*{\bfA\mid_{W\setminus X} - (\bfA\mid_{W\setminus X})_k}_F^2 - \norm*{\bfA\mid_{W\setminus X} - (\bfA\mid_{W\setminus X})_k}_F^2 \\
    &= \norm*{\bfA - (\bfA\mid_{W\setminus X})_k}_F^2 - \norm*{\bfA\mid_X}_F^2 + \norm*{\bfA\mid_W - (\bfA\mid_W)_k}_F^2 \\
    &\hspace{3em}- \norm*{\bfA\mid_{W\setminus X} - (\bfA\mid_{W\setminus X})_k}_F^2.
\end{align*}
Note that if we restrict $(\bfA\mid_W)_k$ to $W\setminus X$, then this is a rank $k$ matrix supported on $W\setminus X$ and thus
\[
    \norm*{\bfA\mid_{W\setminus X} - (\bfA\mid_{W\setminus X})_k}_F^2 \leq \norm*{\bfA\mid_{W\setminus X} - ((\bfA\mid_{W})_k)\mid_{W\setminus X}}_F^2 = \norm*{(\bfA_W - (\bfA\mid_W)_k)\mid_{W\setminus X}}_F^2.
\]
Furthermore, we can remove the restriction to $W\setminus X$ so that the above is bounded above by $\norm*{\bfA_W - (\bfA\mid_W)_k}_F^2$, so
\[
    \norm*{\bfA\mid_W - (\bfA\mid_W)_k}_F^2 - \norm*{\bfA\mid_{W\setminus X} - (\bfA\mid_{W\setminus X})_k}_F^2 \geq 0. 
\]
We thus have that
\begin{align*}
    \norm*{\bfA - (\bfA\mid_W)_k}_F^2 \geq \norm*{\bfA - (\bfA\mid_{W\setminus X})_k}_F^2 - \norm*{\bfA\mid_X}_F^2. 
\end{align*}
Rearranging and noting that
\[
    \norm*{\bfA_X}_F^2 \leq sk\cdot \tau,
\]
we have that
\[
    \norm*{\bfA - (\bfA\mid_{W\setminus X})_k}_F^2 \leq \norm*{\bfA - (\bfA\mid_W)_k}_F^2 + sk\tau.
\]
Finally, since $\bfB$ is a rank $k$ matrix supported on $W$, we have that
\[
    \norm*{\bfA - (\bfA\mid_{W\setminus X})_k}_F^2 \leq \norm*{\bfA - \bfB}_F^2 + sk\tau.\qedhere
\]
\end{proof}

We set the parameter $\tau$ in the above lemma to $\tau = \frac{\eps}{sk}\norm*{\bfA_{\overline{[sk/\eps]},*}}_F^2$, so that there exists a rank $k$ solution $\bfD$ such that
\begin{align*}
    \norm*{\bfA - \bfD}_F^2 &\leq \norm*{\bfA - \bfB}_F^2 + sk\cdot\frac{\eps}{sk}\norm*{\bfA_{\overline{[sk/\eps]},*}}_F^2 \\
    &\leq  \norm*{\bfA - \bfB}_F^2 + \eps\norm*{\bfA_{\overline{[sk/\eps]},*}}_F^2 \\
    &\leq  \norm*{\bfA - \bfB}_F^2 + \eps\norm*{\bfA - \bfB}_F^2 \\
    &= (1+\eps)\norm*{\bfA - \bfB}_F^2
\end{align*}
We can efficiently find such rows and columns using Lemma \ref{lem:cs-row-norm-approximation}:

\begin{Lemma}\label{lem:find-large-norm-row}
    With
    \[
        O\parens*{\frac{sk}{\eps}\log^2(nd)}
    \]
    measurements, one can find a set $S\subseteq[n]$ of rows of size $O(sk/\eps)$ containing all $i\in[n]$ such that
    \[
        \norm*{\bfe_i^\top\bfA}_2^2 \geq \frac{\eps}{sk}\norm*{\bfA_{\overline{[sk/\eps]},*}}_F^2.
    \]
\end{Lemma}
\begin{proof}
We use Lemma \ref{lem:cs-row-norm-approximation} with $\eps$ in the lemma set to $\frac{\eps}{100sk}$ and $\alpha$ set to $1/4$. Suppose row $i\in[n]$ has norm at least
\[
    \norm*{\bfe_i^\top\bfA}_2 \geq \sqrt{\frac{\eps}{sk}}\norm*{\bfA_{\overline{[sk/\eps]},*}}_F.
\]
This requires only
\[
    O\parens*{\frac{sk}{\eps}\log^2(nd)}
\]
measurements to identify. Then,
\begin{align*}
    \norm*{\bfe_i^\top\widehat{\bfA\bfG^\top}}_2 &\geq \norm*{\bfe_i^\top\bfA}_2 - \frac14\norm*{\bfe_i^\top\bfA}_2 - \sqrt{\frac54} \sqrt{\frac{\eps}{100sk}}\norm*{\bfA_{\overline{[sk/\eps]},*}}_F \\
    &\geq \frac5{8}\sqrt{\frac{\eps}{sk}}\norm*{\bfA_{\overline{[sk/\eps]},*}}_F.
\end{align*}
On the other hand, if $\norm*{\bfe_i^\top\widehat{\bfA\bfG^\top}}_2 \geq \frac5{8}\sqrt{\frac{\eps}{sk}}\norm*{\bfA_{\overline{[sk/\eps]},*}}_F$, then
\begin{align*}
    \norm*{\bfe_i^\top\bfA}_2 &\geq \frac5{8}\sqrt{\frac{\eps}{sk}}\norm*{\bfA_{\overline{[sk/\eps]},*}}_F - \frac14\norm*{\bfe_i^\top\bfA}_2 - \sqrt{\frac54} \sqrt{\frac{\eps}{100sk}}\norm*{\bfA_{\overline{[sk/\eps]},*}}_F \\
    \implies \norm*{\bfe_i^\top\bfA}_2&\geq \frac25\sqrt{\frac{\eps}{sk}}\norm*{\bfA_{\overline{[sk/\eps]},*}}_F.
\end{align*}
Thus, by selecting all rows with 
\[
    \norm*{\bfe_i^\top\widehat{\bfA\bfG^\top}}_2 \geq \frac5{8}\sqrt{\frac{\eps}{sk}}\norm*{\bfA_{\overline{[sk/\eps]},*}}_F,
\]
we select all rows with 
\[
    \norm*{\bfe_i^\top\bfA}_2 \geq \sqrt{\frac{\eps}{sk}}\norm*{\bfA_{\overline{[sk/\eps]},*}}_F.
\]
Furthermore, there are at most $\frac{25}{4}\frac{sk}{\eps}$ such rows belonging to $\bfA_{\overline{[sk/\eps]},*}$ and thus we select at most
\[
    \frac{25}{4}\frac{sk}{\eps} + \frac{sk}{\eps} = O\parens*{\frac{sk}{\eps}}
\]
rows of $\bfA$. 
\end{proof}

Finally, we let $q = O(s^2k^2/\eps^4)$ and let $r = \log(nd)$. We then treat $\bfA$ as an $nd$-dimensional vector and use the tail error guarantee of \textsf{CountSketch} (Lemma \ref{lem:cs-tail-error}) in order to reconstruct a matrix $\hat\bfA$ such that
\[
    \norm*{\hat\bfA - \bfA}_\infty^2 \leq \frac1q \norm*{\bfA_{\overline{[q]}}}_F^2 \leq \frac{\eps^4}{s^2 k^2} \norm*{\bfA - \bfB}_F^2.
\]
When restricted to the $O(sk/\eps)\times O(sk/\eps)$ submatrix $S\times T$ of the heavy rows and columns identified before, we have that
\[
    \norm*{\hat\bfA\mid_{S\times T} - \bfA\mid_{S\times T}}_F^2 \leq O\parens*{\frac{sk}{\eps}}^2 \norm*{\hat\bfA - \bfA}_\infty^2 \leq \eps^2 \norm*{\bfA - \bfB}_F^2. 
\]
This requires only
\[
    O\parens*{\frac{s^2k^2}{\eps^4}\log(nd)}
\]
measurements. We can then show that the optimal low rank approximation to $\hat\bfA\mid_{S\times T}$ will give us a relative error sparse low rank approximation to $\bfA$.

\begin{Lemma}\label{lem:hatA-lra-to-A-lra}
    Let $S\subseteq[n]$ and $T\subseteq[n]$ be supersets of the rows and columns of $\bfA$ that have squared norm at least
    \[
        \tau = \frac{\eps}{sk}\norm*{\bfA-\bfB}_F^2
    \]
    and let $\hat\bfA$ be a matrix such that
    \[
        \norm*{\hat\bfA\mid_{S\times T} - \bfA\mid_{S\times T}}_F^2 \leq \eps^2\norm*{\bfA - \bfB}_F^2
    \]
    where for a matrix $\bfM$, $\bfM_{S\times T}$ denotes the restriction of $\bfM$ to the submatrix indexed by rows $S$ and columns $T$. Then,
    \[
        \norm*{\bfA - (\hat\bfA\mid_{S\times T})_k}_F^2 \leq (1+O(\eps))\norm*{\bfA-\bfB}_F^2.
    \]
\end{Lemma}
\begin{proof}
We have that
\begin{align*}
    \norm*{\bfA\mid_{S\times T} - (\hat\bfA\mid_{S\times T})_k}_F &\leq \norm*{\hat\bfA\mid_{S\times T} - (\hat\bfA\mid_{S\times T})_k}_F + \norm*{\hat\bfA\mid_{S\times T} - \bfA\mid_{S\times T}}_F \\
    &\leq \norm*{\hat\bfA\mid_{S\times T} - (\bfA\mid_{S\times T})_k}_F + \norm*{\hat\bfA\mid_{S\times T} - \bfA\mid_{S\times T}}_F \\
    &\leq \norm*{\bfA\mid_{S\times T} - (\bfA\mid_{S\times T})_k}_F + 2\norm*{\hat\bfA\mid_{S\times T} - \bfA\mid_{S\times T}}_F \\
    &\leq \norm*{\bfA\mid_{S\times T} - (\bfA\mid_{S\times T})_k}_F + 2\eps\norm*{\bfA-\bfB}_F
\end{align*}
so by squaring both sides, we get
\begin{align*}
    \norm*{\bfA\mid_{S\times T} - (\hat\bfA\mid_{S\times T})_k}_F^2 &\leq \parens*{\norm*{\bfA\mid_{S\times T} - (\bfA\mid_{S\times T})_k}_F + 2\eps\norm*{\bfA-\bfB}_F}^2 \\
    &= \norm*{\bfA\mid_{S\times T} - (\bfA\mid_{S\times T})_k}_F^2 + 4\eps\norm*{\bfA\mid_{S\times T} - (\bfA\mid_{S\times T})_k}_F\norm*{\bfA-\bfB}_F + 4\eps^2\norm*{\bfA-\bfB}_F^2 \\
    &= \norm*{\bfA\mid_{S\times T} - (\bfA\mid_{S\times T})_k}_F^2 + 4\eps(1+\eps)\norm*{\bfA-\bfB}_F\norm*{\bfA-\bfB}_F + 4\eps^2\norm*{\bfA-\bfB}_F^2 \\
    &\leq \norm*{\bfA\mid_{S\times T} - (\bfA\mid_{S\times T})_k}_F^2 + (4\eps(1+\eps) + 4\eps^2)\norm*{\bfA-\bfB}_F^2 \\
    &= \norm*{\bfA\mid_{S\times T} - (\bfA\mid_{S\times T})_k}_F^2 + (4\eps(1+\eps) + 4\eps^2)\norm*{\bfA-\bfB}_F^2.
\end{align*}
Adding $\norm*{\bfA-\bfA\mid_{S\times T}}_F^2$ on both sides and applying the Pythagorean theorem, we obtain
\[
    \norm*{\bfA - (\hat\bfA\mid_{S\times T})_k}_F^2 \leq \norm*{\bfA - (\bfA\mid_{S\times T})_k}_F^2 + O(\eps)\norm*{\bfA-\bfB}_F^2.
\]
Since $(\bfA\mid_{S\times T})_k$ is the optimal matrix supported on the rows and column at least $\tau$, we have by Lemma \ref{lem:heavy-rows-cols} that
\begin{align*}
    \norm*{\bfA - (\hat\bfA\mid_{S\times T})_k}_F^2 &\leq \norm*{\bfA - \bfB}_F^2 + O(\eps)\norm*{\bfA-\bfB}_F^2 \\
    &= (1+O(\eps))\norm*{\bfA-\bfB}_F^2
\end{align*}
as desired.
\end{proof}

\subsubsection{Polynomial Time Bicriteria Algorithm with Additive Error}\label{sec:frob-slra-polytime-add-err}

The previous polynomial time relative error algorithm has a quadratic dependence on both $s$ and $k$, whereas a linear dependence is possible if we allow for exponential time algorithms. A natural question is whether these can be improved to linear or not. The quadratic dependencies can be attributed to the tail error guarantee we seek from every entry on our $O(sk/\eps)\times O(sk/\eps)$ submatrix that contains an approximately optimal solution, which was used to compute an SVD. If we change our approximate SVD approach to be based on sampling rows of this submatrix, then we would require a tail error guarantee on fewer entries, and thus should lead to better upper bounds. This leads to the idea of implementing the \cite{DBLP:journals/jacm/FriezeKV04} \emph{$\ell_2$ sampling} approach to approximate SVD in this setting. This approximate SVD subroutine gives additive error guarantees rather than the stronger relative error guarantees in the previous section. However, with this relaxation, we are able to improve the quadratic dependence on $s$ to linear. In this section, we prove the following theorem:

\begin{theorem}\label{thm:frob-slra-polytime-add-err}
    Let $\bfA\in\mathbb R^{n\times d}$. There is a randomized sketching algorithm which makes
    \[
        O\parens*{\frac{sk^2}{\eps^6}\log^2 n}
    \]
    measurements to $\bfA$ and outputs a rank $k$ matrix $\bfD$ supported on an $O(sk/\eps)\times O(sk/\eps)$ submatrix such that
    \[
        \norm*{\bfA-\bfD}_F^2 \leq \min_{\bfB\in\mathcal S_{s,k}}\norm*{\bfA-\bfB}_F^2 + \eps\norm*{\bfA}_F^2
    \]
    in polynomial time.
\end{theorem}

\begin{Remark}
    If we optimize over only rank $k$ matrices that are supported on an $s\times s$ submatrix rather than sums of $k$ $s\times s$ matrices, then by making small adjustments to our argument, we can show that we use only
    \[
        O\parens*{\frac{sk}{\eps^6}\log^2 n} = sk\poly(\log n,\eps^{-1})
    \]
    measurements, i.e., linear in both $s$ and $k$, which is the correct dependence on $s$ and $k$. 
\end{Remark}

Our algorithm roughly proceeds in the same way as the previous one: by identifying $O(sk/\eps)$ heavy rows and columns that contain an additive error rank $k$ approximation to $\bfA$, and then performing an approximate SVD on this $O(sk/\eps)\times O(sk/\eps)$ submatrix. The difference is that the approximate SVD will be replaced by the \cite{DBLP:journals/jacm/FriezeKV04} routine which samples rows proportional to their $\ell_2$ norms. 

Note that n\"aively, the \cite{DBLP:journals/jacm/FriezeKV04} algorithm is implemented in two passes over a stream \cite{DBLP:journals/siamcomp/DrineasKM06a}: one pass to compute the $\ell_2$ norms of the rows, and one pass to obtain the sample of the rows. We now show the intuition for conducting this sampling in one pass. The work of \cite{DBLP:conf/stoc/MahabadiRWZ20} recently implemented \emph{subspace sampling}, a seemingly even more sequential sampling scheme, in one pass over a stream, and includes $\ell_2$ sampling as a special case. We show an alternate one pass implementation for $\ell_2$ sampling. 

Our idea is to switch the $\ell_2$ norm computation step and then sampling step. That is, we first obtain a sample of rows each sampled with probability $p$, and then later restrict to the rows of the sample with $\ell_2$ norm $\Theta(p\norm*{\bfA}_F^2)$. By running this procedure in parallel for geometrically decreasing values of $p$, say $p = 1, 1/2, 1/4, \dots, 1/n$, we obtain an \cite{DBLP:journals/jacm/FriezeKV04} sample. The advantage of this approach is that by sampling first, we need to approximate fewer values of the matrix, which improves our upper bounds. 

\paragraph{Approximate \texorpdfstring{$\ell_2$}{l2} Sampling of Rows.}

 By setting 
 \[
     \tau = \frac{\eps}{sk}\norm*{\bfA}_F^2
 \]
 in Lemma \ref{lem:heavy-rows-cols}, we see that an additive error solution is supported in the $O(sk/\eps)\times O(sk/\eps)$ submatrix at the intersection of the rows and columns with $\ell_2$ norm at least $\tau$. Let these rows and columns be indexed by $S\subseteq[n]$ and $T\subseteq[d]$, respectively. Then, Lemma \ref{lem:cs-row-norm-approximation} allows us to find the subset of coordinates $S$ and $T$. Next, we discuss how to approximately find a rank $k$ projection $\hat\bfV\hat\bfV^\top$ that is supported on these coordinates $S$ and $T$. Our strategy is to obtain approximations to the rows of $\bfA$ that are sampled proportionally to their $\ell_2$ norm, formalized in the following definition.

\begin{Definition}[$\ell_2$ sampling \cite{DBLP:journals/jacm/FriezeKV04}]\label{def:l2-sample}
    Let $\bfA\in\mathbb R^{n\times d}$. Then, $P$ is a \emph{$c$-approximate $\ell_2$ sampling distribution} if 
    \[
        P_i \geq \min\braces*{c \frac{\norm*{\bfe_i^\top\bfA}_2^2}{\norm*{\bfA}_F^2}, 1}
    \]
    for each $i\in[n]$. 
\end{Definition}
It is known that a $\Theta(k/\eps^2)$-approximate $\ell_2$ sample of the rows of $\bfA$ yields additive low rank approximations, which we seek \cite[Theorem 2]{DBLP:journals/siamcomp/DrineasKM06a}. 

\begin{theorem}[\textsc{LinearTimeSVD} \cite{DBLP:journals/siamcomp/DrineasKM06a}]
    Let $\bfA\in\mathbb R^{n\times d}$. Then, there is an algorithm that samples rows from a $\Theta(k/\eps^2)$-approximate $\ell_2$ distribution and outputs a matrix $\bfZ\in\mathbb R^{d\times k}$ with orthogonal columns such that
    \[
        \norm*{\bfA - \bfA\bfZ\bfZ^\top}_F^2 \leq \norm*{\bfA - \bfA_k}_F^2 + \eps\norm*{\bfA}_F^2.
    \]
\end{theorem}

Before we proceed with the sampling result, we show several lemmas that allow us to reduce the task of obtaining an additive rank $k$ approximation to $\bfA$ to obtaining an additive rank $k$ approximation to a surrogate matrix $\hat\bfA$ that is close to $\bfA$ in Frobenius norm. 

Our first lemma shows that if $\hat\bfA$ is close to $\bfA$ in Frobenius norm and $\hat\bfA\hat\bfV\bfV^\top$ is a good rank $k$ approximation for $\bfA$, then $\bfA\hat\bfV\bfV^\top$ is also a good rank $k$ approximation for $\bfA$. 

\begin{Lemma}\label{lem:approx-proj}
    Let $\hat\bfV\bfV^\top$ be a rank $k$ projection and let $\hat\bfA$ be such that
    \[
        \norm*{\hat\bfA - \bfA}_F^2 \leq \delta.
    \]
    Then,
    \[
        \norm*{\bfA - \bfA\hat\bfV\hat\bfV^\top}_F^2 \leq \norm*{\bfA - \hat\bfA\hat\bfV\hat\bfV^\top}_F^2 + \delta + 2\sqrt\delta\norm*{\bfA - \hat\bfA\hat\bfV\hat\bfV^\top}_F.
    \]
\end{Lemma}
\begin{proof}
    \begin{align*}
        \norm*{\bfA - \bfA\hat\bfV\hat\bfV^\top}_F^2 &\leq \parens*{\norm*{\bfA - \hat\bfA\hat\bfV\hat\bfV^\top}_F + \norm*{\hat\bfA\hat\bfV\hat\bfV^\top - \bfA\hat\bfV\hat\bfV^\top}_F}^2 \\
        &= \parens*{\norm*{\bfA - \hat\bfA\hat\bfV\hat\bfV^\top}_F + \norm*{(\hat\bfA - \bfA)\hat\bfV\hat\bfV^\top}_F}^2 \\
        &\leq \parens*{\norm*{\bfA - \hat\bfA\hat\bfV\hat\bfV^\top}_F + \norm*{\hat\bfA - \bfA}_F}^2 \\
        &= \norm*{\bfA - \hat\bfA\hat\bfV\hat\bfV^\top}_F^2 + \norm*{\hat\bfA - \bfA}_F^2 + 2\norm*{\bfA - \hat\bfA\hat\bfV\hat\bfV^\top}_F\norm*{\hat\bfA - \bfA}_F \\
        &\leq \norm*{\bfA - \hat\bfA\hat\bfV\hat\bfV^\top}_F^2 + \delta + 2\sqrt\delta\norm*{\bfA - \hat\bfA\hat\bfV\hat\bfV^\top}_F
    \end{align*}
\end{proof}

Our second lemma shows that if $\hat\bfA$ and $\bfA$ are close in Frobenius norm, then a good additive error low rank approximation of $\hat\bfA$ is a good additive error low rank approximation of $\bfA$. 

\begin{Lemma}\label{lem:approx-low-rank}
Let $\bfA\in\mathbb R^{n\times d}$ and let $\hat\bfA$ be such that
\[
    \norm*{\hat\bfA - \bfA}_F^2 \leq \delta.
\]
Let $\bfD$ be a rank $k$ matrix such that
\[
    \norm*{\hat\bfA - \bfD}_F^2 \leq \norm*{\hat\bfA - \hat\bfA_k}_F^2 + \eta.
\]
Then, 
\[
    \norm*{\bfA - \bfD}_F^2 \leq \norm*{\bfA - \bfA_k}_F^2 + 2\sqrt\delta\norm*{\bfA}_F + 2\delta + \eta + 2\sqrt{\delta\parens*{2\delta + 2\norm*{\bfA}_F^2 + \eta}}.
\]
\end{Lemma}
\begin{proof}
    We first estimate
    \begin{align*}
        \norm*{\hat\bfA - \bfA_k}_F^2 &\leq \parens*{\norm*{\hat\bfA - \bfA}_F + \norm*{\bfA - \bfA_k}_F}^2 && \text{by triangle inequality} \\
        &\leq \parens*{\sqrt\delta + \norm*{\bfA - \bfA_k}_F}^2
    \end{align*}
    Then, we can bound
    \begin{align*}
        \norm*{\bfA - \bfD}_F^2 &\leq \parens*{\norm*{\bfA - \hat\bfA}_F + \norm*{\hat\bfA - \bfD}_F}^2 \\
        &\leq \parens*{\norm*{\bfA - \hat\bfA}_F + \sqrt{\norm*{\hat\bfA - \bfA_k}_F^2 + \eta}}^2 \\
        &\leq \parens*{\sqrt\delta + \sqrt{\norm*{\hat\bfA - \bfA_k}_F^2 + \eta}}^2 \\
        &\leq \parens*{\sqrt\delta + \sqrt{\parens*{\sqrt\delta + \norm*{\bfA - \bfA_k}_F}^2 + \eta}}^2 \\
        &= \delta + \parens*{\sqrt\delta + \norm*{\bfA - \bfA_k}_F}^2 + \eta + 2\sqrt{\delta\parens*{\parens*{\sqrt\delta + \norm*{\bfA - \bfA_k}_F}^2 + \eta}} \\
        &\leq \delta + \parens*{\sqrt\delta + \norm*{\bfA - \bfA_k}_F}^2 + \eta + 2\sqrt{\delta\parens*{2\delta + 2\norm*{\bfA - \bfA_k}_F^2 + \eta}} \\
        &\leq \delta + \parens*{\sqrt\delta + \norm*{\bfA - \bfA_k}_F}^2 + \eta + 2\sqrt{\delta\parens*{2\delta + 2\norm*{\bfA}_F^2 + \eta}} \\
        &\leq \norm*{\bfA - \bfA_k}_F^2 + 2\sqrt\delta\norm*{\bfA}_F + 2\delta + \eta + 2\sqrt{\delta\parens*{2\delta + 2\norm*{\bfA}_F^2 + \eta}}. \qedhere
    \end{align*}
\end{proof}

Finally, we show that an $\ell_2$ norm approximation of the rows of $\bfA$ implies a Frobenius norm approximation of $\bfA$. 

\begin{Lemma}\label{lem:approx-row-wise}
    Let $\bfA\in\mathbb R^{n\times d}$ and let $\hat\bfA$ be such that for every $i\in[n]$,
    \[
        \norm*{\bfe_i^\top\hat\bfA - \bfe_i^\top\bfA}_2^2 \leq \eps \norm*{\bfe_i^\top\bfA}_F^2. 
    \]
    Then
    \[
        \norm*{\hat\bfA - \bfA}_F^2 \leq \eps\norm*{\bfA}_F^2
    \]
\end{Lemma}
\begin{proof}
    \[
        \norm*{\hat\bfA - \bfA}_F^2 = \sum_{i=1}^n \norm*{\bfe_i^\top\hat\bfA - \bfe_i^\top\bfA}_2^2 \leq \sum_{i=1}^n\eps \norm*{\bfe_i^\top\bfA}_F^2 \leq \eps \norm*{\bfA}_F^2 \qedhere
    \]
\end{proof}

\paragraph{Outputting a Good Sparse Projection.}

By setting $\tau$ in Lemma \ref{lem:heavy-rows-cols} to $(\eps/sk)\norm*{\bfA}_F^2$, we may identify an  $O(sk/\eps) \times O(sk/\eps)$ submatrix that contains a good additive error $s\times s$-sparse rank $k$ low rank approximation to $\bfA$ using
\[
    O\parens*{\frac{sk}{\eps}\log n}
\]
measurements. Let $S$ and $T$ be the indices of the $O(sk/\eps)$ heavy rows and columns, and let $\bfA_{S\times T}$ denote the $O(sk/\eps)\times O(sk/\eps)$ submatrix indexed by $S$ and $T$. We now wish to perform an SVD on this submatrix. We do this approximately by using the $\ell_2$ sampling machinery developed in \cite{DBLP:journals/jacm/FriezeKV04, DBLP:journals/siamcomp/DrineasKM06a}. 

Note that if $\norm*{\bfA_{S\times T}}_F^2 < \eps \norm*{\bfA}_F^2$, then
\begin{align*}
    \norm*{\bfA}_F^2 &= \norm*{\bfA - \bfA_{S\times T}}_F^2 + \norm*{\bfA_{S\times T}}_F^2 \leq \norm*{\bfA - \bfA_{S\times T}}_F^2 + \eps\norm*{\bfA}_F^2
\end{align*}
so $0$ is already a good additive error low rank approximation. We thus WLOG assume that $\norm*{\bfA_{S\times T}}_F^2 \geq \eps \norm*{\bfA}_F^2$. 

If we conduct the $\ell_2$ sampling process (Definition \ref{def:l2-sample}) with $c = \Theta(k/\eps^2)$, note in particular that we must sample all rows $i\in S$ such that
\[
    \norm*{\bfe_i^\top \bfA_{S\times T}}_2^2 \geq \Omega\parens*{\frac{\eps^2}{k}}\norm*{\bfA_{S\times T}}_F^2
\]
with probability $1$. Let $i\in S$ be such a row. Note then that
\[
    \norm*{\bfe_i^\top \bfA_{S\times T}}_2^2 \geq \Omega\parens*{\frac{\eps^3}{k}}\norm*{\bfA}_F^2
\]
since $\norm*{\bfA_{S\times T}}_F^2 \geq \eps \norm*{\bfA}_F^2$. 

Suppose that we view $\bfA$ as a vector and attempt to reconstruct $\hat\bfA$ using \textsf{CountSketch}. With $O(sk^2/\eps^6)$ \textsf{CountSketch} buckets, we may recover a matrix $\hat\bfA$ such that
\[
    \norm*{\hat\bfA - \bfA}_\infty^2 \leq \frac{\eps^6}{sk^2}\norm*{\bfA}_F^2. 
\]
For row $i$, the error over the $O(sk/\eps)$ coordinates of the row is at most
\[
    \norm*{\bfe_i^\top\hat\bfA_{S\times T} - \bfe_i^\top\bfA_{S\times T}}_2^2 \leq O\parens*{\frac{sk}{\eps}}\cdot \frac{\eps^6}{sk^2} \norm*{\bfA}_F^2 \leq O\parens*{\frac{\eps^5}{k}} \norm*{\bfA}_F^2 \leq \eps^2 \norm*{\bfe_i^\top\bfA_{S\times T}}_2^2.
\]
Similarly, suppose that a row $i\in S$ has squared norm at least
\[
    \norm*{\bfe_i^\top\bfA_{S\times T}}_2^2 \geq \Omega\parens*{\alpha\frac{\eps^2}{k}}\norm*{\bfA_{S\times T}}_F^2 \geq \Omega\parens*{\alpha\frac{\eps^3}{k}}\norm*{\bfA}_F^2,
\]
then we must sample the row with probability at least $\alpha$. Consider sampling each row of $\bfA$ with probability $\alpha$ and let $W$ be this sample and let $\bfP_W$ be the associated sampling matrix. Then the resulting matrix $\bfP_W\bfA$ has an expected squared Frobenius mass of $\E\norm*{\bfP_W\bfA}_F^2 = \alpha\norm*{\bfA}_F^2$, and at most
\[
    \norm*{\bfP_W\bfA}_F^2 \leq O\parens*{\alpha \log\frac{n}{\eps}}\norm*{\bfA}_F^2
\]
with probability at least $1 - O(1/\log(n/\eps))$. Then again treating this subsampled matrix as a vector, we may reconstruct this submatrix using \textsf{CountSketch} with $O(sk^2\log(n/\eps)/\eps^6)$ buckets to recover an approximation $\widehat{\bfP_W\bfA}$ such that
\[
    \norm*{\widehat{\bfP_W\bfA} - \bfP_W\bfA}_\infty^2 \leq \frac{\eps^6}{sk^2\log(n/\eps)} \norm*{\bfP_W\bfA}_F^2 = O(\alpha) \frac{\eps^6}{sk^2} \norm*{\bfA}_F^2. 
\]
Then for this row $i$, if sampled, the error over the $O(sk/\eps)$ coordinates of the row is at most
\[
    \norm*{\bfe_i^\top\widehat{\bfP_W\bfA_{S\times T}} - \bfe_i^\top\bfP_W\bfA_{S\times T}}_2^2 \leq O\parens*{\frac{sk}{\eps}} \cdot O(\alpha)\frac{\eps^6}{sk^2} \norm*{\bfA}_F^2 \leq O\parens*{\alpha\frac{\eps^5}{k}} \leq \eps^2 \norm*{\bfe_i^\top\bfA_{S\times T}}_2^2
\]
with probability at least $1 - O(1/\log(n/\eps))$. By a union bound over all $O(\log(n/\eps))$ sampling levels for $\alpha = 1, 1/2, 1/4, 1/8, \dots, \eps^2/n^2$, this is true for all of these sampling levels with constant probability. 

Thus, we may obtain a $\Theta(k/\eps^2)$-approximate $\ell_2$ sample of some matrix $\hat\bfA_{S\times T}$ that satisfies
\[
    \norm*{\bfe_i^\top \hat\bfA_{S\times T} - \bfe_i^\top \bfA_{S\times T}}_2^2 \leq \eps^2 \norm*{\bfe_i^\top\bfA_{S\times T}}_2^2,
\]
where by Lemma \ref{lem:approx-row-wise}, $\norm*{\hat\bfA_{S\times T} - \bfA_{S\times T}}_F^2 \leq \eps^2\norm*{\bfA_{S\times T}}_F^2$. Now let $\hat\bfC_{S\times T}$ be the rows sampled from $\hat\bfA_{S\times T}$ and let $\hat\bfV$ be the $\Theta(sk/\eps)\times k$ matrix of the top $k$ right singular vectors of $\hat\bfC_{S\times T}$. Then by \cite[Theorem 2]{DBLP:journals/siamcomp/DrineasKM06a}, we have that
\[
    \norm*{\hat\bfA_{S\times T} - \hat\bfA_{S\times T}\hat\bfV\hat\bfV^\top}_F^2 \leq \norm*{\hat\bfA_{S\times T} - (\hat\bfA_{S\times T})_k}_F^2 + \eps\norm*{\bfA_{S\times T}}_F^2
\]
Then by Lemma \ref{lem:approx-low-rank}, 
\begin{align*}
    \norm*{\bfA_{S\times T} - \hat\bfA_{S\times T}\hat\bfV\hat\bfV^\top}_F^2 &\leq \norm*{\bfA_{S\times T} - (\bfA_{S\times T})_k}_F^2 + (3\eps+\eps^2+\eps\sqrt{2+\eps+2\eps^2})\norm*{\bfA_{S\times T}}_F^2 \\
    &\leq \norm*{\bfA_{S\times T} - (\bfA_{S\times T})_k}_F^2 + O(\eps)\norm*{\bfA_{S\times T}}_F^2
\end{align*}
and thus by Lemma \ref{lem:approx-proj},
\[
    \norm*{\bfA_{S\times T} - \bfA_{S\times T}\hat\bfV\hat\bfV^\top}_F^2 \leq \norm*{\bfA_{S\times T} - (\bfA_{S\times T})_k}_F^2 + O(\eps)\norm*{\bfA_{S\times T}}_F^2
\]
as well. 

We now add back the squared Frobenius mass of the entries outside of the $S\times T$ submatrix, i.e., $\norm*{\bfA - \bfS_S^\top \bfA_{S\times T}\bfS_T}_F^2$. Adding $\norm*{\bfA - \bfS_S^\top\bfA_{S\times T}\bfS_T}_F^2$ on both sides, we have that
\[
    \norm*{\bfA - \bfS_S^\top\bfA_{S\times T}\bfS_T}_F^2 + \norm*{\bfA_{S\times T} - \bfA_{S\times T}\hat\bfV\hat\bfV^\top}_F^2 = \norm*{\bfA - \bfS_S^\top\bfA_{S\times T}\hat\bfV\hat\bfV^\top\bfS_T}_F^2
\]
and
\[
    \norm*{\bfA - \bfS_S^\top\bfA_{S\times T}\bfS_T}_F^2 + \norm*{\bfA_{S\times T} - (\bfA_{S\times T})_k}_F^2 = \norm*{\bfA - \bfS_S^\top(\bfA_{S\times T})_k\bfS_T}_F^2 
\]
which shows that
\[
    \norm*{\bfA - \bfS_S^\top\bfA\hat\bfV\hat\bfV^\top\bfS_T}_F^2 \leq \norm*{\bfA - \bfS_S^\top(\bfA_{S\times T})_k\bfS_T}_F^2 + O(\eps)\norm*{\bfA_{S\times T}}_F^2.
\]
Note then that
\[
    \norm*{\bfA - \bfS_S^\top(\bfA_{S\times T})_k\bfS_T}_F^2 \leq \norm*{\bfA - \hat\bfB}_F^2
\]
where $\hat\bfB$ keeps only the columns of an optimal $s\times s$-sparse rank $k$ $\bfB$ with $\ell_2$ norm at least $(\eps^2/sk)\norm*{\bfA}_F^2$, and
\[
    \norm*{\bfA_{S\times T}}_F^2 \leq \norm*{\bfA}_F^2,
\]
so by combining the above bounds,
\begin{align*}
    \norm*{\bfA - \bfS_S^\top\bfA_{S\times T}\hat\bfV\hat\bfV^\top\bfS_T}_F^2  &\leq \norm*{\bfA - \bfS_S^\top(\bfA_{S\times T})_k\bfS_T}_F^2 + O(\eps)\norm*{\bfA_{S\times T}}_F^2 \\
    &\leq \norm*{\bfA - \bfS_S^\top(\bfA_{S\times T})_k\bfS_T}_F^2 + O(\eps)\norm*{\bfA}_F^2 \\
    &\leq \norm*{\bfA - \hat\bfB}_F^2 + O(\eps)\norm*{\bfA}_F^2 \\
    &\leq \min_{\text{$s\times s$-sparse rank $k$ $\bfC$}}\norm*{\bfA - \bfC}_F^2 + O(\eps)\norm*{\bfA}_F^2.
\end{align*}
By rescaling $\eps$ by a constant factor, we conclude that
\[
    \norm*{\bfA - \bfS_S^\top\bfA_{S\times T}\hat\bfV\hat\bfV^\top\bfS_T}_F^2 \leq \min_{\text{$s\times s$-sparse rank $k$ $\bfC$}}\norm*{\bfA - \bfC}_F^2 + \eps\norm*{\bfA}_F^2
\]
as desired.

\paragraph{Outputting a Good Factorization.}

Given the above result, we have coordinate projections down to a submatrix supported on $S\times T$, with $S$ and $T$ each of size $O(sk/\eps)$, and a rank $k$ right projection for this submatrix. Our final task is to obtain an actual $O(sk/\eps)\times O(sk/\eps)$ matrix, rather than just a right factor. 

Recall the current form of our sparse low rank approximation:
\[
    \norm*{\bfA - \bfS_S^\top\bfA_{S\times T}\hat\bfV\hat\bfV^\top\bfS_T}_F^2 = \norm*{\bfA - \bfS_S^\top\bfS_S\bfA\bfP_T^\top\hat\bfV\hat\bfV^\top\bfS_T}_F^2. 
\]
Note that here, we have $S, T$, and $\hat\bfV$, but we don't have all of the entries of $\bfA_{S\times T}$. We first use approximate matrix product (Lemma \ref{lem:cs-apm}) to approximate $\bfA\bfS_T^\top\hat\bfV\hat\bfV^\top$ by $\bfA\bfR^\top\bfR\bfS_T^\top\hat\bfV\hat\bfV^\top$, where $\bfR$ is an $O(k/\eps^2)\times n$ \textsf{CountSketch} matrix. By applying Lemma \ref{lem:cs-apm} with $\eps$ in the lemma set to $\eps/\sqrt k$, we have that
\[
    \norm*{\bfA\bfS_T^\top\hat\bfV\hat\bfV^\top - \bfA\bfR^\top\bfR\bfS_T^\top\hat\bfV\hat\bfV^\top}_F^2 \leq \frac{\eps^2}{k} \norm*{\bfA}_F^2 \norm*{\bfS_T^\top\hat\bfV\hat\bfV^\top}_F^2 = \frac{\eps^2}{k} \cdot \norm*{\bfA}_F^2 \cdot k = \eps^2\norm*{\bfA}_F^2.
\]
Next, suppose that $\bfT^{(1)}, \bfT^{(2)}, \dots, \bfT^{(r)}$ are $(sk/\eps^3)\times n$ \textsf{CountSketch} matrices for $r = O(\log(nsk/\eps))$. Then, by applying row-wise approximation (Lemma \ref{lem:cs-row-wise-approximation}) with $\eps$ in the lemma set to $\eps^3/sk$ and $d$ to $O(sk/\eps)$, if we maintain $\bfT^{(j)}\bfA\bfR^\top$ for $j\in[r]$, from which we can compute
\[
    \bfT^{(j)}\bfA\bfR^\top\bfR\bfS_T^\top\hat\bfV\hat\bfV^\top,
\]
for every $i\in[n]$, we can compute an approximate row $\bfw^{(i)}$ such that
\[
    \norm*{\bfe_i^\top\bfA\bfR^\top\bfR\bfS_T^\top\hat\bfV\hat\bfV^\top - \bfw^{(i)}}_2^2 \leq \frac{\eps^3}{sk}\norm*{\bfA\bfR^\top\bfR\bfS_T^\top\hat\bfV\hat\bfV^\top}_F^2.
\]
Note that
\begin{align*}
    \norm*{\bfA\bfS_T^\top\bfR^\top\bfR\hat\bfV\hat\bfV^\top}_F &\leq \norm*{\bfA\bfS_T^\top\hat\bfV\hat\bfV^\top}_F + \norm*{\bfA\bfS_T^\top\hat\bfV\hat\bfV^\top - \bfA\bfS_T^\top\bfR^\top\bfR\hat\bfV\hat\bfV^\top}_F \\
    &\leq \norm*{\bfA}_F + \eps\norm*{\bfA}_F \\
    &= (1+\eps)\norm*{\bfA}_F
\end{align*}
so
\[
    \frac{\eps^3}{sk}\norm*{\bfA\bfS_T^\top\bfR^\top\bfR\hat\bfV\hat\bfV^\top}_F^2 \leq O\parens*{\frac{\eps^3}{sk}}\norm*{\bfA}_F^2.
\]
Thus, summing over the $O(sk/\eps)$ rows in $S$, and letting $\bfW$ be the matrix corresponding to the approximate rows $\bfw^{(i)}$ for $i\in S$,
\[
    \norm*{\bfS_S\bfA\bfS_T^\top\bfR^\top\bfR\hat\bfV\hat\bfV^\top - \bfW}_F^2 \leq O\parens*{\frac{sk}{\eps}} \frac{\eps^3}{sk}\norm*{\bfA}_F^2 = \eps^2\norm*{\bfA}_F^2. 
\]
Then by a couple of applications of the triangle inequality,
\begin{align*}
    \norm*{\bfA - \bfS_S^\top\bfW\bfS_T^\top}_F &\leq \norm*{\bfA - \bfS_S^\top\bfS_S\bfA\bfS_T^\top\bfR^\top\bfR\hat\bfV\hat\bfV^\top\bfS_T}_F + \eps\norm*{\bfA}_F \\
    &\leq \norm*{\bfA - \bfS_S^\top\bfS_S\bfA\bfS_T^\top\hat\bfV\hat\bfV^\top\bfS_T}_F + 2\eps\norm*{\bfA}_F
\end{align*}
so
\[
    \norm*{\bfA - \bfS_S^\top\bfW\bfS_T}_F^2 \leq \min_{\text{$s\times s$-sparse rank $k$ $\bfC$}}\norm*{\bfA - \bfC}_F^2 + O(\eps)\norm*{\bfA}_F^2
\]
as desired.

To implement the above, we must maintain $\bfT^{(j)}\bfA\bfR^\top$ for $j\in[r]$, which requires
\[
    O\parens*{\frac{sk}{\eps^3}}\cdot O\parens*{\frac{k}{\eps^2}}\cdot \log\frac{nsk}{\eps} = O\parens*{\frac{sk^2}{\eps^5}\log\frac{nsk}{\eps}}
\]
which is within our measurement budget. 

\subsection{Computational Complexity}\label{sec:computational-complexity}

In this section, we show that under a randomized version of the Exponential Time Hypothesis, the $(1+\eps)$-approximate Frobenius sparse low rank approximation problem has no algorithm running in time
\[
    f(s)\parens*{\frac1{\sqrt\eps}}^{o(\sqrt s)}.
\]
Related works have shown NP-completeness and computational hardness results for related low rank approximation variants, including sparse PCA \cite{DBLP:conf/icml/MoghaddamWA06, DBLP:journals/ipl/Magdon-Ismail17, DBLP:conf/colt/ChanPR16}, low rank approximation with weights or missing data \cite{DBLP:journals/siammax/GillisG11,DBLP:conf/stoc/RazenshteynSW16} and $\ell_p$ low rank approximation for $p\in(1,2)$ \cite{DBLP:conf/soda/BanBBKLW19}. The NP-completeness result of \cite{DBLP:journals/siammax/GillisG11} employs a reduction from the maximum-edge biclique problem to show hardness for the rank $1$ weighted low rank approximation problem, which is similar to our reduction from the $s$-biclique problem to show exponential running time in $s$ under a randomized variant of the Exponential Time Hypothesis. Related hardness results were shown by \cite[Theorem 1.4]{DBLP:conf/stoc/RazenshteynSW16}, also by a reduction to the maximum biclique problem. 

\subsubsection{Dependence on \texorpdfstring{$s$}{s} and \texorpdfstring{$\eps$}{eps}}

We first note that the sparse low rank approximation problem is NP-complete for $\eps = 1/n^2$. This is obtained through a reduction from the $s$-\textsc{Biclique} problem, which asks whether a given bipartite graph contains an  $s$-biclique $K_{s,s}$ or not:

\begin{Lemma}\label{lem:exact-slra-np-complete}
    Let $G = (L,R,E)$ be a bipartite graph with $\abs{L} = \abs{R} = n$, and let $\bfA_G$ denote the $n\times n$ biadjacency matrix for $G$, that is, $(\bfA_G)_{i,j}$ is $1$ if vertex $i\in L$ and $j\in R$ are adjacent and $0$ otherwise. Then for $\eps = 1/n^2$, it is NP-hard to compute the cost $\norm{\bfA - \bfA^*}_F^2$ of the optimal solution $\bfA^*$ for the $s\times s$-sparse low rank approximation problem within a $(1+\eps)$ factor.
\end{Lemma}
\begin{proof}
If $G$ contains a $K_{s,s}$, then $\bfA_G$ contains an $s\times s$ all ones matrix, so the optimal cost is $\norm{\bfA - \bfA^*}_F^2 = \norm{\bfA}_F^2 - s^2$. On the other hand, if $G$ contains no $K_{s,s}$, then every $s\times s$ submatrix of $\bfA_G$ must have at most $s^2 - 1$ ones. Thus, the optimal cost is at least $\norm{\bfA - \bfA^*}_F^2 \geq \norm{\bfA}_F^2 - s^2 + 1$. In any case, we have that $\norm{\bfA - \bfA^*}_F^2 \leq \norm{\bfA}_F^2 \leq n^2$ so the costs in these two cases are separated by at least a $1+\eps$ factor for $\eps = 1/n^2$. 
\end{proof}

This is similar to the reduction of sparse PCA to the $s$-\textsc{Clique} problem \cite{DBLP:journals/ipl/Magdon-Ismail17}. The NP-completeness of $s$-biclique $K_{s,s}$ is stated without proof in \cite{DBLP:books/fm/GareyJ79} and proven in e.g.,  \cite{DBLP:journals/jal/AlonDLRY94}.

In fact, it is also known that $s$-\textsc{Biclique} does not admit fixed parameter tractable (FPT) algorithms, assuming a randomized version of the Exponential Time Hypothesis. 
\begin{theorem}[Corollary 1.8, \cite{DBLP:journals/jacm/Lin18}]
Under the randomized ETH, there is no $f(s)\cdot n^{o(\sqrt s)}$-time algorithm to decide whether a given graph contains a subgraph isomorphic to $K_{s,s}$. 
\end{theorem}

Then by a similar reduction as before, under the randomized ETH, the $(1+n^{-2})$-approximate $s\times s$-sparse low rank approximation problem admits no $f(s)\cdot n^{o(\sqrt s)}$-time algorithm as well. Stated in terms of the accuracy parameter $\eps$ and $s$, we have the following:
\begin{corollary}
Under the randomized ETH, there is no $f(s)\cdot (1/\sqrt\eps)^{o(\sqrt s)}$-time algorithm to compute a $(1+\eps)$-approximate $s\times s$-sparse low rank approximation. 
\end{corollary}

\begin{Remark}
Recently, the running time lower bound for the $s$-\textsc{Biclique} problem has been shown to admit no algorithms running in time $f(s) n^{o(s)}$ under a variant of the Planted Clique conjecture, termed the \emph{Strongish Planted Clique Hypothesis} \cite{DBLP:conf/innovations/ManurangsiRS21}. Then using this, one concludes that there is no algorithm running in time $f(s)(1/\sqrt\eps)^{o(s)}$ for the $(1+\eps)$-approximate Frobenius sparse low rank approximation problem as well, under the Strongish Planted Clique Conjecture. 
\end{Remark}
\section{Gaussian Noise Spectral Sparse Low Rank Approximation}\label{sec:gaussian-noise}
We finally consider the setting where we are given a matrix 
\[
    \bfA = \lambda \bfX + \bfG
\]
for $\bfG\sim \mathcal N(0,1)^{n\times n}$, $\lambda = O(\sqrt n)$, and $\bfX\in\mathcal O_{s,k}$ has operator norm $1$. Note that we consider disjoint $s\times s$-sparse rank $k$ matrices $\mathcal O_{s,k}$ rather than $\mathcal S_{s,k}$ in order to obtain tight bounds for this problem. Furthermore, throughout this section, we will assume that the coefficients $\tau_i$ in the Definition \ref{def:sxs-sparse-rank-k-matrix} are bounded by $\tau_i\leq\poly(n)$, for convenience for certain net arguments.

In this setting, we consider both the problem of detection of the low rank signal $\bfX$ as well as the estimation of the low rank signal. 

\subsection{Lemmas for Analysis}

We first introduce simple lemmas which will be useful for our analysis. The first argues that orthogonal sparse components can only add Frobenius mass to each other.

\begin{Lemma}\label{lem:heavy-submatrix}
    Let $\bfu_1,\bfv_1$ and $\bfu_2,\bfv_2$ be two $s$-sparse left and right singular vector pairs. Then, the Frobenius norm of the restriction of $\bfu_1 \bfv_1^\top + \bfu_2 \bfv_2^\top$ to the $s\times s$ support of $\bfu_1 \bfv_1^\top$ is at least $\norm{\bfu_1 \bfv_1^\top}_F$. 
\end{Lemma}
\begin{proof}
    Let $S_\bfu$ be intersection of the supports of $\bfu_1$ and $\bfu_2$. Note then that $\angle*{\bfu_1\vert_{S_u}, \bfu_2\vert_{S_u}} = \angle*{\bfu_1,\bfu_2} = 0$. Similarly, $\angle*{\bfv_1\vert_{S_\bfv}, \bfv_2\vert_{S_\bfv}} = 0$ where $S_\bfv$ is the intersection of the supports of $\bfv_1$ and $\bfv_2$. Then by the Pythagorean theorem, the Frobenius norm of the sum $\bfu_1\bfv_1^\top + \bfu_2\bfv_2^\top$ restricted to $S_\bfu\times S_\bfv$ is
    \begin{align*}
        \norm*{\bfu_1\vert_{S_\bfu}\bfv_1\vert_{S_\bfv}^\top + \bfu_2\vert_{S_\bfu}\bfv_2\vert_{S_\bfv}^\top}_F^2 &= \norm*{\bfu_1\vert_{S_\bfu}\bfv_1\vert_{S_\bfv}^\top}_F^2 + \norm*{\bfu_2\vert_{S_\bfu}\bfv_2\vert_{S_\bfv}^\top}_F^2 + 2\angle*{\bfu_1\vert_{S_\bfu}\bfv_1\vert_{S_\bfv}^\top, \bfu_2\vert_{S_\bfu}\bfv_2\vert_{S_\bfv}^\top} \\
        &= \norm*{\bfu_1\vert_{S_\bfu}\bfv_1\vert_{S_\bfv}^\top}_F^2 + \norm*{\bfu_2\vert_{S_\bfu}\bfv_2\vert_{S_\bfv}^\top}_F^2 + 2\angle*{\bfu_1\vert_{S_\bfu}, \bfu_2\vert_{S_\bfu}}\angle*{\bfv_1\vert_{S_\bfv}, \bfv_2\vert_{S_\bfv}} \\
        &= \norm*{\bfu_1\vert_{S_\bfu}\bfv_1\vert_{S_\bfv}^\top}_F^2 + \norm*{\bfu_2\vert_{S_\bfu}\bfv_2\vert_{S_\bfv}^\top}_F^2 \geq \norm*{\bfu_1\vert_{S_\bfu}\bfv_1\vert_{S_\bfv}^\top}_F^2
    \end{align*}
    By adding in the rest of the entries of $\bfu_1\bfv_1^\top$ outside of $S_\bfu\times S_\bfv$, we conclude as desired. 
\end{proof}

Our second lemma shows that any unit vector is essentially an $s$-sparse vector with entries of squared value roughly $1/s$, for some $s$. 

\begin{Lemma}\label{lem:flat-sparsity}
    Let $\bfv\in\mathbb R^m$. Then, there exists an $s\in[m]$ such that there are at least $s/2$ entries of $\bfA$ with squared value at least $\norm*{\bfv}_2^2(s(1+\log_2 m))^{-1}$. 
\end{Lemma}
\begin{proof}
    By scaling, assume WLOG that $\norm*{\bfv}_2^2 = 1$. If $\bfv$ has an entry with squared value at least $(1+\log_2 m)^{-1}$, then we are already done with $s = 1$. Thus, let $\abs{\bfv_i} < (1+\log_2 m)^{-1}$ for all $i\in[m]$. 

    Suppose for contradiction that for every $\ell\in[\log_2 m]$, there are fewer than $2^\ell/2$ entries of $\bfv$ with absolute value at least $(2^\ell\log_2 m)^{-1}$. We then partition the entries of $\bfA$ into the multiset $S_0$ of entries with squared value at most $1/m(1+\log_2 m)$ and the multisets $S_\ell$ with squared value in
    \[
        \left[ \frac1{2^\ell(1+\log_2 m)}, \frac2{2^\ell(1+\log_2 m)} \right)
    \]
    for each $\ell\in[\log_2 m]$. Note then that 
    \begin{align*}
        \norm{\bfv}_2^2 &< \abs{S_0}\cdot\frac1{m(1+\log_2 m)} + \sum_{\ell=1}^{\log_2 m}\abs{S_\ell}\cdot \frac{2}{2^\ell(1+\log_2 m)} \\
        &\leq \frac{m}{m(1+\log_2 m)} + \sum_{\ell = 1}^{\log_2 m} \frac{2^\ell}{2} \frac{2}{2^\ell(1+\log_2 m)} \\
        &\leq \frac1{1+\log_2 m} + \sum_{\ell = 1}^{\log_2 m} \frac1{1+\log_2 m} \\
        &= \frac{1+\log_2 m}{1+\log_2 m} = 1,
    \end{align*}
    a contradiction.
\end{proof}

\subsection{Detection}

We first consider the problem of detection. 

\subsubsection{Lower Bounds}

As previously noted, \cite{DBLP:journals/siamcomp/LiNW19} give sketching lower bounds for the detection problem when $\bfX$ is a random $n\times n$ rank $1$ matrix. We first adapt this lower bound to the setting where $\bfX\in\mathcal O_{s,k}$. For our sparse low rank signal lower bound, we draw our low rank signal matrix as a random sparse vector. 

\begin{Definition}[Random sparse vector]\label{def:random-sparse-vector}
We define the distribution $\nu_{s,n}$ which samples a vector $\bfu\in\mathbb R^n$ by sampling a subset $S\subseteq[n]$ of $\abs{S} = s$ random coordinates, drawing the values for $\bfu$ on the coordinates of $S$ as $\bfu\vert_S \sim \mathcal N(0,\bfI_s)$, and $0$s everywhere else. 
\end{Definition}


Now consider the problem of distinguishing between two distributions where $\mathcal D_1 = \mathcal N(0,1)^{n\times n}$ is an $n\times n$ matrix with independent standard Gaussian entries, and $\mathcal D_2$ is the distribution drawn as $\sum_{j=1}^k \bfw_j \bfu^j (\bfv^j)^\top$ where $\bfu^j\sim\nu_{n,s_1}$ and $\bfv^j\sim\nu_{d,s_2}$, are all drawn independently (recall Definition \ref{def:random-sparse-vector}). We take $m$ linear measurements and denote the corresponding measurements as $\bfL^1, \bfL^2, \dots, \bfL^m$. We WLOG assume that $\norm*{\bfL^i}_F^2 = 1$ and $\tr(\bfL^i (\bfL^j)^\top) = 0$ for $i\neq j$. Let $\mathcal L_1$ and $\mathcal L_2$ denote the distribution 
of the $m$ dimensional linear sketches of $\mathcal D_1$ and $\mathcal D_2$, respectively. 

\paragraph{Useful Computations.}

We will need the following simple lemma. 
    
\begin{Lemma}\label{lem:frobenius-random-submatrix}
Let $\bfA\in\mathbb R^{n\times d}$. Then,
\[
    \E_{S\sim\binom{[n]}{s}}\norm{\bfS_S \bfA}_F^2 = \sum_{S\in\binom{[n]}{s}} \frac1{\binom{n}{s}}\sum_{i\in S}\norm{\bfe_i^\top \bfA}_2^2 = \sum_{i\in[n]}\frac{\binom{n-1}{s-1}}{\binom{n}{s}}\norm{\bfe_i^\top \bfA}_2^2 = \frac{s}{n}\norm{\bfA}_F^2
\]
\end{Lemma}

We also need the following computation for later use. This slightly modifies \cite[Lemma 3.1]{DBLP:journals/siamcomp/LiNW19}. 

\begin{Lemma}
    Let $\bfA\in\mathbb R^{n\times d}$ with $\norm*{\bfA}_F < 1$. Let $\bfu'\sim\nu_{s_1,n}$ and $\bfv'\sim\nu_{s_2,d}$ be drawn independently. Then,
    \[
        \E_{\bfu',\bfv'}\exp(\bfu'^\top\bfA\bfv') \leq 1 + \frac{s_1 s_2}{nd}\norm*{\bfA}_F^2.
    \]
\end{Lemma}
\begin{proof}
    It is shown in the proof of \cite[Lemma 3.1]{DBLP:journals/siamcomp/LiNW19} that for a matrix $\bfM\in\mathbb R^{s_1 \times s_2}$ and independent Gaussians $\bfx\sim\mathcal N(0,\bfI_{s_1})$ and $\bfy\sim\mathcal N(0,\bfI_{s_2})$ that
    \[
        \E_{\bfx,\bfy} \exp(\bfx^\top\bfM\bfy^\top) = \prod_i \frac1{\sqrt{1-\sigma_i^2(\bfM)}}.
    \]
    Applying this identity, we obtain

    \begin{align*}
        \E_{\substack{\bfu'\sim\nu_{s_1,n} \\ \bfv'\sim\nu_{s_2,d}}} \exp\parens*{\bfu'^\top \bfA\bfv'} &= \E_{S\sim\binom{[n]}{s_1},T\sim\binom{[d]}{s_2}} \E_{\bfx\sim \mathcal N(0,\bfI_{s_1}), \bfy\sim\mathcal N(0,\bfI_{s_2})}\exp\parens*{\bfx^\top (\bfS_S \bfA \bfS_T^\top) \bfy} \\
        &= \E_{S\sim\binom{[n]}{s_1},T\sim\binom{[d]}{s_2}} \prod_{i} \frac1{\sqrt{1-\sigma_i^2(\bfS_S \bfA \bfS_T^\top)}}
    \end{align*}
    We then continue to bound as
    \begin{align*}
        \E_{S\sim\binom{[n]}{s_1},T\sim\binom{[d]}{s_2}} \prod_{i} \frac1{\sqrt{1-\sigma_i^2(\bfS_S \bfA \bfS_T^\top)}} &\leq \E_{S\sim\binom{[n]}{s_1},T\sim\binom{[d]}{s_2}} \frac1{\sqrt{1-\sum_{i} \sigma_i^2(\bfS_S \bfA \bfS_T^\top)}} \\
        &= \E_{S\sim\binom{[n]}{s_1},T\sim\binom{[d]}{s_2}} \frac1{\sqrt{1-\norm{\bfS_S \bfA \bfS_T^\top}_F^2}} \\
        &\leq \E_{S\sim\binom{[n]}{s_1},T\sim\binom{[d]}{s_2}} 1 + \norm{\bfS_S \bfA \bfS_T^\top}_F^2 && \text{$1/\sqrt{1-t}\leq 1+t$ on $[0,1/2]$}\\
        &= 1 + \E_{S\sim\binom{[n]}{s_1}}\E_{T\sim\binom{[d]}{s_2}}\norm{\bfS_S \bfA \bfS_T^\top}_F^2 \\
        &= 1 + \frac{s_1 s_2}{nd}\norm{\bfA}_F^2 && \text{Lemma \ref{lem:frobenius-random-submatrix}}
    \end{align*}
    which is the desired conclusion.
\end{proof}

\paragraph{Main Lower Bound.}

\begin{Lemma}\label{lem:gaussian-detection-lb}
We consider sketching dimension $m$ with normalization vector $\bfw\in\mathbb R^k$. Let $s_1$ and $s_2$ be sparsity parameters and let $n,d$ be the dimensions of the Gaussian matrix. Further suppose
\[
    m\frac{s_1 s_2}{n d}\norm{\bfw}_2^4\leq c
\]
for a small enough constant $c$. Then,
\[
    \norm{\mathcal L_1 - \mathcal L_2}_\TV\leq \frac1{10}.
\]
\end{Lemma}

\begin{proof}
Note that $\mathcal L_1 = \mathcal N(0,\bfI_m)$ and $\mathcal L_2 = \mathcal N(0,\bfI_m) * \mu$ where $\mu$ is the distribution of
\[
    \begin{pmatrix}
        \sum_{j=1}^k \bfw_j (\bfu^j)^\top \bfL^1 \bfv^j \\
        \sum_{j=1}^k \bfw_j (\bfu^j)^\top \bfL^2 \bfv^j \\
        \vdots \\
        \sum_{j=1}^k \bfw_j (\bfu^j)^\top \bfL^m \bfv^j \\
    \end{pmatrix} \in \mathbb R^m
\]
Define
\[
    \xi \coloneqq \sum_{i=1}^m \parens*{\sum_{j=1}^k  \bfw_j(\bfu^j)^\top\bfL^i \bfv^j}^2
\]
Note that for a single $j\in[k]$, $\E (\bfu^j)^\top\bfL^i\bfv^j = 0$ so 
\begin{align*}
    \E\xi &= \sum_{i=1}^m \E\parens*{\sum_{j=1}^k  \bfw_j(\bfu^j)^\top\bfL^i \bfv^j}^2 \\
    &= \sum_{i=1}^m \sum_{j=1}^k \E\parens*{  \bfw_j(\bfu^j)^\top\bfL^i \bfv^j}^2 \\
    &= \sum_{i=1}^m \sum_{j=1}^k \bfw_j^2  \sum_{a=1}^n \sum_{b=1}^d (\bfL^i)_{a,b}^2 \E(\bfu^j)_a^2 \E(\bfv^j)_b^2 \\
    &= \frac{ms_1s_2}{nd}\norm{\bfw}_2^2
\end{align*}

We now define the event 
\[
    \mathcal E \coloneqq \braces*{\norm{\bfw}_2^2\xi < 1/2}.
\]
The expected value of $\norm{\bfw}_2^2\xi$ is
\[
    \frac{ms_1s_2}{nd}\norm{\bfw}_2^4 \leq c
\]
so by Markov's inequality, $\Pr(\mathcal E) \geq 1 - 2c$. Restrict $\mu$ to this event and let $\tilde\mu$ denote the resulting distribution. Let $\tilde{\mathcal L_2} = \mathcal N(0,\bfI_m) * \tilde\mu$. Then using Propositions 2.1 and 2.2 of \cite{DBLP:journals/siamcomp/LiNW19},
\begin{align*}
    \norm*{\mathcal L_1 - \mathcal L_2}_\TV &\leq \norm*{\mathcal L_1 - \tilde{\mathcal L_2}}_\TV + \norm*{\tilde{\mathcal L_2}- \mathcal L_2}_\TV \\
    &\leq \sqrt{\E_{\bfz_1, \bfz_2\sim\tilde\mu}\exp(\angle*{\bfz_1,\bfz_2}) - 1} + \norm*{\tilde\mu - \mu}_\TV \\
    &\leq \sqrt{\frac1{\Pr(\mathcal E)}\E_{\bfz_1\sim\tilde\mu, \bfz_2\sim\mu}\exp(\angle*{\bfz_1,\bfz_2}) - 1} + \Pr(\mathcal E^c).
\end{align*}
Then,
\begin{align*}
    \E_{\bfz_1\sim\tilde\mu, \bfz_2\sim\mu}\exp(\angle*{\bfz_1,\bfz_2}) &= \E\exp\parens*{\sum_{i=1}^m \sum_{j=1}^k\sum_{j'=1}^k  \bfw_j(\bfu^j)^\top\bfL^i \bfv^j \cdot \bfw_{j'}(\bfu^{j'})^\top\bfL^i \bfv^{j'}} \\
    &= \E\exp\parens*{\sum_{j'=1}^k(\bfu^{j'})^\top\bracks*{\sum_{i=1}^m \sum_{j=1}^k  \bfw_j(\bfu^j)^\top\bfL^i \bfv^j \cdot \bfw_{j'}\bfL^i} \bfv^{j'}} \\
    &= \E_{\bfu^1,\dots,\bfu^r,\bfv^1,\dots,\bfv^r\mid \mathcal E}\prod_{j'=1}^k \E_{\substack{\bfu^{j'}\sim \nu_{k,m} \\ \bfv^{j'}\sim \nu_{k,n}}}\exp\parens*{(\bfu^{j'})^\top \bfQ^{j'}\bfv^{j'}}
\end{align*}
where $\bfQ^{j'}$ is an $n\times d$ matrix defined by
\[
    \bfQ^{j'} \coloneqq \sum_{i=1}^m \sum_{j=1}^k  \bfw_j(\bfu^j)^\top\bfL^i \bfv^j \cdot \bfw_{j'}\bfL^i
\]
We have that
\begin{align*}
    \norm{\bfQ^{j'}}_F^2 &= \norm*{\sum_{i=1}^m \sum_{j=1}^k  \bfw_j(\bfu^j)^\top\bfL^i \bfv^j \cdot \bfw_{j'}\bfL^i}_F^2 \\
    &= \sum_{i=1}^m \norm*{\sum_{j=1}^k  \bfw_j(\bfu^j)^\top\bfL^i \bfv^j \cdot \bfw_{j'}\bfL^i}_F^2 && \text{since the $\bfL^i$ are orthogonal} \\
    &= \sum_{i=1}^m \parens*{\sum_{j=1}^k  \bfw_j(\bfu^j)^\top\bfL^i \bfv^j \cdot \bfw_{j'}}^2 \norm*{\bfL^i}_F^2 \\
    &= \bfw_{j'}^2 \sum_{i=1}^m \parens*{\sum_{j=1}^k  \bfw_j(\bfu^j)^\top\bfL^i \bfv^j}^2 \\
    &= \bfw_{j'}^2 \xi < 1
\end{align*}
since we have conditioned on $\mathcal E$. Thus,
\begin{align*}
    \E_{\bfu^1,\dots,\bfu^k,\bfv^1,\dots,\bfv^k\mid \mathcal E}\prod_{j'=1}^k \E_{\bfu^{j'},\bfv^{j'}}\exp\parens*{(\bfu^{j'})^\top \bfQ^{j'}\bfv^{j'}} &\leq \E_{\bfu^1,\dots,\bfu^k,\bfv^1,\dots,\bfv^k\mid \mathcal E}\prod_{j'=1}^k\parens*{1 + \frac{s_1s_2}{nd}\bfw_{j'}^2\xi} \\
    &\leq \E_{\bfu^1,\dots,\bfu^k,\bfv^1,\dots,\bfv^k\mid \mathcal E}\exp\parens*{\sum_{j'=1}^k \frac{s_1s_2}{nd}\bfw_{j'}^2\xi} \\
    &= \E_{\bfu^1,\dots,\bfu^k,\bfv^1,\dots,\bfv^k\mid \mathcal E}\exp\parens*{\frac{s_1s_2}{nd}\norm{\bfw}_2^2\xi} \\
    &\leq \E_{\bfu^1,\dots,\bfu^k,\bfv^1,\dots,\bfv^k\mid \mathcal E}1 + 2\frac{s_1s_2}{nd}\norm{\bfw}_2^2\xi \\
    &\leq 1 + 2c
\end{align*}
Thus, 
\[
    \norm{\mathcal L_1 - \mathcal L_2}_\TV \leq \sqrt{\frac{2c}{1-2c}} + 2c \leq \frac1{10}
\]
when $c$ is small enough. 
\end{proof}

\begin{corollary}\label{cor:gaussian-detection-lb}
Let $k = 1$, $d = n$, and
\[
    w = \frac{\sqrt n}{\sqrt{s_1 s_2}}. 
\]
Then, a randomized sketching algorithm distinguishing between $\mathcal D_1$ and $\mathcal D_2$ must make at least $m = \Omega(s_1 s_2)$ measurements. 
\end{corollary}
\begin{proof}
    By Lemma \ref{lem:gaussian-detection-lb}, we must have that
    \[
        m\frac{s_1s_2}{n^2}w^4 = m \frac{1}{s_1s_2} \geq c
    \]
    for a constant $c$. Thus, $m \geq c s_1 s_2 = \Omega(s_1 s_2)$, as desired. 
\end{proof}

Now note that if we set $s_1 = s$ and $s_2 = sk$ in the above theorem, this yields a disjoint $s\times s$-sparse rank $k$ matrix. To detect this matrix, our theorem states that $\Omega(s^2 k)$ measurements are required. 

\begin{corollary}\label{cor:gaussian-detection-large-sparsity-lb}
    Any randomized sketching algorithm distinguishing between $\bfA = \bfG$ and $\bfA = \lambda \bfX + \bfG$ for $\bfG\sim\mathcal N(0,1)^{n\times n}$, $\lambda = \sqrt n$, and $\bfX\in\mathcal O_{s,k}$ with $\norm*{\bfX}_2 = 1$ requires $\Omega(s^2 k)$ measurements. 
\end{corollary}

Despite this natural lower bound, there is in fact a better lower bound when $s^2 k$ is small. Consider the task of distinguishing between $\bfA = \bfG$ and $\bfA = \bfG + \sqrt n \bfe_i \bfe_j^\top$ for two randomly chosen coordinates $i\sim [n]$ and $j\sim [n]$. This is the hard instance in \cite{DBLP:conf/icalp/AndoniNPW13} for estimating the $p = 4$ norm of an $n^2$ dimensional vector in the context of frequency moment estimation. For this problem, they show that 
\[
    \Omega((n^2)^{1-2/p}\log(n^2)) = \Omega(n\log n)
\]
measurements are required. Putting these together, the lower bound we obtain is
\[
    \Omega(n\log n + s^2 k).
\]

\subsubsection{Upper Bounds}

We have two upper bounds, each matching the two terms, the $n\log n$ and the $s^2 k$, in the lower bound. 

\paragraph{Small $s$ Regime: $s\leq \sqrt{n/k\log n}$.}

We first show that when 
\[
    s^2 k\log (s^2 k) \leq n
\]
or equivalently when $s \leq \sqrt{n / k\log n}$, that is, when the first term in the lower bound is roughly tight, then an algorithm based on $4$-norm estimation achieves a bound of $n\cdot\poly\log(n)$ measurements. In this setting, we essentially need none of the structure provided by the fact that $\bfX\in\mathcal O_{s,k}$. In fact, all we use is that $\norm*{\bfX}_F^2\geq 1$ and that $\bfX$ is supported on $s^2 k$ coordinates. 

For intuition, let $b\leq n$ and suppose that $\bfX$ is supported on $b$ entries with each nonzero entry equal to $1/\sqrt b$, so the corresponding entry in $\bfA$ is $\sqrt{n/b}$. Now consider a random sample $S$ of $n^2/b^2$ elements of $\bfA$, whose $4$-norm we compute in roughly $\sqrt{n^2/b^2} = n/b$ measurements. If $S$ contains no entries of $\bfX$, then $\norm{S}_4^4 = \Theta(n^2/b^2)$, while if $S$ does contain an entry of $\bfX$, then $4$-norm is a constant times larger, so we can distinguish between the two cases. If we now repeat this $b$ times, for a total of $b\cdot (n/b) = n$ measurements, then we sample a total of $b\cdot (n^2/b^2) = n^2/b$ elements, so one of these trials will find one of the $b$ elements of $\bfX$. 

As described earlier, when $s \leq \sqrt{n/k\log n}$, then we may make use of $4$-norm estimation algorithms to detect the presence of the signal.

\begin{theorem}[Sketching frequency moments \cite{DBLP:journals/corr/abs-1011-2571, DBLP:conf/focs/AndoniKO11, DBLP:conf/icalp/Ganguly15, DBLP:conf/icalp/GangulyW18}]
    \label{thm:freq-moments}
For $p>2$, the $p$-norm of an $n$-dimensional vector can be approximated up to constant factors with probability $1-\delta$ with sketching complexity $\Theta(n^{1-2/p}\log(1/\delta) + n^{1-2/p}\log^{2/p}(1/\delta)\log n)$. 
\end{theorem}

\begin{algorithm}[H]
    \caption{Detection algorithm for $s\leq \sqrt{n/k\log n}$}
    {\bf Input}: $\bfA\in\mathbb R^{n\times n}$, sparsity $s$, rank $k$ \\
    {\bf Output}: True if $\bfA = \bfG + \sqrt n \bfX$ and false otherwise
    \begin{algorithmic}[1]
    \FOR{$i=0$ to $i=\ceil{\log_2(s^2 k)}$}
        \STATE{$s'\coloneqq 2^i$, $\alpha = O((s'^2\log(s^2 k))^{-1})$, $m\coloneqq \Theta(n^2\alpha^2)$, $\tau\coloneqq \Theta(m)$}
        \FOR{$j=1$ to $j=O(s'^2\log^2(s^2 k))$}
            \STATE{Sample a set $\mathcal S$ of $m$ entries of $\bfA$}
            \STATE{Obtain an estimate $y$ for $\norm{\mathcal S}_4^4$ with $O(\sqrt m(\log s^2 k)(\log n))$ measurements (Theorem \ref{thm:freq-moments})}
            \IF{$y\geq\tau$}
                \RETURN{True}
            \ENDIF
        \ENDFOR{}
    \ENDFOR{}
    \RETURN{False}
    \end{algorithmic}
    \label{alg:detection-k<=sqrt-n}
\end{algorithm}

\begin{theorem}\label{thm:small-s-upper-bound}
    Let $\bfX\in\mathcal S_{s,k}$. Then Algorithm \ref{alg:detection-k<=sqrt-n} makes $O(n(\log^3(sk))(\log n))$ sketching measurements and distinguishes between the cases of $\bfA = \bfG + \sqrt n \bfX$ and $\bfA = \bfG$ with constant probability, where $\bfG\sim\mathcal N(0,1)^{n\times n}$. 
\end{theorem}
\begin{proof}
We show in general that when $\bfX$ has Frobenius norm $\norm*{\bfX}_F^2 \geq r$ for some $r$ and is supported on $b$ entries, then the sketching complexity is
\[
    O\parens*{\frac{n}{r}(\log^3 b)(\log\log b)(\log n)}.
\]
The result then follows since
\[
    \norm*{\bfX}_F^2 \geq \norm*{\bfX}_2^2 = 1
\]
and $\bfX$ is supported on $s^2 k$ entries. In this case, we assume the parameter regime
\[
    \frac{nr}{b\log b} \geq c
\]
for some sufficiently large constant $c$. 

By applying Lemma \ref{lem:flat-sparsity} to $\bfX$ viewed as a vector in $b$ dimensions, there exists some $s'\in[b]$ such that there are at least $s'/2$ entries of $\bfX$ with squared value at least
\[
    \alpha \coloneqq r(s'(1+\log_2 b))^{-1}.
\]
This translates to an entry of squared value at least $n\alpha$ in $\bfA$ after scaling $\bfX$ by $\sqrt n$. Now consider a random sample of $m = O(n^2\alpha^2)$ entries of $\bfA$. Note that $m\geq 1$ by assumption on our parameter regime. The $4$-norm of this sample can be approximated up to constant factors with probability at least $1-O(1/b)$ with $O(\sqrt m (\log b) (\log n))$ measurements by Theorem \ref{thm:freq-moments}.

If this sample contains only Gaussian entries, then the $4$-norm has an expectation of $O(m^{1/4})$ \cite[Proposition 2.4]{paouris2017random} and concentrates to this expectation up to a $(1+\eps)$ factor with probability at least $1 - \exp(-\Omega(\eps\sqrt m))$ \cite[Theorem 1.1]{paouris2017random}. On the other hand, if this sample contains an entry of squared size $n\alpha$, then this single entry increases the $4$-norm of the sample by a constant factor, so this distinguishes between the two cases. The probability that this random sample contains one of these entries of squared size $n\alpha$ is 
\[
    p = O\parens*{\frac{ms'}{n^2}} = O(\alpha^2 s') = O\parens*{\frac{r^2}{s'\log^2 b}}.
\]
Thus, with $O(\log(1/\delta)/p)$ repetitions for $\delta = 1/\log_2 b$, the probability that we find none of the $s'$ large entries is at most
\[
    (1-p)^{\log(1/\delta)/p} \leq \exp\parens*{-p\frac{\log\frac1\delta}{p}} = \delta
\]
we find one of the $s'$ entries of squared size at least $n\alpha$ with constant probability. 

The total number of sketching measurements assuming that we know $s'$ was
\begin{align*}
    O\parens*{\sqrt m (\log b)(\log n)\cdot \frac{\log\frac1\delta}p} &= O\parens*{n\alpha (\log b)(\log n)\cdot \frac{\log\frac1\delta}{\alpha^2s'}} \\
    &= O\parens*{\frac{n}{r}(\log^2 b)(\log n)\log\frac1\delta}
\end{align*}
By guessing $s'$ up to constant factors by $s' = 1, 2, 4, \dots, b$ in $\log_2 b$ guesses, our final sketching complexity is
\[
    O\parens*{\frac{n}{r}(\log^3 b)(\log\log b)(\log n)}.\qedhere
\]
\end{proof}

Thus, in the regime where $s \leq \sqrt{n / k\log n}$, the detection problem can roughly be viewed as a simpler case of the $4$-norm estimation problem, since both the algorithms and lower bounds are captured by $4$-norm estimation. 

\paragraph{Large $s$ Regime: $s\geq \sqrt{n/k\log n}$.}

When $s\geq \sqrt{n/k \log n}$, the $4$-norm estimation approach cannot work, and our approach shifts to sampling a submatrix and iterating over a net to identify the signal. We illustrate the basic idea of the algorithm on the hard instance of the lower bound in Corollary \ref{cor:gaussian-detection-large-sparsity-lb}, where the signal matrix is $\bfX$ is essentially $\bfu\bfv^\top$ for $\bfu\sim\mathcal N(0,\bfI_s)$ and $\bfu\sim\mathcal N(0,\bfI_{sk})$, normalized appropriately, so that the signal matrix added to the Gaussian noise has operator norm $\sqrt n$. To detect this signal, suppose we sample an $sk\times s$ submatrix. Then, in expectation, we sample a
\[
    \frac{s^2 k}{n} \times \frac{s^2 k}{n}
\]
submatrix $\bfX'$ of $\bfX$. We then iterate over a net $\mathcal N$ for $(s^2 k/n)\times (s^2 k/n)$-sparse rank $1$ matrices, which has log size roughly $\tilde O(s^2 k/n)$, to try to find a net vector that aligns with $\bfX'$. If we find such a net vector, then the inner product between the net vector and $\bfX'$ will be 
\[
    \norm*{\bfX'}_F = \sqrt{\frac{s^2 k}{n} \cdot \frac{s^2 k}{n} \cdot \frac{n}{s^2 k}} = \sqrt{\frac{s^2 k}{n}}.
\]
On the other hand, note that the inner product between the net vector and the Gaussian noise will just be a standard Gaussian, since the net vector has Frobenius norm $1$. Now for any $\delta\in(0,1)$, the Gaussian will be at most $\sqrt{\log\frac1\delta}$ with probability at least $1-\delta$, so with $\delta = 1/\Theta(\abs*{\mathcal N})$, we can union bound over the net to see that the Gaussian noise is at most 
\[
    \sqrt{\log\frac1\delta} = \tilde O\parens*{\sqrt{\frac{s^2 k}{n}}}. 
\]
We thus see that by adjusting constants and log factors appropriately, we can distinguish between the Gaussian noise and the signal and thus solve the detection problem. 

In general, this approach will work to give a sketching upper bound of $\tilde O(s^2 k)$ dimensions, if we can find an $a\times b$ submatrix with operator norm at least $\tilde \Omega(1)$ for $ab\leq \tilde O(s^2 k)$. However, we are unable to show that such a submatrix exists for every $\bfX\in\mathcal O_{s,k}$. Instead, we obtain a looser bound of
\[
    \tilde O(s^2 k^{4/3})
\]
measurements by showing two algorithms that roughly follow the above approach, one for when the Frobenius norm of $\bfX$ is large and one for when the Frobenius norm of $\bfX$ is small:

\begin{Lemma}\label{lem:gaussian-detection-large-frob-norm}
    Let $\bfX\in\mathcal O_{s,k}$ be an $n\times n$ matrix with Frobenius norm at least
    \[
        \norm*{\bfX}_F^2 \geq r.
    \]
    Then, we can distinguish $\bfG + \sqrt n \bfX$ and $\bfG$ for $\bfG\sim\mathcal N(0,1)^{n\times n}$ with
    \[
        O\parens*{\frac{s^2k^2}{r^2} (\log^6 s)(\log^2 n)}
    \]
    measurements. 
\end{Lemma}
\begin{proof}
    Write
    \[
        \bfX = \sum_{i=1}^k \tau_i \bfx_i \bfy_i^\top. 
    \]
    Note then that
    \[
        \norm*{\bfX}_F^2 = \sum_{i=1}^k \norm*{\tau_i \bfx_i \bfy_i^\top}_F^2 = \sum_{i=1}^k \tau_i^2.
    \]
    Then there is some $i\in[k]$ such that $\tau_i^2 \geq r/k$. 

    We now apply Lemma \ref{lem:flat-sparsity} to each of the vectors $\bfx_i$ and $\bfy_i$, so that $\bfx_i$ has at least $s_1/2$ entries, say $S_1\subseteq[n]$ of $\bfx_i$ have squared size at least $(s_1(1+\log_2 s))^{-1}$, and similarly, $\bfy_i$ has at least $s_2/2$ entries, say $S_2\subseteq[n]$, of $\bfy_i$ have squared size at least $(s_2(1+\log_2 s))^{-1}$. We refer to these entries as the \emph{flat entries}. 

    \paragraph{Handling Sparse Components.}

    Suppose that
    \[
        s_1 s_2 \leq O\parens*{\frac{nr}{k(1+\log_2 s)^2\log n}} = \tilde O\parens*{\frac{nr}{k}}.
    \]
    Then, there is a submatrix supported on $s_1 s_2$ entries with squared Frobenius norm at least
    \[
        \frac{n\tau_i^2}{s_1 s_2 (1+\log_2 s)^2} \geq \frac{nr}{k s_1 s_2 (1+\log_2 s)^2} \geq 1,
    \]
    which has a sparsity of $b = s_1 s_2$ and Frobenius norm $r \geq 1$ such that
    \[
        \frac{nr}{b\log b} \geq c
    \]
    so we can detect this submatrix by using the $4$-norm estimation procedure of Theorem \ref{thm:small-s-upper-bound}. We thus assume that
    \begin{equation}\label{eq:dense-component}
        s_1s_2 > \Omega\parens*{\frac{nr}{k(1+\log_2 s)^2\log n}} = \tilde\Omega\parens*{\frac{nr}{k}}.
    \end{equation}

    \paragraph{Handling Dense Components.}

    When $s_1 s_2 = \tilde \Omega(nr/k)$ as in Equation (\ref{eq:dense-component}), then our algorithm is to sample a submatrix. Suppose for the moment that we know $s_1$ and $s_2$, and we will remove this assumption later. We then sample a random $n_1\times n_2$ submatrix $\bfA'$ of $\bfA$, for $n_1$ and $n_2$ to be determined later. Then by Chernoff bounds, we sample $t_1 \coloneqq \Theta(n_1 s_1/n)$ rows and $t_2 \coloneqq \Theta(n_2 s_2/n)$ columns of the component $\tau_i \bfx_i \bfy_i$ in expectation, and with high probability $1 - O(1/\log^2 s)$ as long as
    \[
        \Theta\parens*{\frac{n_1 s_1}{n}} = \Theta\parens*{\frac{n_2 s_2}{n}} \geq \log \log s.
    \]
    Denote this flat sampled submatrix of $\bfX$, supported on the flat entries $S_1\times S_2$, as $\bfX'$. We now consider an $\eps$-net $\mathcal N$ (in the Frobenius norm) over rank $1$ $n_1\times n_2$ matrices with $t_1\times t_2$-sparse components and operator norm $1$. This has size
    \[
        \binom{n_1}{t_1}\parens*{\frac6\eps}^{t_1} \cdot \binom{n_2}{t_2}\parens*{\frac6\eps}^{t_2} \leq \parens*{\frac{6en_1}{\eps t_1}}^{t_1}\parens*{\frac{6en_2}{\eps t_2}}^{t_2}
    \]
    by a straightforward modification of the proof of Corollary \ref{cor:sxs-sparse-rank-k-eps-net}. If $\bfA\sim\mathcal N(0,1)^{n\times n}$, then $\bfA'$ is just a $n_1\times n_2$ Gaussian matrix so for all $\bfB\in \mathcal N$,
    \[
        \angle*{\bfA', \bfB} \sim \mathcal N(0, \norm{\bfB}_F^2) = \mathcal N(0, 1)
    \]
    On the other hand, if $\bfA$ has the rank $1$ signal, then for some $\bfB\in \mathcal N$, $\bfB$ is $\eps$-close to $\bfX'$ in Frobenius norm and thus 
    \[
        \angle*{\bfA', \bfB} \sim \mathcal N(\angle*{\sqrt n\bfX', \bfB}, \norm{\bfB}_F^2) = \mathcal N(\angle*{\sqrt n\bfX', \bfB}, 1)
    \]
    where
    \[
        \angle*{\sqrt n \bfX', \bfB} = \Theta(\norm{\sqrt n\bfX'}_F) = \Theta\parens*{\sqrt{ t_1 t_2 \frac{n\tau_i^2}{s_1 s_2(1+\log_2 s)^2}}} \geq \Omega\parens*{\sqrt{\frac{n_1 n_2 r}{kn\log^2 s}}}.
    \]
    Then the two Gaussians means are separated by $\Theta(\sqrt n\norm{\bfX'}_F)$, while by a union bound over the net $\mathcal N$, the Gaussian noise can be as large as
    \begin{align*}
        \Theta(\sqrt{\log\abs{\mathcal N}}) &= \Theta\parens*{\sqrt{t_1\log\frac{6en_1}{\eps t_1} + t_2\log \frac{6en_2}{\eps t_2}}} \\
        &\leq O\parens*{\frac{\sqrt{n_1s_1 + n_2s_2}}{\sqrt n}\sqrt{\log n}}.
    \end{align*}
    One can check that this is at most $\Theta(\sqrt n\norm{\bfX'}_F)$ if
    \[
        n_1 = O\parens*{\frac{k}{r}s_2(\log^2 s)(\log n)}, \qquad n_2 = O\parens*{\frac{k}{r}s_1(\log^2 s)(\log n)}.
    \]
    Thus, the total sketching complexity so far is
    \[
        O\parens*{\frac{k^2}{r^2}s_1 s_2 (\log^4 s)(\log^2 n)}
    \]
    Now to remove our assumption that we know $s_1$ and $s_2$, we must iterate over $O(\log^2 s)$ guesses for $s_1 = 1, 2, 4, \dots, s$ and $s_2 = 1, 2, 4, \dots, s$, which brings the total sketching complexity up to
    \[
        O\parens*{\frac{s^2k^2}{r^2} (\log^6 s)(\log^2 n)}.\qedhere
    \]
\end{proof}

\begin{Lemma}\label{lem:gaussian-detection-small-frob-norm}
    Let $\bfX\in\mathcal O_{s,k}$ be an $n\times n$ matrix with Frobenius norm at most
    \[
        \norm*{\bfX}_F^2 \leq r.
    \]
    Then, we can distinguish $\bfG + \sqrt n \bfX$ and $\bfG$ for $\bfG\sim\mathcal N(0,1)^{n\times n}$ with
    \[
        O(s^2 kr\log^6(sk)\log^2 n).
    \]
    measurements. 
\end{Lemma}
\begin{proof}
    The proof proceeds in two stages. We first prove the existence of a small rank $1$ approximation to $\bfX$, and then iterate over a net to find it. 

    \paragraph{Existence of a Sparse Rank 1 Approximation.}

    Let $\bfu$ and $\bfv$ be the top left and right singular vectors of $\bfX$. Write the vectors $\bfu$ and $\bfv$ in level sets by
    \[
        \bfu = \sum_{i=1}^{O(\log(sk))} \bfu_i, \qquad \bfv = \sum_{i=1}^{O(\log(sk))} \bfv_i.
    \]
    Then,
    \[
        1 = \abs*{\bfu^\top\bfX\bfv} \leq \sum_{i=1}^{O(\log n)}\sum_{j=1}^{O(\log n)}\abs*{\bfu_i^\top \bfX\bfv_j}
    \]
    and thus by averaging, there exist $i,j$ such that
    \[
        \frac1{\Theta(\log^2(sk))} \leq \abs*{\bfu_i^\top \bfX\bfv_j}
    \]
    where $\bfu_i$ is $s_1$-sparse and $\bfv_j$ is $s_2$-sparse. Let $\bfX\mid_{\bfu,\bfv}$ denote the restriction of $\bfX$ to the supports of $\bfu$ and $\bfv$. Then,
    \[
        \abs*{\bfu_i^\top \bfX\bfv_j} \leq \frac{\norm*{\bfX\mid_{\bfu,\bfv}}_1}{\sqrt{ab}} \leq \frac{\sqrt{\nnz(\bfX\mid_{\bfu,\bfv})}\norm*{\bfX\mid_{\bfu,\bfv}}_2}{\sqrt{ab}} \leq \frac{\sqrt{s^2 k}\sqrt r}{\sqrt{ab}}.
    \]
    Combining the previous bounds,
    \[
        s_1 s_2 \leq \Theta(s^2 k r\log^2(sk)).
    \]
    
    \paragraph{Net Argument.}

    By the previous argument, there exist vectors $\bfu$ and $\bfv$ of norm at most $1$ that are $s_1$-sparse and $s_2$-sparse for $s_1,s_2\in[sk]$, respectively, such that
    \[
        \abs*{\angle*{\bfu\bfv^\top, \bfX}} = \abs*{\bfu^\top\bfX\bfv} \geq \frac1{\Theta(\log^2(sk))}
    \]
    and
    \[
        s_1 s_2 \leq \Theta(s^2 k r\log^2(sk)).
    \]
    We then proceed as in the subsampling and net iteration step of Lemma \ref{lem:gaussian-detection-large-frob-norm}. Since the calculations are essentially the same, we only quickly outline the proof. We first assume we know $s_1$ and $s_2$ and we sample a $n_1\times n_2$ submatrix. This samples a $t_1 \times t_2$ submatrix of $\bfX\mid_{\bfu,\bfv}$ where
    \[
        t_1 = \Theta\parens*{\frac{n_1 s_1}{n}}, \qquad t_2 = \Theta\parens*{\frac{n_2 s_2}{n}}.
    \]
    We then iterate over a net $\mathcal N$ over $t_1\times t_2$ rank $1$ matrices of log size at most
    \[
        \log \abs*{\mathcal N} \leq (t_1 + t_2)\log n = \frac{n_1 s_2 + n_2 s_2}{n}\log n.
    \]
    Let $\bfB\in\mathcal N$. If there is no signal, then 
    \[
        \angle*{\bfB, \bfA} \sim \mathcal N(0,\norm*{\bfB}_F^2) = \mathcal N(0,1)
    \]
    while if there is a signal, then for some $\bfB\in\mathcal N$, 
    \[
        \angle*{\bfB,\bfA} \sim \mathcal N(\angle*{\sqrt n\bfX\mid_{\bfu,\bfv},\bfB},\norm*{\bfB}_F^2) = \mathcal N(\angle*{\sqrt n\bfX\mid_{\bfu,\bfv},\bfB},1)
    \]
    where
    \[
        \angle*{\sqrt n\bfX\mid_{\bfu,\bfv},\bfB} = \Theta\parens*{\sqrt{t_1 t_2\frac{n}{s_1 s_2\log^2(sk)}}} = \Theta\parens*{\sqrt{\frac{n_1 n_2}{n\log^2(sk)}}}.
    \]
    This is larger than the noise from the Gaussian over $\mathcal N$ as long as
    \[
        \sqrt{\log \abs*{\mathcal N}} = \sqrt{\frac{n_1 s_2 + n_2 s_2}{n}\log n} = O\parens*{\sqrt{\frac{n_1 n_2}{n\log^2(sk)}}}
    \]
    which happens for
    \[
        n_1 = O(s_2\log^2(sk)\log n), \qquad n_2 = O(s_1\log^2(sk)\log n).
    \]
    Thus, the required number of samples in this case is
    \[
        O(s_1 s_2\log^4(sk)\log^2 n) = O(s^2 kr\log^4(sk)\log^2 n)
    \]
    Iterating over $\log^2(sk)$ choices of $s_1$ and $s_2$, the total number of measurements is
    \[
        O(s^2 kr\log^6(sk)\log^2 n). \qedhere
    \]
\end{proof}

By combining Lemmas \ref{lem:gaussian-detection-large-frob-norm} and \ref{lem:gaussian-detection-small-frob-norm} and balancing parameters, we obtain the following:

\begin{theorem}\label{thm:gaussian-detection-ub-high-sparsity}
    Let $\bfX\in\mathcal O_{s,k}$ be an $n\times n$ matrix. Then, we can distinguish $\bfG + \sqrt n \bfX$ and $\bfG$ for $\bfG\sim\mathcal N(0,1)^{n\times n}$ with
    \[
        O\parens*{s^2 k^{4/3} (\log^6 sk)(\log^2 n)}
    \]
    measurements. 
\end{theorem}
\begin{proof}
    If $\norm*{\bfX}_F^2 \geq k^{1/3}$, then we use Lemma \ref{lem:gaussian-detection-large-frob-norm} to obtain a bound of
    \[
        O\parens*{\frac{s^2 k^2}{k^{2/3}} (\log^6 s)(\log^2 n)} \leq O\parens*{s^2 k^{4/3} (\log^6(sk))(\log^2 n)}.
    \]
    Otherwise, when $\norm*{\bfX}_F^2 < k^{1/3}$, we use Lemma \ref{lem:gaussian-detection-small-frob-norm} to obtain a bound of 
    \[
        O\parens*{s^2 k^{4/3} (\log^6(sk))(\log^2 n)}. \qedhere
    \]
\end{proof}

\subsection{Estimation}

We now consider the estimation problem, in which, given a matrix $\bfA = \bfG + \sqrt n\bfX$, we must output an $s\times s$-sparse matrix rank $k$ $\bfX'\in\mathcal S_{s,k}$ such that 
\[
    \bfX = \norm{\bfX - \bfX'}_2\leq \eps.
\]
For this problem, for constant $\eps$, we show that iteration over a net gives an algorithm that makes $O(nsk\log(nk))$ measurements in Theorem \ref{thm:gaussian-noise-est-ub}, and this is approximately optimal by a Gaussian channel information capacity argument in Theorem \ref{thm:gaussian-noise-est-lb} by showing a lower bound of $\Omega(nsk\log(n/sk))$ measurements.

\subsubsection{Gaussian Noise Estimation Upper Bound}

We will need an analogue of the approximate matrix product lemma (Lemma \ref{lem:cs-apm}) for matrices with i.i.d.\ Gaussian entries from \cite[Lemma 6]{DBLP:conf/focs/Sarlos06}.

\begin{Lemma}[Gaussian approximate matrix product \cite{DBLP:conf/focs/Sarlos06}]\label{lem:apm-gaussian}
    Let $\bfA\in\mathbb R^{n\times d}$, $\bfB\in\mathbb R^{m\times d}$, and $\eps>0$. Consider an $r\times d$ matrix $\bfS$ of all i.i.d.\ Gaussian entries for $r = \Omega(\eps^{-2}(\log(n+m))(\log\frac1\delta))$. Then,
    \[
        \Pr\braces*{\norm*{\bfA\bfS^\top\bfS\bfB^\top - \bfA\bfB^\top}_F^2 \leq \eps^2 \norm*{\bfA}_F^2 \norm*{\bfB}_F^2} \geq 1-\delta.
    \]
\end{Lemma}

Given this lemma, our algorithm is simply to approximately maximize the inner product between $\bfA$ and net vectors given by Corollary \ref{cor:sxs-sparse-rank-k-eps-net}. 

\begin{theorem}\label{thm:gaussian-noise-est-ub}
Let $\eps, s, n$ be such that 
\[
    \sqrt{2s\log\frac{3en}{\eps s}} \leq \eps^2\sqrt n. 
\]
Let $\bfH$ be an $m\times n^2$ Gaussian matrix for 
\[
    m = O\parens*{\frac{ns}{\eps^4}\log\frac{n}{\eps s}}.
\]
Consider $\bfA = \bfG + \sqrt n\bfX$ for a $s\times s$ sparse rank $1$ matrix $\bfX$ with operator norm $1$ and $\bfG\sim\mathcal N(0,1)^{n\times n}$. Then, there is an algorithm which, given $m$ Gaussian measurements of $\bfA$, outputs an $s\times s$ sparse rank $1$ matrix $\bfX'$ such that
\[
    \norm*{\bfX - \bfX'}_2 \leq \eps.
\]
\end{theorem}
\begin{proof}
    We iterate over an $\eps$-net $\mathcal N$ for $\mathcal S_{s,k}$ of size at most $(\poly(n)k/\eps)^{sk}$, which exists by Corollary \ref{cor:sxs-sparse-rank-k-eps-net}. Our approach is to select the best candidate in the net by finding the $\bfX'\in N$ that approximately maximizes the inner product with $\bfX$. Indeed, if we find a $\bfX'$ such that
    \[
        \angle*{\bfX,\bfX'} \geq 1 - \eps^2/2
    \]
    then
    \begin{align*}
        \norm*{\bfX - \bfX'}_2^2 &\leq \norm*{\bfX - \bfX'}_F^2 \\
        &= \norm*{\bfX}_F^2 + \norm*{\bfX'}_F^2 - 2\angle*{\bfX,\bfX'} \\
        &= 2 - 2\angle*{\bfX,\bfX'} \leq \eps^2.
    \end{align*}
    To do this, we make use of the \emph{approximate matrix product} lemma, Lemma \ref{lem:apm-gaussian}. Then with the failure rate set to $\delta = 1/10\abs{\mathcal N}$ and the accuracy parameter $\eps$ in the lemma set to $\eps^2/\sqrt n$, we use
    \[
        m = \Theta\parens*{\frac{nsk}{\eps^4}\log\frac{nk}{\eps}}
    \]
    measurements to guarantee that for an $m\times n^2$ i.i.d.\ Gaussian matrix $\bfS$ and a net vector $\bfX'\in\mathcal N$,
    \[
        \Pr\braces*{\abs*{\angle*{\bfA,\bfX'} - \angle*{\bfS\vc(\bfA),\bfS\vc(\bfX')}} \leq \frac{\eps^2}{\sqrt n}\norm*{\bfA}_F\norm*{\bfX'}_F} \geq 1 - \frac1{10\abs*{\mathcal N}}. 
    \]
    Note that
    \[
        \frac{\eps^2}{\sqrt n}\norm*{\bfA}_F\norm*{\bfX'}_F = \frac{\eps^2}{\sqrt n}\cdot n \cdot 1 = \eps^2 \sqrt n. 
    \]
    Furthermore, 
    \[
        \angle*{\bfA,\bfX'} = \angle*{\bfG+\sqrt n\bfX, \bfX'} \sim \mathcal N(\angle*{\sqrt n\bfX,\bfX'}, \norm*{\bfX'}_F^2) = \mathcal N(\sqrt{n}\angle*{\bfX,\bfX'},1),
    \]
    and a standard Gaussian is at most $O(\sqrt{\log(1/\delta)})$ with probability at least $1-\delta$. Then by setting $\delta = 1/\abs{N}$, we have by a union bound over the net that, with constant probability,
    \begin{align*}
        \angle*{\bfS\vc(\bfA),\bfS\vc(\bfX')} &= \angle*{\bfA,\bfX'} \pm \eps^2\sqrt n \\
        &= \sqrt n\angle*{\bfX,\bfX'} \pm \parens*{O(sk)\log\frac{nk}{\eps} + \eps^2 \sqrt n} \\
        &= \sqrt n\angle*{\bfX,\bfX'} \pm 2\eps^2\sqrt n \\
        &= \sqrt n\parens*{\angle*{\bfX,\bfX'} \pm 2\eps^2}
    \end{align*}
    for all $\bfX'\in \mathcal N$. Thus by rescaling $\eps$, we may approximately maximize the inner product between $\bfX$ and $\bfX'$, and thus retrieve a net vector $\bfX'$ such that
    \[
        \norm*{\bfX - \bfX'}_2\leq \eps.\qedhere
    \]
\end{proof}

\subsubsection{Gaussian Noise Estimation Lower Bound}

\begin{theorem}\label{thm:gaussian-noise-est-lb}
    Suppose that $\bfS$ is a random $m\times (nd)$ matrix such that, given $\bfS(\vc(\bfA))$ for $\bfA = \sqrt{n}\bfX + \frac1{100}\bfG$ for $s\times s$-sparse  rank $k$ matrix $\bfX$ with $\norm*{\bfX}_2 = 1$ and $\bfG\sim\mathcal N(0,1)^{n\times n}$, there is an algorithm which outputs an $s\times s$-sparse rank $k$ matrix $\bfB$ such that
    \[
        \norm*{\bfA - \bfB}_2 \leq \frac1{10}\norm*{\bfG}_2.
    \]
    Then, $m = \Omega(nsk\log(n/sk))$. 
\end{theorem}
\begin{proof}

We follow the sparse recovery lower bound technique of \cite{DBLP:conf/focs/PriceW11} using the information capacity of a Gaussian channel. The techniques are essentially the same, so we only briefly outline their argument and the modifications that are necessary. 

Consider a family $\mathcal F$ of $sk$-sparse supports such that:
\begin{itemize}
    \item $\abs{S \triangle S'} \geq sk$ for $S\neq S'\in\mathcal F$
    \item $\Pr_{S\in\mathcal F}[i\in S] = sk/n$ for all $i\in[n]$
    \item $\log \abs{\mathcal F} = \Omega(sk \log(n/sk))$
\end{itemize}
which can be taken to be a random linear code in $[n/sk]^{sk}$ with relative distance $1/2$. 

We then define a communication game as follows. Alice first chooses a random $S\sim\mathcal F$. Then to construct a random rank $1$ matrix $\bfX$, she chooses a random row $i\in[n]$ and chooses a random sign vector supported on the support $S$ on the $i$th row, and normalizes by $\sqrt{n/sk}$ so that the resulting matrix $\bfX$ has operator norm $\sqrt n$. She then sets $\bfy = \bfS(\vc(\bfX + \frac1{100}\bfG))$ for $\bfG\sim\mathcal N(0,1)^{n\times n}$ and sends $\bfy$ to Bob. Bob performs the Gaussian noise estimation algorithm on $\bfy$ to retrieve some $s\times s$ sparse rank $k$ matrix $\bfB$, rounds to a matrix $\bfX'$ in the set of possible matrices $\bfX$, and returns the support $S' = \supp(\bfX')$. 

We WLOG assume that $\bfS$ has orthonormal rows. Note that for any fixed matrix $\bfV$ with Frobenius norm $1$,
\[
    \bfE_{\bfX}\angle*{\bfX, \bfV}^2 = \frac{n}{sk}\bfE_{S\sim\mathcal F}\bfE_{i\sim[n]}\bracks*{\sum_{j\in S}(\bfe_i^\top\bfV\bfe_j)^2} = \frac{n}{sk}\frac{sk}{n^2} = \frac1n
\]
while $\E\angle*{\bfG,\bfV}^2 = \E_{x\sim\mathcal N(0,1)} x^2 = 1$. Thus, for $i\in[m]$, if $\bfV = \bfe_i^\top\bfS$, then
\[
    \bfe_i^\top\bfy = \bfe_i^\top\bfS\parens*{\vc\parens*{\bfX + \frac1{100}\bfG}} = \angle*{\bfV,\bfX} + \frac1{100}\angle*{\bfV,\bfG}
\]
is a Gaussian channel with noise-to-signal ratio
\[
    \alpha = \frac{\E\angle*{\bfG,\bfV}^2}{\bfE_{\bfX}\angle*{\bfX, \bfV}^2} = n.
\]

The following lemma follows exactly as in \cite{DBLP:conf/focs/PriceW11}:
\begin{Lemma}[Lemma 4.1 of \cite{DBLP:conf/focs/PriceW11}]
\[
    I(S;S') = O\parens*{m\log\parens*{1 + \frac1\alpha}}
\]
\end{Lemma}

We now show that a successful Gaussian noise estimation output must necessarily retrieve $S'$ with constant probability. Because the recovery is successful with probability at least $9/10$, with this probabiliy, we have that $\norm*{\bfA - \bfB}_2 \leq \norm*{\bfG}_2 / 10$ so
\begin{align*}
    \norm*{\bfX - \bfB}_2 &\leq \norm*{\bfX + \frac1{100}\bfG - \bfB - \frac1{100}\bfG}_2 \\
    &\leq \norm*{\bfX + \frac1{100}\bfG - \bfB}_2 +  \frac1{100}\norm*{\bfG}_2 \\
    &= \norm*{\bfA - \bfB}_2 +  \frac1{100}\norm*{\bfG}_2 \\
    &\leq \frac{11}{100}\norm*{\bfG}_2.
\end{align*}
By \cite{DBLP:books/cu/12/VershyninEK12}, $\norm*{\bfG}_2 \leq 2.1\sqrt n$ so the above bound is at most $(24/100) \sqrt n$. Now note that if $\bfX \neq \bfX'$, then 
\[
    \norm*{\bfX - \bfX'} \geq \sqrt n
\]
so
\[
    \norm*{\bfX' - \bfB}_2 \geq \norm*{\bfX - \bfX'}_2 - \norm*{\bfX - \bfB}_2 \geq \frac{76}{100}\sqrt n
\]
and thus $\bfX'$ is not the closest matrix to $\bfB$ in the allowed matrices, which is a contradiction. 

By Fano's inequality, we then have that
\[
    I(S;S') = H(S) - H(S\mid S') \geq -1  + \frac9{10}\log\abs{\mathcal F} = \Omega(sk\log(n/sk))
\]
but we showed that $I(S;S') = O(m\log (1+1/n)) = O(m/n)$, so
\[
    m = \Omega(nsk\log(n/sk))
\]
as claimed.

\end{proof}

\section{Acknowledgements}
We thank the anonymous reviewers for their insightful comments. This research was supported in part by ONR grant N00014-18-1-2562, and a Simons Investigator Award.

\bibliographystyle{alpha}
\bibliography{citations}

\newcommand{\etalchar}[1]{$^{#1}$}
\begin{thebibliography}{MRWZ20}

\bibitem[ADL{\etalchar{+}}94]{DBLP:journals/jal/AlonDLRY94}
Noga Alon, Richard~A. Duke, Hanno Lefmann, Vojtech R{\"{o}}dl, and Raphael
  Yuster.
\newblock The algorithmic aspects of the regularity lemma.
\newblock {\em J. Algorithms}, 16(1):80--109, 1994.

\bibitem[AHLW16]{DBLP:conf/coco/AiHLW16}
Yuqing Ai, Wei Hu, Yi~Li, and David~P. Woodruff.
\newblock New characterizations in turnstile streams with applications.
\newblock In {\em Computational Complexity Conference}, volume~50 of {\em
  LIPIcs}, pages 20:1--20:22. Schloss Dagstuhl - Leibniz-Zentrum f{\"{u}}r
  Informatik, 2016.

\bibitem[AKO11]{DBLP:conf/focs/AndoniKO11}
Alexandr Andoni, Robert Krauthgamer, and Krzysztof Onak.
\newblock Streaming algorithms via precision sampling.
\newblock In {\em {FOCS}}, pages 363--372. {IEEE} Computer Society, 2011.

\bibitem[ANPW13]{DBLP:conf/icalp/AndoniNPW13}
Alexandr Andoni, Huy~L. Nguy{\^{e}}n, Yury Polyanskiy, and Yihong Wu.
\newblock Tight lower bound for linear sketches of moments.
\newblock In {\em {ICALP} {(1)}}, volume 7965 of {\em Lecture Notes in Computer
  Science}, pages 25--32. Springer, 2013.

\bibitem[BB19]{DBLP:conf/colt/BrennanB19}
Matthew Brennan and Guy Bresler.
\newblock Optimal average-case reductions to sparse {PCA:} from weak
  assumptions to strong hardness.
\newblock In {\em {COLT}}, volume~99 of {\em Proceedings of Machine Learning
  Research}, pages 469--470. {PMLR}, 2019.

\bibitem[BBB{\etalchar{+}}19]{DBLP:conf/soda/BanBBKLW19}
Frank Ban, Vijay Bhattiprolu, Karl Bringmann, Pavel Kolev, Euiwoong Lee, and
  David~P. Woodruff.
\newblock A {PTAS} for {\(\ell_p\)}-low rank approximation.
\newblock In {\em {SODA}}, pages 747--766. {SIAM}, 2019.

\bibitem[BO10]{DBLP:journals/corr/abs-1011-2571}
Vladimir Braverman and Rafail Ostrovsky.
\newblock Recursive sketching for frequency moments.
\newblock {\em CoRR}, abs/1011.2571, 2010.

\bibitem[BSBP16]{DBLP:journals/tsp/BenidisSBP16}
Konstantinos Benidis, Ying Sun, Prabhu Babu, and Daniel~P. Palomar.
\newblock Orthogonal sparse {PCA} and covariance estimation via procrustes
  reformulation.
\newblock {\em {IEEE} Trans. Signal Process.}, 64(23):6211--6226, 2016.

\bibitem[BSS20]{bandeira2020mathematics}
A~Bandeira, A~Singer, and T~Strohmer.
\newblock Mathematics of data science.
\newblock {\em Book draft available at https://people. math. ethz.
  ch/abandeira/BandeiraSingerStrohmer-MDS-draft. pdf}, 2020.

\bibitem[BWZ16]{DBLP:conf/stoc/BoutsidisWZ16}
Christos Boutsidis, David~P. Woodruff, and Peilin Zhong.
\newblock Optimal principal component analysis in distributed and streaming
  models.
\newblock In {\em {STOC}}, pages 236--249. {ACM}, 2016.

\bibitem[CCF04]{DBLP:journals/tcs/CharikarCF04}
Moses Charikar, Kevin~C. Chen, and Martin Farach{-}Colton.
\newblock Finding frequent items in data streams.
\newblock {\em Theor. Comput. Sci.}, 312(1):3--15, 2004.

\bibitem[CPR16]{DBLP:conf/colt/ChanPR16}
Siu~On Chan, Dimitris Papailliopoulos, and Aviad Rubinstein.
\newblock On the approximability of sparse {PCA}.
\newblock In {\em {COLT}}, volume~49 of {\em {JMLR} Workshop and Conference
  Proceedings}, pages 623--646. JMLR.org, 2016.

\bibitem[CST13]{DBLP:journals/tsp/CandesST13}
Emmanuel~J. Cand{\`{e}}s, Carlos Sing{-}Long, and Joshua~D. Trzasko.
\newblock Unbiased risk estimates for singular value thresholding and spectral
  estimators.
\newblock {\em {IEEE} Trans. Signal Process.}, 61(19):4643--4657, 2013.

\bibitem[CW09]{DBLP:conf/stoc/ClarksonW09}
Kenneth~L. Clarkson and David~P. Woodruff.
\newblock Numerical linear algebra in the streaming model.
\newblock In {\em {STOC}}, pages 205--214. {ACM}, 2009.

\bibitem[CW13]{DBLP:conf/stoc/ClarksonW13}
Kenneth~L. Clarkson and David~P. Woodruff.
\newblock Low rank approximation and regression in input sparsity time.
\newblock In {\em {STOC}}, pages 81--90. {ACM}, 2013.

\bibitem[DG03]{DBLP:journals/rsa/DasguptaG03}
Sanjoy Dasgupta and Anupam Gupta.
\newblock An elementary proof of a theorem of johnson and lindenstrauss.
\newblock {\em Random Struct. Algorithms}, 22(1):60--65, 2003.

\bibitem[DKM06]{DBLP:journals/siamcomp/DrineasKM06a}
Petros Drineas, Ravi Kannan, and Michael~W. Mahoney.
\newblock Fast monte carlo algorithms for matrices {II:} computing a low-rank
  approximation to a matrix.
\newblock {\em {SIAM} J. Comput.}, 36(1):158--183, 2006.

\bibitem[FKV04]{DBLP:journals/jacm/FriezeKV04}
Alan~M. Frieze, Ravi Kannan, and Santosh~S. Vempala.
\newblock Fast monte-carlo algorithms for finding low-rank approximations.
\newblock {\em J. {ACM}}, 51(6):1025--1041, 2004.

\bibitem[Gan15]{DBLP:conf/icalp/Ganguly15}
Sumit Ganguly.
\newblock Taylor polynomial estimator for estimating frequency moments.
\newblock In {\em {ICALP} {(1)}}, volume 9134 of {\em Lecture Notes in Computer
  Science}, pages 542--553. Springer, 2015.

\bibitem[GG11]{DBLP:journals/siammax/GillisG11}
Nicolas Gillis and Fran{\c{c}}ois Glineur.
\newblock Low-rank matrix approximation with weights or missing data is
  np-hard.
\newblock {\em {SIAM} J. Matrix Anal. Appl.}, 32(4):1149--1165, 2011.

\bibitem[GJ79]{DBLP:books/fm/GareyJ79}
M.~R. Garey and David~S. Johnson.
\newblock {\em Computers and Intractability: {A} Guide to the Theory of
  NP-Completeness}.
\newblock W. H. Freeman, 1979.

\bibitem[Gu15]{DBLP:journals/siamsc/Gu15}
Ming Gu.
\newblock Subspace iteration randomization and singular value problems.
\newblock {\em {SIAM} J. Sci. Comput.}, 37(3), 2015.

\bibitem[GW18]{DBLP:conf/icalp/GangulyW18}
Sumit Ganguly and David~P. Woodruff.
\newblock High probability frequency moment sketches.
\newblock In {\em {ICALP}}, volume 107 of {\em LIPIcs}, pages 58:1--58:15.
  Schloss Dagstuhl - Leibniz-Zentrum f{\"{u}}r Informatik, 2018.

\bibitem[HBCY21]{DBLP:journals/nla/HernandezBCY21}
Taylor~M. Hernandez, Roel~Van Beeumen, Mark~A. Caprio, and Chao Yang.
\newblock A greedy algorithm for computing eigenvalues of a symmetric matrix
  with localized eigenvectors.
\newblock {\em Numer. Linear Algebra Appl.}, 28(2), 2021.

\bibitem[KP20]{DBLP:conf/stoc/KallaugherP20}
John Kallaugher and Eric Price.
\newblock Separations and equivalences between turnstile streaming and linear
  sketching.
\newblock In {\em {STOC}}, pages 1223--1236. {ACM}, 2020.

\bibitem[Lin18]{DBLP:journals/jacm/Lin18}
Bingkai Lin.
\newblock The parameterized complexity of the \emph{k}-biclique problem.
\newblock {\em J. {ACM}}, 65(5):34:1--34:23, 2018.

\bibitem[LNW14]{DBLP:conf/stoc/LiNW14}
Yi~Li, Huy~L. Nguyen, and David~P. Woodruff.
\newblock Turnstile streaming algorithms might as well be linear sketches.
\newblock In {\em {STOC}}, pages 174--183. {ACM}, 2014.

\bibitem[LNW19]{DBLP:journals/siamcomp/LiNW19}
Yi~Li, Huy~L. Nguyen, and David~P. Woodruff.
\newblock On approximating matrix norms in data streams.
\newblock {\em {SIAM} J. Comput.}, 48(6):1643--1697, 2019.

\bibitem[LVTW09]{lagendijk2009fifty}
Aart Lagendijk, Bart Van~Tiggelen, and Diederik~S Wiersma.
\newblock Fifty years of anderson localization.
\newblock {\em Phys. Today}, 62(8):24--29, 2009.

\bibitem[Mac08]{DBLP:conf/nips/Mackey08}
Lester~W. Mackey.
\newblock Deflation methods for sparse {PCA}.
\newblock In {\em {NIPS}}, pages 1017--1024. Curran Associates, Inc., 2008.

\bibitem[Mag17]{DBLP:journals/ipl/Magdon-Ismail17}
Malik Magdon{-}Ismail.
\newblock Np-hardness and inapproximability of sparse {PCA}.
\newblock {\em Inf. Process. Lett.}, 126:35--38, 2017.

\bibitem[MM15]{DBLP:conf/nips/MuscoM15}
Cameron Musco and Christopher Musco.
\newblock Randomized block krylov methods for stronger and faster approximate
  singular value decomposition.
\newblock In {\em {NIPS}}, pages 1396--1404, 2015.

\bibitem[MRS21]{DBLP:conf/innovations/ManurangsiRS21}
Pasin Manurangsi, Aviad Rubinstein, and Tselil Schramm.
\newblock The strongish planted clique hypothesis and its consequences.
\newblock In {\em {ITCS}}, volume 185 of {\em LIPIcs}, pages 10:1--10:21.
  Schloss Dagstuhl - Leibniz-Zentrum f{\"{u}}r Informatik, 2021.

\bibitem[MRWZ20]{DBLP:conf/stoc/MahabadiRWZ20}
Sepideh Mahabadi, Ilya~P. Razenshteyn, David~P. Woodruff, and Samson Zhou.
\newblock Non-adaptive adaptive sampling on turnstile streams.
\newblock In {\em {STOC}}, pages 1251--1264. {ACM}, 2020.

\bibitem[MWA06]{DBLP:conf/icml/MoghaddamWA06}
Baback Moghaddam, Yair Weiss, and Shai Avidan.
\newblock Generalized spectral bounds for sparse {LDA}.
\newblock In {\em {ICML}}, volume 148 of {\em {ACM} International Conference
  Proceeding Series}, pages 641--648. {ACM}, 2006.

\bibitem[NH15]{nandkishore2015many}
Rahul Nandkishore and David~A Huse.
\newblock Many-body localization and thermalization in quantum statistical
  mechanics.
\newblock {\em Annu. Rev. Condens. Matter Phys.}, 6(1):15--38, 2015.

\bibitem[NmJ15]{ning2015literature}
Shen Ning-min and Li~Jing.
\newblock A literature survey on high-dimensional sparse principal component
  analysis.
\newblock {\em International Journal of Database Theory and Application},
  8(6):57--74, 2015.

\bibitem[PSC18]{pastor2018eigenvector}
Romualdo Pastor-Satorras and Claudio Castellano.
\newblock Eigenvector localization in real networks and its implications for
  epidemic spreading.
\newblock {\em Journal of Statistical Physics}, 173(3):1110--1123, 2018.

\bibitem[PVZ17]{paouris2017random}
Grigoris Paouris, Petros Valettas, and Joel Zinn.
\newblock Random version of dvoretzky’s theorem in $\ell_p^n$.
\newblock {\em Stochastic Processes and their Applications},
  127(10):3187--3227, 2017.

\bibitem[PW11]{DBLP:conf/focs/PriceW11}
Eric Price and David~P. Woodruff.
\newblock {(1} + eps)-approximate sparse recovery.
\newblock In {\em {FOCS}}, pages 295--304. {IEEE} Computer Society, 2011.

\bibitem[RSW16]{DBLP:conf/stoc/RazenshteynSW16}
Ilya~P. Razenshteyn, Zhao Song, and David~P. Woodruff.
\newblock Weighted low rank approximations with provable guarantees.
\newblock In {\em {STOC}}, pages 250--263. {ACM}, 2016.

\bibitem[Sar06]{DBLP:conf/focs/Sarlos06}
Tam{\'{a}}s Sarl{\'{o}}s.
\newblock Improved approximation algorithms for large matrices via random
  projections.
\newblock In {\em {FOCS}}, pages 143--152. {IEEE} Computer Society, 2006.

\bibitem[SN13]{DBLP:journals/ma/ShabalinN13}
Andrey~A. Shabalin and Andrew~B. Nobel.
\newblock Reconstruction of a low-rank matrix in the presence of gaussian
  noise.
\newblock {\em J. Multivar. Anal.}, 118:67--76, 2013.

\bibitem[TZ04]{DBLP:conf/soda/ThorupZ04}
Mikkel Thorup and Yin Zhang.
\newblock Tabulation based 4-universal hashing with applications to second
  moment estimation.
\newblock In {\em {SODA}}, pages 615--624. {SIAM}, 2004.

\bibitem[Ver12]{DBLP:books/cu/12/VershyninEK12}
Roman Vershynin.
\newblock Introduction to the non-asymptotic analysis of random matrices.
\newblock In {\em Compressed Sensing}, pages 210--268. Cambridge University
  Press, 2012.

\bibitem[Ver15]{vershynin2015estimation}
Roman Vershynin.
\newblock Estimation in high dimensions: a geometric perspective.
\newblock In {\em Sampling theory, a renaissance}, pages 3--66. Springer, 2015.

\bibitem[Ver18]{vershynin2018high}
Roman Vershynin.
\newblock {\em High-dimensional probability: An introduction with applications
  in data science}, volume~47.
\newblock Cambridge university press, 2018.

\bibitem[Woo14]{DBLP:journals/fttcs/Woodruff14}
David~P. Woodruff.
\newblock Sketching as a tool for numerical linear algebra.
\newblock {\em Found. Trends Theor. Comput. Sci.}, 10(1-2):1--157, 2014.

\bibitem[ZL16]{DBLP:conf/nips/ZhuL16}
Zeyuan~Allen Zhu and Yuanzhi Li.
\newblock Even faster {SVD} decomposition yet without agonizing pain.
\newblock In {\em {NIPS}}, pages 974--982, 2016.

\bibitem[ZX18]{DBLP:journals/pieee/ZouX18}
Hui Zou and Lingzhou Xue.
\newblock A selective overview of sparse principal component analysis.
\newblock {\em Proc. {IEEE}}, 106(8):1311--1320, 2018.

\bibitem[ZYC{\etalchar{+}}20]{DBLP:journals/corr/abs-2009-04414}
Li~Zhou, Lihao Yan, Mark~A. Caprio, Weiguo Gao, and Chao Yang.
\newblock Solving the k-sparse eigenvalue problem with reinforcement learning.
\newblock {\em CoRR}, abs/2009.04414, 2020.

\bibitem[ZZS02]{DBLP:journals/siammax/ZhangZS02}
Zhenyue Zhang, Hongyuan Zha, and Horst~D. Simon.
\newblock Low-rank approximations with sparse factors {I:} basic algorithms and
  error analysis.
\newblock {\em {SIAM} J. Matrix Anal. Appl.}, 23(3):706--727, 2002.

\bibitem[ZZS04]{DBLP:journals/siammax/ZhangZS04}
Zhenyue Zhang, Hongyuan Zha, and Horst~D. Simon.
\newblock Low-rank approximations with sparse factors {II:} penalized methods
  with discrete newton-like iterations.
\newblock {\em {SIAM} J. Matrix Anal. Appl.}, 25(4):901--920, 2004.

\end{thebibliography}

\appendix

\end{document}